\documentclass[12pt]{amsart}

\usepackage[utf8]{inputenc}
\usepackage{cite}

\usepackage[usenames,dvipsnames]{xcolor}
\usepackage{pgfplots}
\usepackage{tikz}
\usepackage{tkz-euclide}
\usepackage{geometry}
\usepackage{thmtools}
\usepackage{thm-restate}
\usepackage{amsmath,amssymb}
\usepackage{thm-restate}
\usepackage{enumerate}
\usepackage{framed}
\usepackage{geometry}
\usepackage[normalem]{ulem}
\usepackage[foot]{amsaddr}
\usepackage{comment}

\usepackage[unicode=true,bookmarks=false,
 breaklinks=false,pdfborder={0 0 1},backref=false,colorlinks=false,bookmarksdepth=3]
 {hyperref}

\hypersetup{colorlinks=true, linkcolor=Maroon, citecolor=OliveGreen, filecolor=magenta, urlcolor=Blue}

\newcommand{\eq}[1]{\hyperref[eq:#1]{(\ref*{eq:#1})}}
\renewcommand{\sec}[1]{\hyperref[sec:#1]{Section~\ref*{sec:#1}}}
\newcommand{\app}[1]{\hyperref[app:#1]{Appendix~\ref*{app:#1}}}
\newcommand{\theo}[1]{\hyperref[thm:#1]{Theorem~\ref*{thm:#1}}}
\newcommand{\algo}[1]{\hyperref[alg:#1]{Algorithm~\ref*{alg:#1}}}
\newcommand{\lemm}[1]{\hyperref[lem:#1]{Lemma~\ref*{lem:#1}}}
\newcommand{\defin}[1]{\hyperref[defn:#1]{Definition~\ref*{defn:#1}}}
\newcommand{\corr}[1]{\hyperref[cor:#1]{Corollary~\ref*{cor:#1}}}
\newcommand{\fig}[1]{\hyperref[fig:#1]{Figure~\ref*{fig:#1}}}
\newcommand{\propos}[1]{\hyperref[prop:#1]{Proposition~\ref*{prop:#1}}}
\newcommand{\rema}[1]{\hyperref[rem:#1]{Remark~\ref*{rem:#1}}}

\newtheorem{thm}{Theorem}[section]
\newtheorem{lem}[thm]{Lemma}
\newtheorem{prop}[thm]{Proposition}
\newtheorem{cor}[thm]{Corollary}
\newtheorem{prob}[thm]{Problem}
\newtheorem{conj}[thm]{Conjecture}
\theoremstyle{definition}
\newtheorem{dfn}[thm]{Definition}
\theoremstyle{remark}

\usepackage{algpseudocode}
\algnewcommand{\A}{\textbf{and}\space}
\algnewcommand{\Or}{\textbf{or}\space}
\algnewcommand{\Xor}{\textbf{xor}\space}
\algnewcommand{\Not}{\textbf{not}\space}

\algnewcommand{\Fixed}{\textbf{Fixed input:}\space}
\algnewcommand{\Input}{\textbf{Input:}\space}
\algnewcommand{\Offline}{\textbf{Offline:}\space}
\algnewcommand{\Online}{\textbf{Online:}\space}
\algnewcommand{\Output}{\textbf{Output:}\space}
\algnewcommand{\Spec}{\textbf{Specification:}\space}

\makeatletter
\let\OldStatex\Statex
\renewcommand{\Statex}[1][3]{%
  \setlength\@tempdima{\algorithmicindent}%
  \OldStatex\hskip\dimexpr#1\@tempdima\relax}
\makeatother

\newcommand{\fxinput}{\Statex[-1] \Fixed }
\newcommand{\inp}{\Statex[-1] \Input }
\newcommand{\out}{\Statex[-1] \Output }
\newcommand{\offline}{\Statex[-1] \Offline }
\newcommand{\online}{\Statex[-1] \Online }
\newcommand{\spec}{\Statex[-1] \Spec }
\algnewcommand{\assert}{\textbf{assert}\space}

\makeatother

\author{Vadym Kliuchnikov$^{1}$\email{vadym@microsoft.com}}
\author{Alex Bocharov$^{1}$\email{alexei.bocharov@microsoft.com}}
\author{Martin Roetteler$^{1}$\email{martinro@microsoft.com}}
\author{Jon Yard$^{1}$\email{jonyard@microsoft.com}}
\address{$^1$ Quantum Architectures and Computation Group, Microsoft Research, Redmond, WA, USA}

\newcommand{\teol}[0]{\tabularnewline \hline}

\def\z{\mathbb{Z}}

\def\c{\mathbb{C}}
\def\q{\mathbb{Q}}
\def\f{\mathbb{F}}
\def\r{\mathbb{R}}

\global\long\def\rg#1#2{#1,\ldots,#2}
\global\long\def\trg#1#2{#1\cdot\ldots\cdot#2}
\global\long\def\prg#1#2{#1+\ldots+#2}

\def\Re{\mathrm{Re}}
\def\Im{\mathrm{Im}}

\global\long\def\zf{\mathbb{Z}_{F}} %ring of intebgers of number field F
\global\long\def\p{\mathfrak{p}} %integral prime ideal 

\global\long\def\i{\boldsymbol{i}}
\global\long\def\j{\boldsymbol{j}}
\global\long\def\k{\boldsymbol{k}}

\global\long\def\vp{\varphi}
\global\long\def\ve{\varepsilon}
\global\long\def\s{\sigma}
\global\long\def\vecs{\boldsymbol{\sigma}}
\global\long\def\mo{\mathcal{M}}

\global\long\def\UG#1#2#3#4{\left(\begin{array}{cc}
 #1  &  #2\\
#3  &  #4 
\end{array}\right)}

\def\rnd#1{\mathopen{}\left\lceil #1\right\rfloor\mathclose{}}
\def\ip#1{\mathopen{}\left\langle #1\right\rangle\mathclose{}}
\def\nrm#1{\mathopen{}\left\Vert #1\right\Vert\mathclose{}}
\def\floor#1{\mathopen{}\left\lfloor #1\right\rfloor\mathclose{}}
\def\ceil#1{\mathopen{}\left\lceil #1\right\rceil\mathclose{}}
\def\at#1{\mathopen{}\left(#1\right)\mathclose{}}
\def\appr#1{\mathopen{}\left(#1\right)'\mathclose{}}
\def\of#1{\mathopen{}\left[#1\right]\mathclose{}}
\def\set#1{\mathopen{}\left\{  #1\right\}\mathclose{}}

\def\ket#1{\mathopen{}\left|#1\right\rangle\mathclose{}}
\def\bra#1{\mathopen{}\left\langle #1\right|\mathclose{}}
\def\abss#1{\mathopen{}\left|#1\right|^{2}\mathclose{}}
\def\abs#1{\mathopen{}\left|#1\right|\mathclose{}}

\def\UG#1#2#3#4{{\at{\begin{array}{cc} #1 & #2 \\ #3 & #4 \end{array} }}}

\DeclareMathOperator{\nrd}{nrd}
\DeclareMathOperator{\atan}{atan}
\DeclareMathOperator{\tr}{Tr}
\DeclareMathOperator{\trd}{trd}

\DeclareMathOperator{\re}{Re}
\DeclareMathOperator{\im}{Im}

\DeclareMathOperator{\crt}{Circuit}

\DeclareMathOperator{\Cand}{Cand}
\DeclareMathOperator{\sln}{Sln}
\DeclareMathOperator{\term}{Term}
\DeclareMathOperator{\pterm}{PolyTerm}

\newcommand{\nc}{\newcommand}

\DeclareMathOperator{\disc}{disc}

\def\GL{\mathrm{GL}}

\makeatletter
\let\O\@undefined
\def\O{\mathrm{O}}

\def\SU{\mathrm{SU}}

\def\pmat#1{{\begin{pmatrix}#1\end{pmatrix}}}

\nc{\CA}{\mathcal{A}} \nc{\CB}{\mathcal{B}} \nc{\CC}{\mathcal{C}}
\nc{\CD}{\mathcal{D}} \nc{\CE}{\mathcal{E}} \nc{\CF}{\mathcal{F}}
\nc{\CG}{\mathcal{G}} \nc{\CH}{\mathcal{H}} \nc{\CI}{\mathcal{I}}
\nc{\CJ}{\mathcal{J}} \nc{\CK}{\mathcal{K}} \nc{\CL}{\mathcal{L}}
\nc{\CM}{\mathcal{M}} \nc{\CN}{\mathcal{N}} \nc{\CO}{\mathcal{O}}
\nc{\CP}{\mathcal{P}} \nc{\CQ}{\mathcal{Q}} \nc{\CR}{\mathcal{R}} 
\nc{\CS}{\mathcal{S}} \nc{\CT}{\mathcal{T}} \nc{\CU}{\mathcal{U}} 
\nc{\CV}{\mathcal{V}} \nc{\CW}{\mathcal{W}} \nc{\CX}{\mathcal{X}} 
\nc{\CY}{\mathcal{Y}} \nc{\CZ}{\mathcal{Z}}

\nc{\bA}{\mathbb{A}} \nc{\bB}{\mathbb{B}} \nc{\bC}{\mathbb{C}}
\nc{\bD}{\mathbb{D}} \nc{\bE}{\mathbb{E}} \nc{\bF}{\mathbb{F}}
\nc{\bG}{\mathbb{G}} \nc{\bH}{\mathbb{H}} \nc{\bI}{\mathbb{I}}
\nc{\bJ}{\mathbb{J}} \nc{\bK}{\mathbb{K}} \nc{\bL}{\mathbb{L}}
\nc{\bM}{\mathbb{M}} \nc{\bN}{\mathbb{N}} \nc{\bO}{\mathbb{O}}
\nc{\bP}{\mathbb{P}} \nc{\bQ}{\mathbb{Q}} \nc{\bR}{\mathbb{R}} 
\nc{\bS}{\mathbb{S}} \nc{\bT}{\mathbb{T}} \nc{\bU}{\mathbb{U}} 
\nc{\bV}{\mathbb{V}} \nc{\bW}{\mathbb{W}} \nc{\bX}{\mathbb{X}}
\nc{\bZ}{\mathbb{Z}}

\nc{\BA}{\mathbf{A}} \nc{\BB}{\mathbf{B}} \nc{\BC}{\mathbf{C}}
\nc{\BD}{\mathbf{D}} \nc{\BE}{\mathbf{E}} \nc{\BF}{\mathbf{F}}
\nc{\BG}{\mathbf{G}} \nc{\BH}{\mathbf{H}} \nc{\BI}{\mathbf{I}}
\nc{\BJ}{\mathbf{J}} \nc{\BK}{\mathbf{K}} \nc{\BL}{\mathbf{L}}
\nc{\BM}{\mathbf{M}} \nc{\BN}{\mathbf{N}} \nc{\BO}{\mathbf{O}}
\nc{\BP}{\mathbf{P}} \nc{\BQ}{\mathbf{Q}} \nc{\BR}{\mathbf{R}} 
\nc{\BS}{\mathbf{S}} \nc{\BT}{\mathbf{T}} \nc{\BU}{\mathbf{U}} 
\nc{\BV}{\mathbf{V}} \nc{\BW}{\mathbf{W}} \nc{\BX}{\mathbf{X}} 
\nc{\BY}{\mathbf{Y}} \nc{\BZ}{\mathbf{Z}}

\nc{\msA}{\mathscr{A}} \nc{\msB}{\mathscr{B}} \nc{\msC}{\mathscr{C}}
\nc{\msD}{\mathscr{D}} \nc{\msE}{\mathscr{E}} \nc{\msF}{\mathscr{F}}
\nc{\msG}{\mathscr{G}} \nc{\msH}{\mathscr{H}} \nc{\msI}{\mathscr{I}}
\nc{\msJ}{\mathscr{J}} \nc{\msK}{\mathscr{K}} \nc{\msL}{\mathscr{L}}
\nc{\msM}{\mathscr{M}} \nc{\msN}{\mathscr{N}} \nc{\msO}{\mathscr{O}}
\nc{\msP}{\mathscr{P}} \nc{\msQ}{\mathscr{Q}} \nc{\msR}{\mathscr{R}} 
\nc{\msS}{\mathscr{S}} \nc{\msT}{\mathscr{T}} \nc{\msU}{\mathscr{U}} 
\nc{\msV}{\mathscr{V}} \nc{\msX}{\mathscr{X}} \nc{\msW}{\mathscr{W}} 
\nc{\msY}{\mathscr{Y}} \nc{\msZ}{\mathscr{Z}}

\nc{\mfa}{{\mathfrak a}} \nc{\mfb}{{\mathfrak b}} \nc{\mfc}{{\mathfrak c}}
\nc{\mfd}{{\mathfrak d}} \nc{\mfe}{{\mathfrak e}} \nc{\mff}{{\mathfrak f}}
\nc{\mfg}{{\mathfrak g}} \nc{\mfh}{{\mathfrak h}} \nc{\mfi}{{\mathfrak i}}
\nc{\mfj}{{\mathfrak j}} \nc{\mfk}{{\mathfrak k}} \nc{\mfl}{{\mathfrak l}}
\nc{\mfm}{{\mathfrak m}} \nc{\mfn}{{\mathfrak n}} \nc{\mfo}{{\mathfrak o}}
\nc{\mfp}{{\mathfrak p}} \nc{\mfq}{{\mathfrak q}} \nc{\mfr}{{\mathfrak r}}
\nc{\mfs}{{\mathfrak s}} \nc{\mft}{{\mathfrak t}} \nc{\mfu}{{\mathfrak u}}
\nc{\mfv}{{\mathfrak v}} \nc{\mfw}{{\mathfrak w}} \nc{\mfx}{{\mathfrak x}}
\nc{\mfy}{{\mathfrak y}} \nc{\mfz}{{\mathfrak z}}

\nc{\mfA}{{\mathfrak A}} \nc{\mfB}{{\mathfrak B}} \nc{\mfC}{{\mathfrak C}}
\nc{\mfD}{{\mathfrak D}} \nc{\mfE}{{\mathfrak E}} \nc{\mfF}{{\mathfrak F}}
\nc{\mfG}{{\mathfrak G}} \nc{\mfH}{{\mathfrak H}} \nc{\mfI}{{\mathfrak I}}
\nc{\mfJ}{{\mathfrak J}} \nc{\mfK}{{\mathfrak K}} \nc{\mfL}{{\mathfrak L}}
\nc{\mfM}{{\mathfrak M}} \nc{\mfN}{{\mathfrak N}} \nc{\mfO}{{\mathfrak O}}
\nc{\mfP}{{\mathfrak P}} \nc{\mfQ}{{\mathfrak Q}} \nc{\mfR}{{\mathfrak R}}
\nc{\mfS}{{\mathfrak S}} \nc{\mfT}{{\mathfrak T}} \nc{\mfU}{{\mathfrak U}}
\nc{\mfV}{{\mathfrak V}} \nc{\mfW}{{\mathfrak W}} \nc{\mfX}{{\mathfrak X}}
\nc{\mfY}{{\mathfrak Y}} \nc{\mfZ}{{\mathfrak Z}}

\def\emph#1{{\bf #1}}
\nc{\proj}[1]{\ket{#1}\bra{#1}}
\nc{\braket}[2]{{\langle #1 | #2 \rangle}}

\def\sig{\sigma}

\def\norm#1{ {|\hspace{-.022in}|#1|\hspace{-.022in}|} }

\nc{\smfrac}[2]{\mbox{$\frac{#1}{#2}$}}

\def\bsone{{\boldsymbol{1}}}

\def\mmin{\mathrm{min}}

\def\poly{\mathrm{poly}}

\nc\mypar[1]{{\textbf{#1.}}}
\nc\nix[1]{{}}
\def\span{\mathrm{span}} % commonly used macros

\begin{document}

\title{A Framework for Approximating Qubit Unitaries}

\begin{abstract}
We present an algorithm for efficiently approximating of qubit unitaries over gate sets derived from totally definite quaternion algebras. It achieves $\ve$-approximations using circuits of length $O\at{\log\at{1/\ve}}$, which is asymptotically optimal. The algorithm achieves the same quality of approximation as previously-known algorithms for Clifford+T [arXiv:1212.6253],  V-basis [arXiv:1303.1411] and Clifford+$\pi/12$ [arXiv:1409.3552], running on average in time polynomial in $O\at{\log\at{1/\ve}}$ (conditional on a number-theoretic conjecture). Ours is the first such algorithm that works for a wide range of gate sets and provides insight into what should constitute a ``good'' gate set for a fault-tolerant quantum computer.
\end{abstract}

\maketitle

%%Paper body

%\input{sec-intro.tex}

%!TEX root = quaf.tex

\section{Introduction}

Each time we build a new computing device, we ask the question: what problems can it solve? We wonder the same thing about the quantum computers we will build. When addressing such questions, we usually start with a crude analysis, asking how resources like time, memory, cost and the size of the computer scale with the problem size. In particular, how do these is how these resources depend on the particular gate set supported by a quantum computer? The  algorithm of Solovay and Kitaev~\cite{KSV:2002,DN:2005} shows that any two universal gate are equally good from the perspective of polynomially-scaling resources. However, once we can build a small quantum computer we will be asking more refined questions: How large of a problem can we solve on it?  How can we compile our algorithms in the most resource-efficient way possible?

Typically, a circuit implementing a quantum algorithm uses a large number of gates, or local unitaries.  Each local unitary must be compiled into the gate set supported by a target fault-tolerant quantum computer.  Whereas the unitary groups are uncountable, most promising 
quantum computer architectures known today (topological or based on error correcting codes) natively support only a {\it finite} set of unitary gates.
The problem of optimal compilation into circuits over such a gate set can be naturally formulated as that of approximation in such groups. 

%Let us state the problem we are studying in more precise way.

In this paper, we focus on the problem of compiling circuits for single-qubit unitaries, i.e.\ that of approximation by finitely-generated subgroups of $\SU(2)$.  Let us start with a systematic description of the latter problem.
%The 
Let $\CG \subset \SU(2)$ be a finite set of $2\times 2$ unitary matrices, or \emph{gates}.
Given 
an arbitrary unitary
$U \in \SU(2)$, we want to express it in terms of unitaries from $\CG$. In most cases $U$ can not be expressed exactly using elements of $\CG$ and must be therefore approximated.
%For this reason we approximating $U$. Given precision 
For a selected absolute precision $\ve$,  our task is to find a sequence of gates    $g_1,\dotsc,g_N \in \CG$ (usually called a circuit over $\CG$) such that $\nrm{U-g_N \cdots g_1 } \le \ve$. If we can approximate any unitary over the gate set $\CG$ (or in other words, if $\CG$ generates a dense subgroup $\ip{\CG}$ of $\SU(2)$), we call $\CG$ a \emph{universal}.
Given that each unitary can be so approximated, we may then ask for the shortest, or least costly, such circuit. 
A volume argument shows that there exist unitaries $U$ requiring circuits of length at least $C\log\at{1/\ve}$, where $C$ is a constant that depends on the gate set $\CG$. A natural question to ask is whether there is a matching upper bound, i.e.\ whether we can approximate any unitary using a circuit of length $O\at{\log\at{1/\ve}}$. To answer this question, one must employ non-trivial mathematical ideas~\cite{B1,H1,S1,S2}. For example, it was recently shown~\cite{B1} that such approximations exist if the unitaries in $\CG$ have entries that are algebraic numbers. All known gate sets associated to fault-tolerant quantum computing architectures have this property.

Unfortunately, the result \cite{B1} is non constructive. 
%Let us know look on the constructive side of the question. 
% TODO: Vadym: is there way to make the next sentence better ? 
Furthermore, there is no obvious way to make it constructive that would realistically work for even moderately small precision target $\varepsilon$.
The result of \cite{B1} implies that brute-force search can yield approximations saturating the  lower bound. In practice, however, the precision of approximation achievable with brute force search is limited to $10^{-4}$ or $10^{-5}$. Ideally, we would like to have an algorithm that finds an $\ve$-approximation of a given unitary with a circuit of length $O\at{\log\at{1/\ve}}$ and, furthermore, we would like the algorithm to run in $O(\poly(\log\at{1/\ve})$ time.

Recently such algorithms were found for several gate sets such as Clifford+T~\cite{S,RS}, the V-basis~\cite{BGS}, Clifford+$R_z\at{\pi/6}$~\cite{BRS:2015b} and the braiding of Fibonacci anyons~\cite{KBS}. The question why it is possible to construct such an algorithm for these gate sets and what general properties such gate set should has been an outstanding challenge in the field. In this paper we for the first time present a general mathematical framework that enables the development of efficient approximation algorithms for entire families of gate sets, instead of for specific examples. We develop such an algorithm in the general setting for gate sets derived from totally definite quaternion algebras. We also implemented the algorithm and show results of applying it to a wide range of gate sets, including Clifford+$\sqrt{T}$. We also reproduce results from~\cite{S,BGS,BRS:2015b}. As noted in~\cite{KY1}, the Fibonacci gate set is related to indefinite quaternion algebras; we will analyze this case elsewhere as it requires more work. Our algorithm does not allow us to find absolutely optimal circuits given factoring oracle like this is done in~\cite{RS} for Clifford+T gate set. Achieving this is an interesting topic for future research.

The proof that our algorithm terminates and runs on average in polynomial time relies on a number-theoretic conjecture that generalizes and refines similar conjectures appearing in~\cite{S,BGS,BRS:2015b}. %The question related to this 
The mathematics behind conjectures of this type were recently studied in~\cite{Sarnak:2015} for Clifford+T, V-basis and some other gate sets. Results of our numerical experiments provide indirect evidence that some of results in~\cite{Sarnak:2015} can be true for a wider range of gate sets. This is related to the ``Golden Gates" introduced in~\cite{Sarnak:2015}. We discuss this in more detail in Section~\ref{sec:conjecture}. Next we state the problem of unitary approximation more formally and provide a high level overview of our approximation framework.

% Next we discuss previous results related to our work in more details.
% One factor in scaling up quantum computing is the ability to compile complex quantum algorithms into sequences of instructions that a physical quantum computer can then execute fault-tolerantly. This requirement for compilation creates a need for methods that allow to find short sequences of instructions while not requiring
% %and at the same time of and that at the same time do not lead to
% prohibitively long compilation times.

\subsection{Ancillae free approximation} Formally the problem of ancillae free approximation for single qubit gate sets can be stated as follows:

% TODO:  State Closest unitary problem , discuss Ross-Selinger, brute force and \cite{KMM3}. Discuss its relaxation.

% TODO: Now we explain how to obtain the cost bound function for brute force search algorithms. $\td$

\begin{prob}[unitary approximation problem in two dimensions, UAP]\label{prob:uap}
Given
\begin{enumerate}  
\item finite universal unitary gate set $\CG \subset \mathrm{SU}(2)$
\item cost function $c:\CG\rightarrow \r^{+}$ ($c:\CG\rightarrow\set{1}$ corresponds to circuit length)
\item distance function $\rho$ on the set of unitaries
\item cost bound function $\mathrm{cost}_{\max}: \r^{+} \rightarrow \r^{+}$
\item target unitary $U$ from $U_{targ} \subset \mathrm{SU}(2)$
\item target precision $\ve$
\end{enumerate}
Find $\rg{g_1}{g_N}$ from $G$ such that $\rho\at{\trg{g_1}{g_N},U} \le \ve$ and $\sum_{k=1}^{N} c\at{g_k} \le \mathrm{cost}_{\max} \at{\ve} $.
\end{prob}
We say that the algorithm solves UAP in polynomial time, if it solves Problem \ref{prob:uap} for arbitrary unitaries $U$ from $U_{targ}$ and its runtime is polynomial in $\log(1/\ve)$. We also allow to spend arbitrary time on precomputation based on (1) -- (4) and store arbitrary amount of results of the precomputation. The set $U_{targ}$ can be equal to $\mathrm{SU}(2)$ or some its subset. For example, it can be the set of all unitaries $e^{i \vp Z}$ for $Z$ being Pauli $Z$ matrix and $\vp$ being arbitrary real number.

The hardness of solving UAP and the existence of the solution to it depends on the choice of cost bound function $\mathrm{cost}_{\max}$. We summarize known algorithms for solving UAP in Table~\ref{tab:algo}. In practice, for target precisions $10^{-10}$ to $10^{-30}$, the overhead from using the Solovay Kitaev algorithm can be between one to three orders of magnitude\cite{KMM1,KBS}. On the other hand, the methods based on brute force search find the best possible solution, but are frequently limited to precisions $10^{-5}$ or even less because their runtime and required memory scale exponentially with $\log(1/\ve)$ \cite{Fowler:2011,KMM1}. Methods \cite{RS,S,BGS} together with the one that is the focus of this paper (see also Table~\ref{tab:algo}, lines 3-8) are based on exact synthesis algorithm and produce results that has cost bound similar to brute force search but have polynomial runtime subject to some number theoretic conjectures~(similar to Conjecture~\ref{cnj:main}). In the next subsection we will explain the idea of this methods in more detail.

\begin{table}
{\small
\begin{centering}
%auto-ignore
%!TEX root = quaf.tex

\begin{tabular}{ l c c c c c }
\hline
\hline
\\[-2ex]
$\quad$ Algorithm & Gate set  & Cost     & Cost bound & Set of target & runtime as  \tabularnewline
          &   $\CG$        & function & function   & unitaries     & function of \tabularnewline
          &           & $c:\CG\rightarrow \r^{+}$ &  $\mathrm{cost}_{\max}$          & $U_{targ}$              & $\log(1/\ve)$ \tabularnewline
\\[-2ex]
\hline \hline\\[-2ex]
1. Solovay-Kitaev  &    any        & $c\at{g} = 1$ & $ O\at{\log^{3.97}\at{1/\ve}}$ & SU(2) & polynomial \tabularnewline
\\[-2ex]
 $\quad$ \cite{DN:2005,KSV:2002} & universal & & & & \tabularnewline
\\[-2ex]
\hline
\\[-2ex]
2. brute-force    & unitaries with &                &                            &  &    \tabularnewline
$\quad$ search \cite{Fowler:2011}       & algebraic entries  & $c\at{g} = 1$ &     $ O\at{\log(1/\ve)}$ &  SU(2) & exponential      \tabularnewline
\\[-2ex]
\hline
\\[-2ex]
3. \cite{S,RS,KMM3}        & Clifford+T     & $c($T$)=1$, &  $ O\at{\log(1/\ve)}$ & $e^{iZ\vp}$ & polynomial \tabularnewline
          &                & $c($Clifford$)=0$ & & & \tabularnewline
\\[-2ex]
\hline
\\[-2ex]
4. \cite{BGS,BBG:2014}        & V-basis     & $c(g)=1$ &  $ O\at{\log(1/\ve)}$ & $e^{iZ\vp}$ & polynomial \tabularnewline
\\[-2ex]
\hline
\\[-2ex]
5.  \cite{KBS}       & Fibonacci     & $c(g)=1$ &  $ O\at{\log(1/\ve)}$ & $e^{iZ\vp},e^{iX\vp}$ & polynomial \tabularnewline
\\[-2ex]
\hline
\\[-2ex]
6. \cite{BRS:2015b}        & Clifford+$e^{i\pi Z/12}$     & $c(e^{i\pi Z/12})=1$, &   $ O\at{\log(1/\ve)}$ & $e^{iZ\vp}$ & polynomial \tabularnewline
          &                & $c($Clifford$)=0$ & & & \tabularnewline
\\[-2ex]
\hline
\\[-2ex]
7. this paper  & totally definite     & $c(g)=1,$ &  $ O\at{\log(1/\ve)}$ & $e^{iZ\vp}$ & polynomial \tabularnewline
 & quaternion algebras & $c_{\mathrm{canonical}}$ & &   & \tabularnewline
\\[-2ex]
 \hline \hline
\end{tabular}
\end{centering}
\caption{\label{tab:algo} Known algorithms for solving the unitary approximation problem (UAP, Problem~\ref{prob:uap}). The distance function used is $\nrm{U-V}$, where $\nrm{U} = \frac{1}{2}\sqrt{\tr\at{UU^{\dagger}}}$
. Cost function $c_{\mathrm{canonical}}$ is defined in Section~\ref{sec:es} and also discussed in more details in Section~\ref{sec:examples}. }
}
\end{table}

% Discuss the best possible solutions to UAP

\subsection{Approximation methods based on exact synthesis} \label{sec:esmethods}

Algorithms that solve UAP~(Problem~\ref{prob:uap}) and that are based on the exact synthesis algorithm (see Table~\ref{tab:algo}, lines 3-8) can be described by the flow chart on Figure~\ref{fig:online-flow}. The focus of this paper is such an algorithm that works for gate sets described by an arbitrary {\em totally definite} quaternion algebra and generalizes algorithms from \cite{S,BGS,BRS:2015b} (see also Table~\ref{tab:algo}, lines 3-6). We leave the analog of the algorithm for indefinite quaternion algebras for future work. In this subsection we first discuss what does it mean for the gate set to be described by totally definite quaternion algebra and next look at the flow on Figure~\ref{fig:online-flow} in more details.

We refer to Section~\ref{sec:basics} for definitions and more details discussion of mathematical objects discussed below. The aim of this part to explain connections between them and to the algorithm presented in this paper on a high level. We say that the gate set $\CG$ is described by quaternion algebra if the following list of objects can be specified and related to the gate set.

\begin{restatable}{dfn}{qgc}
\label{dfn:quaternion-spec} A quaternion gate set specification is a tuple $\ip{F,\s,a,b,\mo,S}$ where:
\begin{itemize} 
\item $F$ is a totally real number field and  $\s$ is an embedding of $F$ into $\r$
\item $a,b$ are elements of $F$ that define the quaternion algebra $\at{\frac{a,b}{F}}$ over $F$
\item $\mo$ is a maximal order of $\at{\frac{a,b}{F}}$
\item $S=\set{\p_1,\ldots,\p_M}$ is a set of prime ideals of $F$
\end{itemize}
\end{restatable}
Using the embedding $\s$ any quaternion $q$ from the quaternion algebra can be mapped to a special unitary $U_q \in SU(2)$. We discuss the precise construction of this map in Subsection~\ref{sec:quat-and-unitaries}. This map has the following important properties:
\[
 U_{q_1 q_2 } = U_{q_1} U_{q_2} , U^{\dagger}_{q}=U_{q^\ast}
\]
where $q^{\ast}$ is the conjugate of $q$. Let us also define the following closed under multiplication set
\[
\mo_S = \set{ q \in \mo : \nrd\at{q}\z_F = \trg{\p_1^{L_1}}{\p_M^{L_M}}, L_k \in \z,\,L_k \ge 0  },
\]
and call $\at{L_1,\ldots,L_M}$ the \emph{cost vector} of $q$. We will discuss the meaning of a cost vector in more details further in this section and also in Section~\ref{sec:es}. Above $\nrd\at{q}$ is the reduced norm of quaternion and $\z_F$ is a ring of integers of number field $F$. The set $\mo_S$ is closed under multiplication because $\mo$ is closed under multiplication and $\nrd\at{q_1 q_2} = \nrd\at{q_1} \cdot \nrd\at{q_2}$.

A simplified set of conditions that must hold for the  gate set to be described by the quaternion gate set specification is:
\begin{enumerate}
\item There must exist subset $\CG_Q$ of $\mo_S$ such that $\CG=\set{U_q : q \in \CG_{\mo,S} }$.
\item The group generated by $\CG$ must be equal to group $\set{U_q : q \in \mo_S }$.
\end{enumerate}
Condition (1) implies that group generated by elements from $G$ is a subgroup of $\set{U_q : q \in \mo_S }$. Condition (2) can be checked using the framework developed in the recent paper~\cite{KY1} given set $\CG_{\mo,S}$. We will give a brief overview of results in~\cite{KY1} in Subsection~\ref{sec:qfes-review}. One way of checking the condition (2) is to first compute a finite set of quaternions $\CG_{\mo,S}^{\star}$ such that every element of $\mo_S$ can be written as a product of elements of $\CG_{\mo,S}^{\star}$  and a scalar (using algorithms from~\cite{KY1}). We will say that $\CG_{\mo,S}^{\star}$ is a \emph{set of canonical generators} of $\mo_S$. Second, for each $q$ from $\CG_Q^{\star}$ find a representation of $U_q$ as a product of elements of $\CG$. For all $q$ from $\CG_{\mo,S}^{\star}$ we then can define:
\begin{equation}
\crt\at{q} = (U_1, \ldots, U_n) , \text{ where } U_q = U_1 \ldots U_n , U_k \in \CG. \label{eq:circuit}
\end{equation}
One natural way of defining the cost of elements of $\CG_{\mo,S}^{\star}$ is
\[
c\at{q} = \sum_{k=1}^n c\at{U_k}.
\]
For the other cost function definitions related to the cost vector of the quaternion see Section~\ref{sec:es}. To summarize we give the following definition:
\begin{restatable}{dfn}{qag}
\label{dfn:qag} We say that a gate set $\CG$ is described by quaternion algebra if the following data is defined:
\begin{enumerate}
\item A quaternion gate set specification $\ip{F,\s,a,b,\mo,S}$,
\item A set set of canonical generators $\CG_{\mo,S}^{\star}$ of $\mo_S$,
\item A map $\crt$ as described by equation (\ref{eq:circuit}).
\end{enumerate}
\end{restatable}

In Section~\ref{sec:examples} we will give examples of the gate sets described by quaternion algebras (including Clifford$+T$ and $V$ basis) and will explicitly specify (1)--(3) for each example. Now we have enough background to discuss the flow of our algorithm in more details.

Steps 2 and 3 in Figure~\ref{fig:online-flow} are the crucial in the algorithm. For simplicity we focus on the case when set $S$ consists of only one prime ideal. Our input to Step 2 is $L_1,\vp,\ve$. The output of Step 2 is a quaternion $q$ from $\mo_S$ such that its cost is equal $L_1$ and $U_q$ is within distance $\ve$ from $R_z\at{\vp}$. The fact that our target gate set $\CG$ is described by quaternion algebra (Definition~\ref{dfn:qag}) immediately implies that $U_q$ can be expressed as a circuit over $\CG$. Steps 3 and 4 construct such a circuit. In Step 3 we express $q$ as a product $q_1\ldots q_n$ of elements of $\CG_{\mo,S}^{\star}$ and a scalar from $F$ using exact synthesis algorithm described in~\cite{KY1}. In Step 4 we find a circuit for $U_q$ over gate set $\CG$ as concatenation of circuits for each $q_k$.

Let us now discuss Step 1 in more detail. To give some intuition we start with the Clifford+T gate set example~(analyzed in more details in Section~\ref{sec:examples}). In this case set $S$ contains precisely one prime ideal $\p_1$ and $L_1$ is greater or equal to the T-count of the resulting circuit. To ensure that approximation step succeeds the input to the algorithm must satisfy inequality
\[
L_1\log\at{N\at{\p_1}} \ge 4\log(1/\ve) + C_{\min} \text{, where } N\at{\p_1} \text{ is the norm of } \p_1.
\]
This reproduces result in~\cite{S} that the $T$-count scales as $4 \log_2(1/\varepsilon) + C_{\min}$  because $N(\p_1)=2$ holds for the Clifford$+T$ case. The bound also saturates the lower bound proved in~\cite{S} up to an additive constant. In this simple case, in Step 1 of Figure~\ref{fig:online-flow} we just assign $L_1 = \ceil{ (4\log(1/\ve) + C_{\min})/\log\at{N\at{\p_1}} }$. In our algorithm we precompute constant $C_{\min}$ based on quaternion gate set specification.

More generally, the cost vector $\at{L_1,\ldots,L_M}$ that we input to Step 2 must satisfy the following inequality
\[
L_1\log\at{N\at{\p_1}} + \ldots + L_M\log\at{N\at{\p_M}} \ge 4\log(1/\ve) + C_{\min}.
\]
The length of the circuit output by our algorithm is proportional to $L_1 + \ldots + L_M$ and therefore proportional to $\log(1/\ve)$ which is up to multiplicative factor is the best possible. Cost optimality up to an additive constant is more subtle and is dependent on the choice of cost function and the gate set. We will discuss this more in Section~\ref{sec:examples}.

\begin{figure}[ht!]
\rule[0.5ex]{1\columnwidth}{1pt}

\begin{centering}
%auto-ignore
%!TEX root = quaf.tex

\usetikzlibrary{matrix}
\usetikzlibrary{shapes} %
\usetikzlibrary{shapes.geometric} %
\usetikzlibrary{shapes.multipart} %

\begin{tikzpicture}[
	every text node part/.style={align=left},
	block/.style = {rounded rectangle,fill=yellow!20!white,draw=black,align=center, inner sep=0.5em} 
]

\fill[white] (-7,-5) rectangle (7,5);
\matrix(M)[  column sep=1cm, row sep=0.8cm,
  matrix of nodes,
  nodes=block ]{
 
{}   \\
1. Cost optimizer  \\
2. Quaternion Approximation \\
3. Exact synthesis algorithm \\
4. Circuit rewriter \\
\\
\\
{} \\
};

\fill[white] (M-1-1) circle (0.3cm);
\draw[->] (M-1-1) -- (M-2-1) node[midway,right] () {target angle $\varphi$, target pecision $\varepsilon$};
\draw[->] (M-2-1) -- (M-3-1) node[midway,right] () {Cost vectors, $\varphi,\varepsilon$};
\draw[->] (M-3-1) -- (M-4-1) node[midway,right] () {Quaternions};
\draw[->] (M-4-1) -- (M-5-1) node[midway,right] () {Products of quaternions };
\fill[white] (M-8-1) circle (0.3cm) node [above] (end) {""};
\draw[->] (M-5-1) -- (end) node[midway,right] () {Circuits over \\ the gate set $G$  such that \\ corresponding unitaries $U$ \\ satisfy: $d(R_z(\varphi),U)\le\varepsilon$ };

\end{tikzpicture}
\par\end{centering}

\rule[0.5ex]{1\columnwidth}{1pt}

\caption{\label{fig:online-flow} High level flow of the approximation algorithm. }
\end{figure}
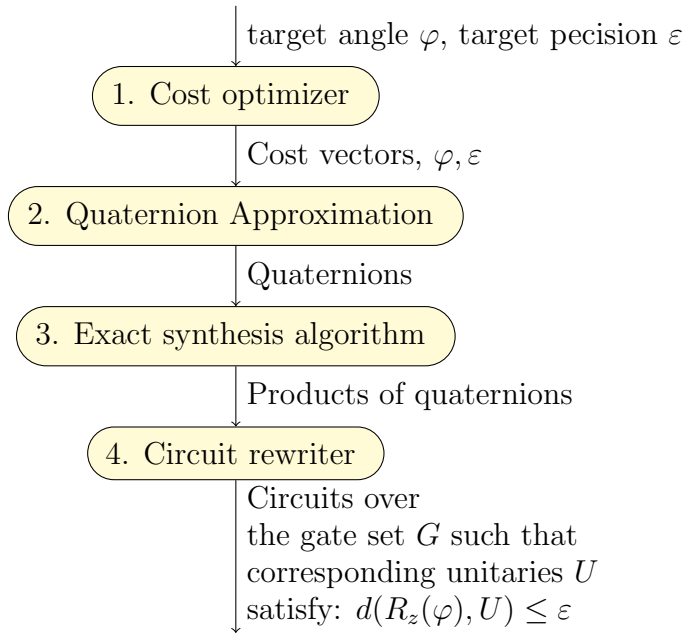

In some cases we might not have very fine control of the cost of output circuit using cost vector. In the worst case cost vector will allows us to control the cost of the output circuit up to multiplicative factor. In this situation we can use the following strategy to improve the cost of output
\begin{itemize}
\item In Step 1 make a request for several cost vectors
\item In Step 2 for each cost vector request several quaternions that achieve required quality of approximation
\item In Step 4 try to optimize final circuits using rewriting rules for gate set $\CG$ and pick the ones with minimal cost. Rewriting rules for the gate set $\CG$ can be obtained using exact synthesis framework~\cite{KY1}.
\end{itemize}

More formally the problem we are solving on Step 2 in Fig.~\ref{fig:online-flow} is the following: 

\begin{restatable}{prob}{qap}
(Quaternion approximation problem, QAP). \label{prob:qap}
Given
\begin{enumerate}  
\item A quaternion gate set specification $\ip{F,\s,a,b,\mo,S = \set{\p_1,\ldots,\p_M} }$,
\item target angle $\vp$ 
\item target precision $\ve$
\item target cost vector $\at{L_1,\ldots,L_M}$ satisfying 
\[
 L_1 \log N \at{\p_1} + \ldots + L_M \log N \at{\p_M} - 4\log\at{1/\ve} \in \of{C_{\min},C_{\max}},
\] where constants $C_{\min},C_{\max}$ depend only on the quaternion gate set specification.
\end{enumerate}
Find $q$ from the generalized Lipschitz order (see Section~\ref{sec:lip-order}) in quaternion algebra $\at{\frac{a,b}{F}}$ such that $\nrd\at{q}\z_F = \p_1^{L_1} \ldots \p_M^{L_M}$ and $\nrm{R_z\at{\vp} - U_q} \le \ve$, where $\nrm{U} = \frac{1}{2}\sqrt{\tr\at{UU^{\dagger}}}$.
\end{restatable}

The polynomial time algorithm (in $\log\at{1/\ve}$) for QAP gives us polynomial time algorithm for solving unitary approximation problem for gate sets that can be described by totally quaternion algebra. The circuit for $U_q$ can be found in time polynomial in $L_1,\ldots,L_M$ using exact synthesis algorithm from~\cite{KY1}. The cost of the resulting circuit is a linear function in $L_1,\ldots,L_M$ and therefore we can solve UAP with cost bound function that is in $\Theta\at{\log\at{1/\ve}}$. Next we go through the basic definitions needed to describe our solution to QAP. 
% \subsection{Short overview of the paper}

%auto-ignore
%!TEX root = quaf.tex

\section{Basic results and definitions}\label{sec:basics}

\subsection{Using quaternions to represent unitaries}\label{sec:quat-and-unitaries}

Let $F$ be a totally real number field of degree $d$. Let $\s_{1},\ldots,\s_{d}$
be embeddings of $F$ into $\r$. Let $\zf$ be a ring of integers
of $F$. Let $a,b$ be two totally negative elements of $\zf$ . In
other words for all $k=\rg 1d$ we have $\s_{k}\at a<0,\s_{k}\at b<0$.
We next consider a quaternion algebra $Q=\at{\frac{a,b}{F}}$ given by 
\[
Q=\set{a_{0}+a_{1}\i+a_{2}\j+a_{3}\k\,:\,a_{0},a_{1},a_{2},a_{3}\in F},
\]
where $\i^{2}=a$, $\j^{2}=b$ and $\k=\i\j=-\j\i$. The fact that
$a,b$ are totally negative implies that $Q$ is totally definite
quaternion algebra. The conjugate of a quaternion $q=a_{0}+a_{1}\i+a_{2}\j+a_{3}\k$
is defined as $q^{\ast}=a_{0}-a_{1}\i-a_{2}\j-a_{3}\k$. The reduced norm $\nrd$
and reduced trace $\trd$ are defined as 
\[
\nrd\at q=qq^{\ast},\; \trd\at q=q+q^{\ast}
\]
Let $\s=\s_{1}$ be a fixed embedding that we will use to construct
unitaries out of quaternions. We first define a homomorphism from quaternion
algebra into the algebra of complex $2\times2$ matrices as follows:
\begin{eqnarray}
h\at q&=&h\at{a_{0}+a_{1}\i+a_{2}\j+a_{3}\k}\nonumber \\ &=&\s\at{a_{0}}I+\s\at{a_{1}}\sqrt{\s\at a}Z+\s\at{a_{2}}\sqrt{\s\at b}Y+\s\at{a_{3}}\sqrt{-\s\at b\s\at a}X\label{eq:hom-q-to-m}
\end{eqnarray}
Here $I,X,Y,Z$ are the four Pauli matrices. Note that $h$ has additional nice properties:
\[
\det\at{h\at q}=\sigma\at{\nrd\at q},\:\tr\at{h\at q}=\sigma\at{\trd\at q}.
\]

To construct special unitaries out of quaternions we use the following
mapping: 
\begin{equation}
U\at q=\frac{h\at q}{\sqrt{\det\at{h\at q}}}=\frac{h\at q}{\sqrt{\s\at{\nrd\at q}}}\label{eq:q-to-u}
\end{equation}
Note that for any non-zero $\alpha$ from $F$ we have the following
\[
U\at{\alpha q}=\frac{\s\at{\alpha}}{\abs{\s\at{\alpha}}}U\at q=\pm U\at q.
\]

Let us now study the structure of image of $U\at q$ in more details.
We will express this structure using a number field that can be embedded
into $Q$. Let $K=F\at{\sqrt{a}}$ be a totally imaginary extension
of $F$ of degree $2$. Such number fields $K$ are called CM fields \cite{Washington:82}. This
is ensured by the condition that $\sigma_{k}\at a<0,k=\rg 1d.$ Let $\beta$
be an element of $K$ such that $\beta^{2}-a=0$. The degree of the
field $K$ is $2d$ and there are $2d$ embeddings of $K$ into $\c$.
Each element of $K$ can be represented as $a_{0}+a_{1}\beta$ where
$a_{0}$ and $a_{1}$ are from $F$. We can define $2d$ embeddings
of $K$ into $\c$ as following:
\begin{eqnarray*}
\s_{k,+}\at{a_{0}+a_{1}\beta} & = & \s_{k}\at{a_{0}}+i\s_{k}\at{a_{1}}\sqrt{\abs{\s_{k}\at a}}\\
\s_{k,-}\at{a_{0}+a_{1}\beta} & = & \s_{k}\at{a_{0}}-i\s_{k}\at{a_{1}}\sqrt{\abs{\s_{k}\at a}}
\end{eqnarray*}
Futher we will use notation $\s$ for $\s_{1,+}$ which is in agreement
with $\s=\s_{1}$ for elements of $F$ because for $a_{0}$ from $F$
we have $\s_{k,\pm}\at{a_{0}}=\s_{k}\at{a_{0}}$. 

Each element of the quaternion algebra 
\[
q=a_{0}+a_{1}\i+a_{2}\j+a_{3}\k=\at{a_{0}+a_{1}\i}+\at{a_{2}+a_{3}\i}\j
\]
 can be mapped to two elements of $K$ in the following way: 
\begin{eqnarray}
e_{1}\at q & = & a_{0}+\beta a_{1},\:e_{2}\at q=a_{2}+a_{3}\beta.\label{eq:nf-projection}
\end{eqnarray}
Conversely, the map $e_{1}^{-1}$ describes an embedding of $K$ into quaternion
algebra $Q$. Note that now homomorphism $h\at q$ can be written as:
\begin{equation}
h\at q=\UG{\s\at{e_{1}\at q}}{\s\at{e_{2}\at q}\sqrt{\abs{\s\at b}}}{-\s\at{e_{2}\at q}^{\ast}\sqrt{\abs{\s\at b}}}{\s\at{e_{1}\at q}^{\ast}}.\label{eq:unitary-CM-field}
\end{equation}
Using this notation we also have that: 
\[
\s\at{\nrd\at q}=\abs{\s\at{e_{1}\at q}}^{2}+\abs{\sigma\at b}\abs{\s\at{e_{2}\at q}}^{2}.
\]
Or in other words, in terms of relative norm $N_{K/F}$ we have: 
\[
\nrd\at q=N_{K/F}\at{e_{1}\at q}-bN_{K/F}\at{e_{2}\at q}.
\]
For any CM field we can define an automorphism $^{\ast}:K\rightarrow K$
which is called complex conjugation and which has the following properties: 
\begin{eqnarray*}
\s_{k,\pm}\at{\at{a_{0}+a_{1}\beta}^{\ast}} & = & \s_{k,\pm}\at{a_{0}+a_{1}\beta}^{\ast},\\
\at{a_{0}+a_{1}\beta}^{\ast} & = & a_{0}-a_{1}\beta.
\end{eqnarray*}
Using it we can express our relative norm $N_{K/F}\at x=xx^{\ast}$
and see that $\s_{k}\at{N_{K/F}\at x}=\abs{\s_{k,\pm}\at{a_{0}+a_{1}\beta}}^{2}.$
In addition we have that 
\[
e_{1}\at q^{\ast}=e_{1}\at{q^{\ast}}, \; e_{1}^{-1}\at q^{\ast}=e_{1}^{-1}\at{q^{\ast}}.
\]

\subsection{Distance to \texorpdfstring{$R_{z}$}{Rz} rotations } \label{sec:distance}

The distance functions that we use for unitaries are 
\begin{equation}
\rho\at{U,V}=\sqrt{1-\frac{\abs{\tr\at{UV^{\dagger}}}}{2}},\,d_{2}\at{U,V}=\frac{1}{2}\sqrt{\tr\at{\at{U-V}\at{U-V}^{\dagger}}}.\label{eq:pdist}
 \end{equation}
Our notation for $R_{z}$ is the following: 
\begin{eqnarray}
% R_{x}\at{\vp} & =e^{-\frac{i\vp}{2}X}= & \UG{\cos\at{\frac{\vp}{2}}}{-i\sin\at{\frac{\vp}{2}}}{-i\sin\at{\frac{\vp}{2}}}{\cos\at{\frac{\vp}{2}}}\label{eq:rz-rx-ry}\\
% R_{y}\at{\vp} & =e^{-\frac{i\vp}{2}Y}= & \UG{\cos\at{\frac{\vp}{2}}}{-\sin\at{\frac{\vp}{2}}}{\sin\at{\frac{\vp}{2}}}{\cos\at{\frac{\vp}{2}}}\\
R_{z}\at{\vp} & =e^{-\frac{i\vp}{2}Z}= & \UG{e^{-i\vp/2}}00{e^{i\vp/2}}.
\end{eqnarray}

%\subsubsection{Distance to \texorpdfstring{$R_{z}$}{Rz} rotations}

Let us now further analyze the distance between $R_z\at{\vp}$ and unitary $U_q$ for a given quaternion $q$: 
\[
\rho\at{R_{z}\at{\vp},U_q} = \rho\at{R_{z}\at{\vp},\frac{h\at q}{R}}.
\]
where $R=\sqrt{\s\at{\nrd\at q}}$ and $h\at q$ is defined by Equation~\ref{eq:hom-q-to-m}.
We can further rewrite it as: 
\[
\rho\at{R_{z}\at{\vp},\frac{1}{R}\UG z{-w^{\ast}}w{z^{\ast}}},
\]
where $z=\s\at{e_{1}\at q}$, $w=-\s\at{e_{1}\at q}^{\ast}\sqrt{\abs{\sigma\at b}}$
and $\abss x+\abss y=R^2$. Now we are going to solve the inequalities
: 
\begin{equation}
\rho\at{R_{z}\at{\vp},\frac{1}{R}\UG z{-w^{\ast}}w{z^{\ast}}}\le\ve\label{eq:Rz-dist-ineq}
\end{equation}

\begin{equation}
d_{2}\at{R_{z}\at{\vp},\frac{1}{R}\UG z{-w^{\ast}}w{z^{\ast}}}\le \ve\label{eq:Rz-dist2-ineq}
\end{equation}
Inequalities above do not constrain $w$. Introducing $z_0=R(1-\ve^2)e^{-i\vp/2}$ inequality (\ref{eq:Rz-dist-ineq}) simplifies to the following two inequalities:
\[
	\re\at{\at{z-z_0 }e^{i\vp/2}} \ge 0 \text{ or } \re\at{\at{z+z_0}e^{i\vp/2}} \le 0,
\]
and inequality (\ref{eq:Rz-dist2-ineq}) simplifies to 
\[
	\re\at{\at{z-z_0}e^{i\vp/2}} \ge 0.
\]
In additon, the fact that $U_q$ is a unitary matrix implies that $\abs{z}\le R$. See Figure~\ref{fig:inequality} for the visualization of the solution regions. We summarize result of the section in the following proposition: 
\begin{prop} \label{prop:distance} Let $q=e_1^{-1}\at{z_1}+e_2^{-1}\at{z_2}$ and let $\s_{1,+}\at{z_1}$ belong to the following set
\[
\set{ z \in \c : \re\at{\at{z-z_0}e^{i\vp/2}}\ge 0, \abs{z} \le R } 
\]
(blue region on Fig.~\ref{fig:inequality}) where $z_0=R(1-\ve^2)e^{-i\vp/2}$ and $R=\sqrt{\sigma_1\at{\nrd\at{q}}}$. Then $d_2\at{U_q,R_z\at{\vp}} \le \ve$ and $\rho\at{U_q,R_z\at{\vp}} \le \ve$. 
\end{prop}

\begin{figure}[t]
\rule[0.5ex]{1\columnwidth}{1pt}

\begin{centering}
%auto-ignore
%!TEX root = quaf.tex

\begin{tikzpicture}[scale=3,>=latex]
\pgfmathsetmacro{\textshift}{1}
\pgfmathsetmacro{\axessize}{1.4}
\pgfmathsetmacro{\xmax}{3}
\pgfmathsetmacro{\xmin}{-\xmax}
\pgfmathsetmacro{\D}{1.3} %D
\pgfmathsetmacro{\eps}{0.5} %epsilon
\pgfmathsetmacro{\phi}{40} %phi 
\pgfmathsetmacro{\cPh}{cos(-\phi/2)}
\pgfmathsetmacro{\sPh}{sin(-\phi/2)}
\pgfmathsetmacro{\SqrtD}{sqrt(\D)}
\pgfmathsetmacro{\SDE}{\SqrtD*(\eps*\eps-1)}
\pgfmathsetmacro{\DY}{sqrt(\D*\eps*\eps*(2-\eps*\eps)*\cPh*\cPh)}
\pgfmathsetmacro{\DX}{\DY*\sPh / \cPh}
\pgfmathsetmacro{\xf}{\SDE*\cPh}
\pgfmathsetmacro{\yf}{\SDE*\sPh}

\tkzDefPoint(0,0){O}
\tkzDefPoint(-\xf+\DX,-\yf-\DY){A}
\tkzDefPoint(-\xf-\DX,-\yf+\DY){B}
\tkzDefPoint( \xf+\DX, \yf-\DY){C}
\tkzDefPoint(\xf-\DX,\yf+\DY){D}
\tkzDefPoint(\xf,\yf){CC}
\tkzDefPoint(-\xf,-\yf){Z0}
\tkzDefPoint(\SqrtD * \cPh, \SqrtD * \sPh){R}

\clip(\xmin,-1.2) rectangle (\xmax,1.8);

\begin{scope}
\tkzClipCircle(O,R)
\fill [cyan!20!white] (A) -- (B) -- ($2*(B)$) -- ($2*(A)$);
\fill [magenta!20!white] (C) -- (D) -- ($2*(D)$) -- ($2*(C)$);
\end{scope}

\tkzDrawCircle(O,R)
\draw (A)--(B);
\draw (C)--(D);

\node[right] (pA) at (A) {$\, z_0 - ide^{-i\vp/2}$};
\node[right] (pB) at (B) {$\, z_0 + ide^{-i\vp/2}$};

\node[left] (pC) at (C)  {$\,- z_0 - ide^{-i\vp/2}$};
\node[left] (pD) at (D) {$\, -z_0 + ide^{-i\vp/2}$};
\node[right] (pC) at (CC)  {$\,- z_0$};
\node[right] (pC) at (Z0)  {$\, z_0$};

\node[right] (pD) at (R) {$\, R e^{-i\vp/2}$};

\draw[->] (O)--(R) node [midway, below, sloped] (tN) {$R$};
\draw[-] (0.5,0) arc (0:-\phi/2:0.5) node [midway, right ] (tN) {$\varphi/2$};
\draw[<->] ($-1*(R)$)--(CC) node [midway, below, sloped ] (tN) {$\ve^2 R$};

%axes
\draw[->,gray] (O)--(\axessize,0) node [left,below,black] (tN) {$\quad x=\Re(z)$};
\draw[->,gray] (O)--(0,\axessize) node [above,right,black] (tN) {$y=\Im(z)$};
\node () at (O) {} node [left] (tN) {$(0,0)$};

\node () at (1.7,1.1) {$\begin{array}{lcl}%
z_0 & = & R(1-\ve^2)e^{-i\vp/2} \\ % 
d & = & R\ve\sqrt{2-\ve^2}\\ %
% d_{\varepsilon} & = & \varepsilon^2 R  \\ %
\end{array}$};

\tkzDrawPoints[color=black,fill=red,size=6](O,A,B,C,D,R,CC,Z0);

\end{tikzpicture}
\par\end{centering}

\rule[0.5ex]{1\columnwidth}{1pt}

\protect\caption{\label{fig:inequality}Colored regions correspond to complex numbers
$z$ that satisfy inequality~\ref{eq:Rz-dist-ineq} and $\abs{z}\le R$ given parameters
$\protect\vp,\protect\ve,R$. The blue region corresponds to complex numbers
$z$ that satisfy inequality~\ref{eq:Rz-dist2-ineq} and $\abs{z}\le R$. The vertical axis
is $\protect\im\protect\at z$ and the horizontal axis is $\protect\re\protect\at z$.}
\end{figure}
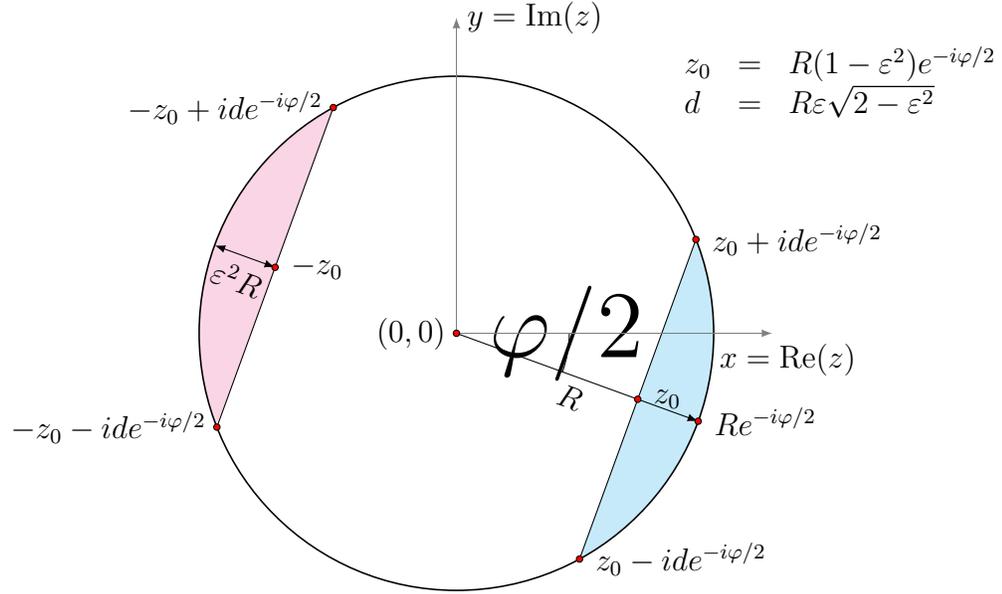

% \subsubsection{Distance to \texorpdfstring{$R_{x}$ and $R_{y}$}{Rx and Ry} rotations }

% \begin{itemize}
% \item Use Hadamard and Phase gates + unitary invariance of the distance to reduce the problem to $R_{z}$ case over possibly different quaternion algebra (e.g. switch $a$ and $b$)
% \end{itemize}

\subsection{Exact synthesis results for totally definite quaternion algebras} \label{sec:qfes-review}

%auto-ignore
%!TEX root = quaf.tex 
\label{sec:es}

We recall several definitions about rings, ideals, orders, and quaternions that are required to study a special case of \cite{KY1}, namely the case of  totally definite quaternion algebras. We closely follow Section 2 of~\cite{KY1}. For further references to the literature and other facts used in the following we refer to the reference section of \cite{KY1}. Other good references on the topic include \cite{Vigneras:1980,Reiner:1975,Gille:2006,Lam:2005,KV}. 

Let $Q$ be a quaternion algebra over number field $F$ (defined in Section~\ref{sec:quat-and-unitaries}). A \emph{$\z_F$-lattice} $I$ is finitely generated $\z_F$-submodule of $Q$ such that $F\,I = Q$. In other words $I$ has a full rank in $Q$. An \emph{order} $\CO$ is a $\z_F$-lattice that is a subring of $Q$. An order is a \emph{maximal order} if it is not properly contained in any other order. There is a right and left order associated with any $\z_F$ lattice $I$ defined as 
\[
 \CO_R\at{I} = \set{ q \in Q : Iq \subset I }, \CO_L\at{I} = \set{ q \in Q : q I \subset I }.
\]
When we want to emphasize that $I$ has particular right and left order we call $I$ \emph{right-$\CO_R\at{I}$ fractional ideal} or \emph{left-$\CO_L\at{I}$ fractional ideal}. A fractional right-$\CO$ ideal is a \emph{normal ideal} if the order $\CO$ is a maximal order. Note that, order $\CO_R\at{I}$ is maximal if and only if $\CO_L\at{I}$ is maximal. All \emph{normal ideals} are invertible. A normal ideal $I$ is \emph{principal} if $I = q \CO_R{I}$ for some $q$ from $Q$. 

A normal ideal $I$ is two sided $\CO$-ideal if $\CO_R\at{I} = \CO_L\at{I} = \CO $. The principal two sided $\CO$-ideals form a subgroup of the group of all two sided $\CO$-ideals (under multiplication). The quotient of the group of all two sided $\CO$ ideals modulo principal two sided $\CO$-ideals is the \emph{two-sided ideal class group} of $\CO$. It is known that the two sided ideal class group of $\CO$ is always finite. The \emph{two sided class number} of $Q$ is the size of the two-sided ideal class group of any maximal order of $Q$. It known that the size of two-sided ideal class group is independent on the choice of maximal order $\CO$. Here we restrict our attention to the special case of results from \cite{KY1} for totally definite quaternion algebras with two sided class number $1$. All examples considered in this paper has this property.

Let $\disc\at{q_1,\ldots,q_4} = \det(\trd\at{q_i q_j}_{i,j =\rg{1}{4}})$. 
The \emph{discriminant} of an order $\CO$ is $\z_F$ ideal generated by the set
\[
 \set{ \disc\at{q_1,\ldots,q_4} : q_1, \ldots, q_4 \in   \CO  }.
\]
It turns out that the discriminant always is a square. Its square root is \emph{the reduced discriminant} denoted by $\disc\at{\CO}$. All maximal orders in quaternion algebra $q$ has the same discriminant. The reduced norm of $\z_F$ lattice $I$ is $\z_F$ ideal $\nrd\at{I}$ generated by $\set{ \nrd\at{q} : q \in I }$. 

The unit group $\CO^{\times}$ of $\CO$ is $\set{ q \in \CO : \nrd{q} \in \z^{\times}_F }$ where $\z^{\times}_F$ is a unit group of $\z_F$. For orders in totally definite quaternion algebras the quotient group $\CO^{\times} / \z^{\times}_F$ is always finite. The \emph{normalizer} of order $\CO$ is the set 
\[
\mathrm{Normalizer}\at{\CO} = \set{ q \in \CO : q \CO q^{-1} = \CO }. 
\]
which is a monoid under under multiplication. For totally definite quaternion algebras the quotient $\mathrm{Normalizer}\at{\mo} / \z_F$ (considered as a quotient of two monoids) is finite similarly to $\CO^{\times} / \z^{\times}_F$. 

We say that $\nrd\at{q}$ is \emph{supported} on the set $S$ of primes ideals of $\z_F$ if: 
\begin{equation} \label{eq:norm-val}
\nrd\at{q}\z_F = \prod_{\p \in S} \p^{v(q,\p)}
\end{equation} We also recall that map $T_2 : F \rightarrow \r^{+}$ is defined as:
\begin{equation}
T_2\at{x} = \sum_{k=1}^d \s^2_k \at{x}.
\end{equation}

We have now all definitions in place to state the special case of the one of main results of~\cite{KY1} (Theorem 3.18) for totally definite quaternion algebras with two sided class number 1.
\begin{thm} \label{thm:es:main}
Let $Q$ be a totally definite quaternion algebra over totally real number field $F$ with two sided class number one, let $\mo$ be a maximal order in $Q$, let $S = \at{\p_1,\ldots,\p_M}$ be a finite set of prime ideals of $\z_F$. There exists set $\mathrm{gen}_S\at{\mo}$  such that every quaternion $q$ from the set 
\[ 
\mo_S = \set{ q \in \mo : \nrd\at{q} \text{ is supported on S } } 
\]
can be written as the product $q_1 \ldots q_n q_{\mathrm{rem}} $ where $q_1, \ldots, q_n$ are from $\mathrm{gen}_S\at{\mo}$ and  $q_{\mathrm{rem}}$ is from $\mathrm{Normalizer}\at{\mo}$. If all ideals from $S$ do not divide $\disc\at{\mo}$ then $q_{\mathrm{rem}}$ is from $\mo^{\times} $. 

There exist algorithms for deciding if the set  $\mathrm{gen}_S\at{\mo}$ is finite and computing it if this is the case. There is also an algorithm for finding factorization $q_1 \ldots q_n q_{r}$ in time polynomial in $\log T_2\at{\nrd\at{q}}$. 
\end{thm} 

To find the factorization we essentially do trial division of $q$ by elements of $\mathrm{gen}_S\at{\mo}$ and greedily reduce values $v\at{\p,q}$ in equation (\ref{eq:norm-val}) on each step. We perform trial-division step until we are left with an element of $\mathrm{Normalizer}\at{\mo}$. The map $U_q$ discussed in Section~\ref{sec:quat-and-unitaries} depends only on $q_r / F^{\times}$ up to a sign, therefore there is only finitely many possible unitaries $U_{q_r}$. The canonical gate set corresponding to $\mo_S$ is 
\[
\CG^{\star}_{\mo,S} = \mathrm{gen}_S\at{\mo} \cup  \mathrm{Normalizer}\at{\mo}/\z_F \cup \set{-1} .
\]
The main difficulty of the exact synthesis of quaternions and unitaries is computing $\mathrm{gen}_S\at{\mo}$ such that described simple trial division algorithm works. For specific examples illustrating the above definitions, our approximation and exact synthesis algorithms we refer the reader to Section~\ref{sec:examples}. 

The canonical cost function for $U_q$ can be defined using $v\at{q,\p}$ (see Equation~\ref{eq:norm-val}) as: 
\[
c_{\mathrm{canonical}}\at{U_q} = \min_{ q' \in \mo_S : U_q = U_{q'} } \mathrm{cost}\at{q'}\text{, where } \mathrm{cost}\at{q} = \sum_{\p \in S, \p \nmid \disc\at{\mo} } v\at{q,\p}
\]
As we will discuss in  more details in Section~\ref{sec:examples}, the canonical cost function corresponds to the $T$-count for Clifford+$T$ case and to the $V$-count for $V$-basis case. The cost vector $\at{\rg{L_1}{L_M}}$ of quaternion $q$ is equal to $\at{v\at{q,\p_1},\ldots,v\at{q,\p_M}}$. Given the cost vector it is always possible to upper bound $c_{\mathrm{canonical}}\at{U_q}$ as: 
\[
c_{\mathrm{canonical}}\at{U_q} \le \mathrm{cost}\at{q} = \sum_{k = \rg{1}{M} : \p_k \nmid \disc\at{\mo}} L_k.
\]
For the decomposition $q = q_1 \ldots q_n q_{r} $ described in the Theorem~\ref{thm:es:main} we have $\mathrm{cost}\at{q} = \sum_{k=1}^n \mathrm{cost}\at{q_k}$. This also implies that the length of the circuit corresponding to $U_q$ can be upper bounded by the function linear in cost vector $\at{\rg{L_1}{L_M}}$. 
% The product $IJ$ of two normal ideals $I$ and $J$ is proper if $\CO_R\at{I} = \CO_L\at{J}$. When the product $IJ$ is proper $\z_K$-lattice $IJ$ is a right-$\CO_L\at{I}$ ideal and a left-$\CO_R\at{J}$ fractional ideal. Normal ideal is integral if $I \subset \CO_R\at{I}$. Note that for any normal ideal $I \subset \CO_R\at{I}$ if and only if $I \subset \CO_L\at{I}$. 

\subsection{Generalized Lipschitz order}\label{sec:lip-order}

The Lipschitz order $L$ in the quaternion algebra $\at{\frac{-1,-1}{\q}}$
can be expressed in the following way: 
\[
L=e_{1}^{-1}\at{\z\of i}+e_{2}^{-1}\at{\z\of i}=\z+\z\i+\z\j+\z\k,
\]
where $\z\of i$ is a maximal order of $\q\at{\sqrt{-1}}$. For the definition of $e_1,e_2$ see eq.~(\ref{eq:nf-projection}) in Section~\ref{sec:quat-and-unitaries}. We generalize this construction to arbitrary totally definite quaternion algebra. 
Let $\z_{K}$ be a ring of integers of $K$. It is a two dimensional
$\z_{F}$ module, therefore it has a $\z_{F}$ pseudo basis and can
be written (in modified Hermite Normal Form, see Section 1.5, or
Corollary 2.2.9 in \cite{Coh2}) as: 
\[
\z_{K}=\z_{F}+\gamma I
\]
 where $I$ is integral $\z_{F}$ ideal and $\gamma$ an element of $K$
such that $1,\gamma$ is a $F$-basis of $K$. 

% \begin{rem}
% MAGMA seems to compute basis in this form, though documentation does
% not tell so. This this why we have an \textbf{assert} statement in
% the code checking this condition. See Figure~\ref{fig:code-liphitz-order}
% \end{rem}

We define the generalized Lipschitz order as: 
\begin{eqnarray}
L & = & e_{1}^{-1}\at{\z_{K}}+e_{2}^{-1}\at{\z_{K}}\nonumber \\
L & = & \z_{F}+e_{1}^{-1}\at{\gamma}I+\z_{F}\j+e_{1}^{-1}\at{\gamma}\j I\label{eq:liphitz-order}
\end{eqnarray}

% Code on Figure~\ref{fig:code-liphitz-order} constructs $L$ according
% to equation~\ref{eq:liphitz-order}. 

% Our goal is to approximate unitaries with quaternions from Maximal
% order $\mo$ that contains $L$. We additionally require our quaternions
% to have norm supported on set $S=\set{\p:\p\mbox{ -prime ideal of }\z_{F}}$.
% In other words $\nrd\at q\z_{F}=\prod_{\p\in S}\p^{e\at{\p}}$. 
\subsection{Lattices} \label{sec:lattice}
In the paper we use lattices related a) to the ring of integers $\z_K$ of CM field $K$, b) to ideals in $\z_K$, and c) to the unit group of a the totally real subfield of $F$ of $K$. We briefly recall here the necessary definitions related to integer lattices. More detailed discussion of the definitions and related results can be found in \cite{MG:2002}. 

Let $B=\{b_1,\ldots,b_n\}$ be a set of linearly-independent vectors in $\bR^m$, where $m \geq n$. The discrete group $\CL(B) = B(\bZ^n) = \bZ \, b_1 + \ldots + \bZ \,b_n$ is called the \emph{integer lattice of rank $n$} with basis $B$. Let $\span(\CL) = \bR\CL$ be the real span of an $n$-dimensional lattice $\CL$ and write $\span(\CL(B)) = \span(B)$. A subset $\CF \subset \span(\CL)$ is called a \emph{fundamental domain} of the lattice $\CL$ if for every vector $t \in \span(\CL)$ there exists a unique lattice vector $v(t) \in \CL$ such that $t-v(t) \in \CF$. There are at least two different centrally-symmetric fundamental domains associated with each lattice basis. The \emph{centered fundamental parallelepiped} $\CC(B)$ associated to a lattice basis $B$ is given by the inequalities
\[\CC(B) = B[-\smfrac 12,\smfrac 12)^n = \{B\, x \, : \, -1/2 \leq x_k < 1/2, k=1,\ldots,n\}.\]
The second fundamental domain is defined in terms of \emph{Gram-Schmidt orthogonalization (GSO)} $B^* =[b_1^*,\ldots,b_n^*]$  of a lattice basis $B$: 
\begin{eqnarray} \label{eqnarr:GSO}
b_1^* &=& b_1 \\ \nonumber
b_i^* &=& b_i - \sum_{j < i} \mu_{i,j} b_j^*,
\end{eqnarray}
where the \emph{orthogonalization coefficients} $\mu_{i,j}$ are defined as
\[\mu_{i,j}= \langle b_i, b_j^* \rangle/\langle  b_j^*, b_j^* \rangle.\]
Note that GSO of a lattice basis is not necessary a basis of $\CL(B)$.  It is related to the original basis via
\[B = B^* \pmat{1  & &   \mu_{ji} \\  &  \ddots \\ 0 & & 1 }.\]

The \emph{centered orthogonalized fundamental parallelepiped} $\CC(B^*)$ associated to a lattice basis $B$ is given by the inequalities
\[\CC(B^*) = B^*[-\smfrac 12,\smfrac 12)^n =  \{B^*\, x \, : \, -1/2 \leq x_k < 1/2, k=1,\ldots,n\}\]
where $B^*$ is the Gram-Schmidt orthogonalization of $B$.
% \end{dfn}

For every integer lattice of rank $n > 1$ there are infinitely many choices of bases. Indeed, for any transformation $G \in \GL_n(\bZ)$, the basis $BG$ spans the same set of vectors over $\bZ$ as the basis $B$. We use $\langle x, y \rangle$ for the standard Euclidean inner product of vectors $x,y \in \bR^m$, and $\nrm{x}$ for corresponding norm.  Reduced lattice bases obtained using the Lenstra-Lenstra-Lovasz (LLL)\cite{NS:2005,LLL:1982} or Hermite-Korkine-Zolotaroff (HKZ) reduction algorithms allow us to ensure that the sizes of the above fundamental domains are essentially independent of an initial choice of basis, depending on only on the lattice determinant. 

Elements of $\z_K$ correspond to $2d$ dimensional real vectors via map 
\[
\boldsymbol{\sigma} : \z_K \rightarrow \r^{2d}
\]
\[ 
z \mapsto \at{\re\at{\s_{1,+}\at{z}},\im\at{\s_{1,+}\at{z}},\ldots,\re\at{\s_{1,+}\at{z}},\im\at{\s_{1,+}\at{z}}}.
\]
The image $\CL$ of $\z_K$ under $\boldsymbol{\sigma}$ is a $2d$ dimensional integer lattice with associated bilinear form given by $\tr_{K/\q}\at{xy^{\ast}}$. Each integral basis of $\z_K$ corresponds to the basis of $\CL$. Similarly each $\z_K$ ideal has $\z$ basis. The images of $\z_K$ ideals under map $\boldsymbol{\sigma}$ correspond to a sublattices of $\CL$. Determiant of $\CL$ is equal to the discriminant of $\z_F$. 

%TODO: quadratic form and gram matrix ? 

% \clearpage

\section{Approximation algorithm}

%auto-ignore
%!TEX root = quaf.tex

\subsection{High level description of the algorithm} \label{sec:high:level:desc}

% \begin{itemize}
% \item Discussion of the idea of the algoritm
% \item Correctness and runtime of APPROXIMATE ( Fig.~\ref{fig:proc-approximate}) procedure
% \item Correctness and runtime of TARGET-SIZE ( Fig.~\ref{fig:proc-target-size}) procedure
% \item Discussion of the choice of input parameters given required precision $\ve$. Maybe algorithm for this
% \item Posprocessing to produce multiple circuits ( if we hope to get better cost )
% \end{itemize}

\begin{figure}[!t]

\protect\caption{\label{fig:proc-approximate}High level description of approximation
algorithm}

\rule[0.5ex]{1\columnwidth}{1pt}

\begin{centering}
%auto-ignore
%!TEX root = quaf.tex

\begin{algorithmic}[1]
\fxinput $ F,\s,a,b,S = \set{\p_1,\ldots,\p_M} , P$
\Statex $F$ is a totally real number field of degree $d$, $a,b$ are totally negative elements of $\z_F$
\Statex $\p_1,\ldots,\p_M$ are prime ideals of $\z_F$, $P$ is a natural number 
\inp $\rg{L_1}{L_M},\vp,\ve $
\Statex $\rg{L_1}{L_M}$ are natural numbers,
\Statex $\varphi$ is a real number that defines $R_z\at{\vp}$ to be approximated,
\Statex real number $\ve$ is the required approximation quality
\Procedure{APPROXIMATE}{}
\offline
\State $C_{\min},C_{\max} \gets$ \Call{RANDOM-INTEGER-POINT}{$\z_K$} \Comment See Fig.~\ref{fig:proc-find-point}
\State \Call{IS-EASILY-SOLVABLE}{$F$,$\z_F$,$P$} \Comment See Fig.~\ref{fig:proc-is-easily-solvable}
\State \Return $C_{\min},C_{\max}$
\spec $C_{\min},C_{\max}$ are real numbers; contstants defining the quality of the approximation
\online
\State \assert $\prg{L_1 \log N\at{\p_1}}{L_M\log N\at{\p_M}} - 4\log(1/\ve) \in [C_{\min},C_{\max}]$ \label{ln:appr:assert-eps}
\State NORM-EXISTS$, r \gets$ \Call{SUITABLE-Q-NORM}{$\rg{L_1}{L_M},\ve$} \Comment See Fig.~\ref{fig:proc-suitable-q-norm}  \label{ln:appr:suitable-norm}
% \If{ $\Not$  NORM-EXISTS }
% \State \Return 0 \Comment There is no quaternion $q$ with $\nrd\at{q}\zf = \trg{\p_1^{L_1}}{\p_M^{L_M}}$
% \EndIf
\State Solution-found $\gets$ FALSE
\While{\Not Solution-found }
\State $z_1 \gets $ \Call{RANDOM-INTEGER-POINT}{$\vp,\ve,r$} \label{ln:appr:find-pt} \Comment See Fig.~\ref{fig:proc-find-point}
\State $ e \gets \at{r - z_1 z_1 ^{\ast}} / \at{-b} $ \label{ln:appr:e}
\If{ $e \in \z_F$ and \Call{IS-EASILY-SOLVABLE}{$e$} } \Comment See Fig.~\ref{fig:proc-is-easily-solvable} \label{ln:appr:is-easy}
\State	Solution-found, $z_2 \gets $ \Call{FAST-SOLVE-NORM-EQ}{$e$} \Comment See Fig.~\ref{fig:proc-fast-solve} \label{ln:appr:norm-eq}
\EndIf
\EndWhile
\State \Return $q = e_1^{-1}\at{z_1}+e_2^{-1}\at{z_2}$ \Comment $\nrd\at{q}=r$  \label{ln:appr:q}
\EndProcedure
\out Quaternion $q$ from generalized Lipshitz order, such that 
\Statex $\rho\at{U_q,R_z\at{\vp}} \le \ve$, $\nrd\at{q}\z_F=\trg{\p_1^{L_1}}{\p_M^{L_M}}$

\end{algorithmic}
\par\end{centering}

\rule[0.5ex]{1\columnwidth}{1pt}

\end{figure}

In this section we first give a formal description of the core algorithm of this paper (solving Problem~\ref{prob:qap} and implementing Step~2 on Fig.~\ref{fig:online-flow}) and give a high level description of ideas behind it. This section is organized as following: each procedure presented in the section is accompanied with the Theorem or Proposition that proves its correctness and bounds on the runtime. In the proofs of the theorems we refer to variables defined in the pseudo-code of the corresponding procedures.

The goal of the rest of the paper is to prove the following theorem:
\begin{thm} \label{thm:main}There exist constants $C_{\min},C_{\max}$ and an algorithm (the online part of the procedure APPROXIMATE, Fig.~\ref{fig:proc-approximate}) that given
\begin{itemize}
\item quaternion gate set specification (see Definition~\ref{dfn:quaternion-spec}),
\item real numbers $\vp$ and $\ve \in (0,1/2)$,
\item cost vector $\at{L_1,\ldots,L_M}$,
\end{itemize}
such that
\begin{itemize}
\item narrow class of $\trg{\p_1^{L_1}}{\p_M^{L_M}}$ is trivial and
\item $L_1 \log\at{N\at{\p_1}} + \ldots +  L_M \log\at{N\at{\p_M}} - 4\log\at{1/\ve} \in [C_{\min},C_{\max}]$
\end{itemize}
finds quaternion $q$ from the generalized Lipschitz order such that $\nrd\at{q}\z_F=\trg{\p_1^{L_1}}{\p_M^{L_M}}$ and $d_2\at{U_q,R_z\at{\vp}}\le \ve $. Constants $C_{\min},C_{\max}$ depend only on the quaternion gate set specification and can be computed in advance~(by the offline part of procedure APPROXIMATE, Fig.~\ref{fig:proc-approximate}). The runtime of the algorithm is on average polynomial in $\log\at{1/\ve}$ under the Conjecture \ref{cnj:main}. 
\end{thm}

Before looking at details of the algorithm on Fig.~\ref{fig:proc-approximate} let us discuss the convention we use for all pseudo-code shown in the paper. All procedures have offline and online part. The offline part of all procedures have to be executed only once for given quaternion gate set specification. Its input we denote by words {\bf Fixed input}. For example, for procedure APPROXIMATE the fixed input consists of the most part of quaternion gate set specification~(see Definition~\ref{dfn:quaternion-spec}) and technical parameter $P$ related to the method for solving relative norm equation. The output of the offline part of the procedure we denote by term {\bf Specification}. For example, in the case of procedure APPROXIMATE, the output of the offline part is the additive constant defining the quality of approximation.

The online part of all procedures is executed for each instance of the approximation problem we are solving. The instance of the problem is defined by angle $\vp$, target precision $\ve$ and target cost vector $\at{\rg{L_1}{L_M}}$. These are precisely the inputs for the online part of procedure APPROXIMATE. The input to the online part of each procedure we denote by word {\bf Input}. An online part of each procedure uses results of computations done in offline part. Naturally, any offline part can not depend on the results of online computation. The output of online part of each procedure is denoted by word {\bf Output}. In our complexity analysis we are mainly concerned with the online part and show that the online part of procedure APPROXIMATE has polynomial runtime on average under a certain number theoretic conjecture. In Section~\ref{sec:examples} we provide the runtime of both online and offline parts for some examples and demonstrate that offline part is not prohibitively expensive for instances of the problem interesting for applications. Let us now discuss the online part of procedure APPROXIMATE in more details.

The algorithm shown on Fig.~\ref{fig:proc-approximate} finds a quaternion $q$ from the generalized Lipschitz order that has two following properties:
\begin{enumerate}
\item $d_2\at{U_q,R_z\at{\vp}}\le \ve,$
\item $\nrd\at{q}\z_F=\trg{\p_1^{L_1}}{\p_M^{L_M}}.$
\end{enumerate}
As we discussed in Section~\ref{sec:lip-order} each quaternion $q$ from the generalized Lipschitz order can be represented using two elements $z_1,z_2$ of the ring of integers $\z_K$ of the CM-field $K$ in the following way:
\[
	q = e_1^{-1}\at{z_1} + e_2^{-1}\at{z_2}.
\]
As discussed in Section~\ref{sec:distance}, the distance $d_2\at{U_q,R_z\at{\vp}}$ depends only on $z_1$. For this reason, in our algorithm we first pick $z_1$~(procedure RANDOM-INTEGER-POINT, line~\ref{ln:appr:find-pt} in Fig.~\ref{fig:proc-approximate}) such that condition (1) above is satisfied and then find $z_2$~(procedure FAST-SOLVE-NORM-EQ, line~\ref{ln:appr:norm-eq} in Fig.~\ref{fig:proc-approximate}) such that condition (2) is also satisfied.

Procedure FAST-SOLVE-NORM-EQ solves the relative norm equation $N_{K/F}\at{z_2} = z_2 z_2^{\ast} = e$ in the relative extension $K/F$
for a special class of right hand sides $e$. There are two challenges related to this procedure. First, the solution does not always exist for arbitrary right hand side $e$. Second, solving the arbitrary instance of the norm equation (for fixed extension $K/F$) can be as hard as factoring. We address both these challenges. First, we identify the necessary conditions on the right hand side $e$ of the equation to be solvable. Second, we identify the set of right hand sides $e$ for which the equation can be solved in probabilistic polynomial time~(using procedure IS-EASILY-SOLVABLE, line~\ref{ln:appr:is-easy}, in Fig.~\ref{fig:proc-approximate}) and attempt to solve the equation only for such right hand sides. The claim that our algorithm works in polynomial time is conditional on the conjecture discussed in Section~\ref{sec:conjecture} that the procedure IS-EASILY-SOLVABLE returns true with probability $O\at{1/\log\at{1/\ve}}$. For more details on this see Section~\ref{sec:relativeNormEqs}. We adopt ideas used to solve principal ideal problem in~\cite{KV} and show how to use LLL algorithm to solve the relative norm equation efficiently.

The necessary condition for the norm equation to be solvable is that for all embeddings $\s_k$ of $F$ into $\r$ it must be the case that $\s_k\at{e} > 0$ (for $k=1,\ldots,d$). Procedure RANDOM-INTEGER-POINT, line~\ref{ln:appr:find-pt} in Fig.~\ref{fig:proc-approximate} returns $z_1$ such that these conditions are satisfied together with condition $d_2\at{U_q,R_z\at{\vp}}\le \ve $. Recall, that we can associate an integer lattice $\CL$ with $\z_K$~(see Section~\ref{sec:lattice}). Essentially, procedure RANDOM-INTEGER-POINT samples lattice points from the following convex subset of $\r^{2d}$:
\begin{equation} \label{eq:target}
S_{r,\vp,\ve} = \set{ x \in \r^{2d} :  \re \at{\at{x_1 + i x_2 -z^{r,\vp,\ve}_0}e^{-i\vp/2} } \ge 0, \abss{x_{2k} + i x_{2k+1} } \le \s_k\at{r} }
\end{equation}
where $z^{r,\vp,\ve}_0 = \sqrt{\s_1\at{r}}\at{1-\ve^2}e^{-i\vp/2}$. See Fig.~\ref{fig:inequality} for the visualization of the projection of the set above. The geometry of the set $S_{r,\vp,\ve}$ is determined by the reduced norm $r$ of the quaternion output by procedure APPROXIMATE. We have a freedom in choosing $r$ up to a unit of $\z_F$. We use it to find $r$ such that we can sample lattice points from $S_{r,\phi,\ve}$ in time that is logarithmic in the volume of $S_{r,\vp,\ve}$. Such $r$ is found by procedure SUITABLE-Q-NORM~(line \ref{ln:appr:suitable-norm} in Fig.~\ref{fig:proc-approximate}).

Before proceeding to the proof of Theorem~\ref{thm:main} we state results that we will prove later in the paper and are using in our proof.

\begin{restatable}{thm}{findpoint}
\label{thm:find-point} There exist real numbers $p_0,M$ and vectors $R^{\min},R^{\max}$ from $\at{0,\infty}^d$~(computed in the offline part  of RANDOM-INTEGER-POINT procedure, Fig.~\ref{fig:proc-find-point}) such that for any real number $\vp$, real number $\ve \in (0,1/2)$ and totally positive $r$ from $\z_F$ such that:
\[
  \sqrt{\s_1\at{r}}\ve^2/4 \in [R_1^{\min},R_1^{\max}],\, \sqrt{\s_k\at{r}} \in [R_k^{\min},R_k^{\max}]\text{ for }k=2,\ldots,d
\]
there is an element $z$ from $\z_F$ such that $\vecs\at{z}$ is in $S_{r,\vp,\ve}$ (see Equation~(\ref{eq:target})). Procedure RANDOM-INTEGER-POINT runs in time polynomial in $log\at{1/\ve}$ and returns each element of the set  
\begin{equation} \label{eq:cand} \Cand_{r,\vp,\ve} = \set{ z \in \z_K : \vecs\at{z} \in S_{r,\vp,\ve} }
\end{equation}
with probability at least $p_0 / \abs{\Cand_{r,\vp,\ve}} $. The size of the set  $\Cand_{r,\vp,\ve}$ belongs to the interval
\[ 
\of{ 2\sqrt{4-\ve^2}/\ve, (4\sqrt{4-4\ve^2}/\ve + 2) M }
\]
\end{restatable}

The proof of Theorem~\ref{thm:find-point} and can be found in Section~\ref{sec:sampling} on Page~\pageref{sec:sampling}.

\begin{restatable}{thm}{suitablenorm}
\label{thm:suitable-norm} Given vector $R^{\min}$ from $\at{0,\infty}^d$, there exists constants $C_{\min},C_{\max}$ and vector $R^{\max}$~(computed by the offline part of SUITABLE-Q-NORM procedure, Fig.~\ref{fig:proc-suitable-q-norm}) such that for all non-negative integers $\rg{L_1}{L_M}$ and a positive real number $\ve \in (0,1)$ satisfying:
\[
 \prg{L_1 \log N\at{\p_1}}{L_M\log N\at{\p_M}} - 4\log(1/\ve) \in [C_{\min},C_{\max}]
\]
there is an algorithm that decides if narrow class number of ideal $\trg{\p_1^{L_1}}{\p_M^{L_M}}$ is trivial. If this is the case the algorithm outputs totally positive element $r$ of $\z_F$ such that $r \z_F = \trg{\p_1^{L_1}}{\p_M^{L_M}} $ and
\[
  \sqrt{\s_1\at{r}}\ve^2/4 \in [R_1^{\min},R_1^{\max}],\, \sqrt{\s_k\at{r}} \in [R_k^{\min},R_k^{\max}]\text{ for }k=2,\ldots,d
\]
The algorithm runs in polynomial time in $\log\at{1/\ve}$~(see online part of SUITABLE-Q-NORM procedure in Fig.~\ref{fig:proc-suitable-q-norm}). 
\end{restatable}

The informal discussion and the proof of the Theorem~\ref{thm:suitable-norm} can be found in Section~\ref{sec:snorm} on Page~\pageref{sec:snorm}.

\begin{restatable}{thm}{normeq}
\label{thm:norm-eq} Given totally positive element $e$ of $F$ there exists an algorithm for testing if the instance of integral relative norm equation in $K/F$
\[
 zz^{\ast} = e , z \in \z_K
\]
can be solved in polynomial time in $\log T_2\at{e}$ (procedure IS-EASILY-SOLVABLE, where $T_2\at{e} = \sum_{k=1}^d \s^2_k \at{e}$). If the test is passed, there exist another algorithm for deciding if the solution exists and finding it that runs in time polynomial in $\log T_2\at{e}$~(procedure FAST-SOLVE-NORM-EQ). Procedure IS-EASILY-SOLVABLE returns true for at least those cases when the ideal $e\z_F$ is prime.
\end{restatable}

\begin{proof}[Proof of Theorem~\ref{thm:main}] \label{thm:main:proof}
Let us first prove that output is correct. The norm of the quaternion $q$ we compute on line \ref{ln:appr:q}~(Fig.~\ref{fig:proc-approximate}) is equal to totally positive element $r$ of $\z_F$ computed on line \ref{ln:appr:suitable-norm}~(Fig.~\ref{fig:proc-approximate}). This is because on line~\ref{ln:appr:norm-eq}~(Fig.~\ref{fig:proc-approximate}) we find $z_2$ such that $r =\nrd\at{q} = \abss{z_1} - b \abss{z_2}$.  By Theorem~\ref{thm:suitable-norm} the output $r$ of the procedure SUITABLE-Q-NORM~(line \ref{ln:appr:suitable-norm}, Fig.~\ref{fig:proc-approximate}) satisfies $r\z_F = \nrd\at{q}\z_F = \trg{\p_1^{L_1}}{\p_M^{L_M}}$.
According to Theorem~\ref{thm:find-point} procedure RANDOM-INTEGER-POINT (line~\ref{ln:appr:find-pt}, Fig.~\ref{fig:proc-approximate}) returns an element $z_1$ of $\z_K$ such that
 \[
 \re \at{\at{\s_{1,+}\at{z}-z_0}e^{-i\vp/2} } \ge 0, \abs{ \s_{1,+}\at{z} } \le \sqrt{\s_1\at{r}}
\]
where $z_0 = \sqrt{\s_1\at{r}}\at{1-\ve^2}e^{-i\vp/2}$. According to Proposition~\ref{prop:distance} this implies that $d_2\at{R_z\at{\vp},U_q} \le \ve$. Numbers $z_1, z_2 $ are in $\z_F$ which immediately implies that $q$ is in generalized Lipschitz order (see Section~\ref{sec:lip-order}).

% Procedure FAST-SOLVE-NORM-EQ~(line \ref{ln:appr:norm-eq}, Fig.~\ref{fig:proc-approximate}) finds solution $z_2$ to the relative norm equation such that
% \[
% \nrd\at{q}= r =z_1 z^{\ast}_1 - b z_2 z^{\ast}_2
% \]
% Therefore, inequality  $\abs{\s_{1,+}\at{z}-v_1} \le R_1$ implies that $\s_{1,+}\at{z}$ belongs to the set
% \[
%  \set{ z \in \c : \re\at{\at{z-z_0}e^{i\vp/2}}\ge 0, \abs{z} \le R }.
% \]

%
Next, let us show that the restrictions on inputs of all the procedures called within online part of procedure APPROXIMATE are satisfied. Procedure SUITABLE-Q-NORM~(line \ref{ln:appr:suitable-norm}, Fig.~\ref{fig:proc-approximate}) always succeeds if the narrow class group of $\trg{\p_1^{L_1}}{\p_M^{L_M}}$ is trivial and
\[
 \prg{L_1 \log N\at{\p_1}}{L_M\log N\at{\p_M}} -4\log(1/\ve) \in [C_{\min},C_{\max}],
\]
which is required in the statement of the theorem. From Theorem~\ref{thm:suitable-norm} we know that procedure SUITABLE-Q-NORM~(line \ref{ln:appr:suitable-norm}, Fig.~\ref{fig:proc-approximate}) finds $r$ such that: 
\[
  \sqrt{\s_1\at{r}}\ve^2/4 \in [R_1^{\min},R_1^{\max}],\, \sqrt{\s_k\at{r}} \in [R_k^{\min},R_k^{\max}]\text{ for }k = 2,\ldots,d
\]
therefore procedure RANDOM-INTEGER-POINT always succeeds. Let us show that $e=(r-z_1\, z^*_1)/(-b)$ is totally positive, where $b$ is the parameter from the definition of the quaternion algebra $Q$. Note that $e$ being totally positive is required by procedures IS-EASILY-SOLVABLE~(line \ref{ln:appr:is-easy}, Fig.~\ref{fig:proc-approximate}) and FAST-SOLVE-NORM-EQUATION~(line \ref{ln:appr:norm-eq}, Fig.~\ref{fig:proc-approximate}). For all $k=\rg{1}{d}$ it the case that
\[
\s_k\at{e}=\at{\s_k\at{r} - \abss{\s_{k,+}\at{z_1}}}/{-\s_k\at{b}}.
\]
By definition, $b$ is totally negative and $\s_k\at{e} >0 $ if and only if $\s_k\at{r} - \abss{\s_{k,+}\at{z_1}} > 0$ . By Theorem~\ref{thm:find-point} the output of procedure RANDOM-INTEGER-POINT~(line \ref{ln:appr:find-pt}, Fig.~\ref{fig:proc-approximate}) satisfies $\abs{\s_{k,+}\at{z_1}} \le R^{\min}_k$ for $k=\rg{2}{d}$. By Theorem~\ref{thm:suitable-norm} totally positive element $r$ of $\z_F$ satisfies $R^{\min}_k \le \sqrt{\s_k\at{r}}$ for $k=\rg{2}{d}$ and therefore $\s_k\at{e} >0 $ for $k=\rg{2}{d}$. We have already shown above that $\abs{\s_{k,+}\at{z_1}} \le \sqrt{\s_1\at{r}}$ which implies $\s_1\at{e} >0 $.

It remains to show that our algorithm terminates and runs on average in time polynomial in $\log\at{1/\ve}$. Procedure  SUITABLE-Q-NORM~(line \ref{ln:appr:suitable-norm}, Fig.~\ref{fig:proc-approximate}) runs in time polynomial in $\log\at{1/\ve}$ by Theorem~\ref{thm:suitable-norm}. Let us show that all procedures inside the loop run in polynomial time. Procedure RANDOM-INTEGER-POINT runs in polynomial time in $\log\at{1/\ve}$ according to Theorem~\ref{thm:find-point}.

 % To estimate the runtime of FIND-POINT procedure it is sufficient to bound the $\nrm v$. Indeed $\nrm{v} = \abs{v_1}$ and according to Theorem~\ref{thm:random-region} we have $\abs{v_1} \le \sqrt{\s_1\at{r}}$. We conclude that $\log\nrm v$ is also bounded by the function linear in $L_1,\ldots,L_M$ and therefore by Theorem~\ref{thm:find-point} procedure FIND-POINT also runs in time polynomial in $\rg{L_1}{L_M}$.
Next we show that the logarithm of
 \[
 T_2\at{e}= \sum_{k=1}^d \s^2_k\at{e}
 \]
is bounded by polynomial in $\log\at{1/\ve}$. This implies that procedures IS-EASILY-SOLVABLE and FAST-SOLVE-NORM-EQ run on average in polynomial time according to Theorem~\ref{thm:norm-eq}. Indeed we have $\s_k\at{e} \le \s_k\at{r}/\s_k\at{-b}$, $\sqrt{\s_1\at{r}}$ is bounded by $4 R^{\max}_1/\ve^2$ and $\sqrt{\s_k\at{r}}$ are bounded by $R^{\max}_k$ for $k=\rg{2}{d}$. 
Finally, arithmetic in the number field~(line~\ref{ln:appr:e}, Fig.~\ref{fig:proc-approximate}) can be done in time polynomial in $\log T_2\at{z_1}$ and $\log T_2\at{r}$\cite{BF:2012}. We bound $\log T_2\at{z_1}$ using inequalities $\abss{\s_{k+1}\at{z_1}} \le \s_k\at{r}$. We also perform the test $e \in \z_F$~(line~\ref{ln:appr:is-easy}, Fig.~\ref{fig:proc-approximate}). It can be performed in polynomial time in $\log T_2\at{r}$, because it can be reduced to multiplying a vector over $\q$ with the norm bounded by $C \log T_2\at{r}$ by fixed  matrix over~$\q$.

We conclude that procedure APPROXIMATE runs on average in time polynomial in $\log\at{1/\ve}$ under the conjecture that the fraction of points in set $\set{ z \in \z_K : \vecs\at{z} \in S_{r,\vp,\ve} }$ for which we can reach line \ref{ln:appr:norm-eq} (Fig.~\ref{fig:proc-approximate}) and successfully find $z_2$ scales as $\Omega\at{\log\at{1/\ve}}$. This is because according to Theorem~\ref{thm:find-point} we sample all points from $\set{ z \in \z_K : \vecs\at{z} \in S_{r,\vp,\ve} }$ sufficiently uniformly. 
\end{proof}

\subsection{Picking a suitable quaternion norm} \label{sec:snorm}

In this subsection we are going to prove the following theorem:

\suitablenorm*

\begin{figure}[!ht]
\protect\caption{\label{fig:proc-suitable-q-norm}SUITABLE-Q-NORM procedure. }
\rule[0.5ex]{1\columnwidth}{1pt}

\begin{centering}
%auto-ignore
%!TEX root = quaf.tex

\begin{algorithmic}[1]
\fxinput $\rg{\p_1}{\p_M},R^{\min}$
\Statex[1] $\p_1,\ldots,\p_M$ are prime ideals of $\z_F$
\Statex[1] $R^{\min}$ is a vector from $\at{0,\infty}^d$
\inp $\rg{L_1}{L_M}, \ve$
\Statex[1] $\rg{L_1}{L_M}$ non-negative integers defining ideal $\trg{\p_1^{L_1}}{\p_M^{L_M}}$
\Statex[1] $\ve$ is a real number
\Procedure{SUITABLE-Q-NORM}{}
\offline
\State $\rg{\delta_0}{\delta_d} \gets$ \Call{UNIT-ADJUST}{$F$}
\State \Call{TOTALLY-POS-GENERATOR}{$\rg{\p_1}{\p_M}$} 
\State $C_{\min},C_{\max},R^{\max} \gets $ \Call{TARGET-SIZE}{$\rg{\p_1}{\p_M},R^{\min},\rg{\delta_0}{\delta_d}$} \Comment See Fig.~\ref{fig:proc-target-size}
\State \Return $C_{\min},C_{\max},R^{\max}$
\spec  $C_{\min},C_{\max}$ are real numbers, $R^{\max}$ is the vector in $\at{0,\infty}^d$
\online
\State HAS-TP-GEN,$r \gets$ \Call{TOTALLY-POS-GENERATOR}{$\rg{L_1}{L_M}$}\label{ln:sqn:r}
\If{$\Not$ HAS-TP-GEN }
\State \Return FALSE, 0 \Comment Narrow class group of $\trg{\p_1^{L_1}}{\p_M^{L_M}}$ is non-trivial \label{ln:sqn:false}
\EndIf
\State $\rg{t_1}{t_d} \gets $ \Call{TARGET-SIZE}{$r,\rg{L_1}{L_M},\ve$}  \Comment See Fig.~\ref{fig:proc-target-size}
\State $u \gets$ \Call{UNIT-ADJUST}{$\rg{t_1}{t_d}$} \label{ln:appr:unit-adj} \label{ln:sqn:u} \Comment See Fig.~\ref{fig:proc-unit-adjust}
\State \Return TRUE, $ru^{2}$ \label{ln:appr:ru}
\EndProcedure
\out FALSE,0 if ideal $\trg{\p_1^{L_1}}{\p_M^{L_M}}$ does not have a totally positive generator, or
\Statex[2] TRUE,$r$ otherwise. Totally positive algebraic integer $r$ from $\z_F$ is such that
\Statex (a) $r \z_F = \trg{\p_1^{L_1}}{\p_M^{L_M}} $, 
\Statex (b) $\ve^2 \sqrt{\s_1\at{r}} /4 \in  [R^{\min}_1,R^{\max}_1]$ and $\sqrt{\s_k\at{r}} \in [R^{\min}_k,R^{\max}_k]$ for $k=\rg{2}{d}$.
\end{algorithmic}
\par\end{centering}

\rule[0.5ex]{1\columnwidth}{1pt}

\end{figure}

The proof relies on the following proposition proven in this and the next sections.

\begin{restatable}{prop}{totallypos}
\label{prop:totally-pos} Given non-negative integers $L_1,\ldots,L_M$ there is an algorithm~(procedure TOTALLY-POS-GEN, Fig.~\ref{fig:proc-totally-pos-generator}) that decides if ideal $\trg{\p_1^{L_1}}{\p_M^{L_M}}$ has a totally positive generator. The algorithm  also finds a totally positive generator $r$ of the ideal if it exists. The algorithm runs in time polynomial in $L_1,\ldots,L_M$ and $\log T_2\at{r}$ is bounded by the function that is linear in $L_1,\ldots,L_M$.
\end{restatable}

\begin{restatable}{prop}{unitadjust}
\label{prop:unit-adjust} There exists real numbers $\delta_0$ and $\rg{\delta_1}{\delta_d}$~(computed by the offline part of the procedure UNIT-ADJUST, Fig.~\ref{fig:proc-unit-adjust}) such that there exists an algorithm~(online part of the procedure UNIT-ADJUST, Fig.~\ref{fig:proc-unit-adjust}) that for any real numbers $\rg{t_1}{t_d}$ finds a unit $u$ of $\z_F$ such that the following inequalities hold
\[
 \abs{ \log\abs{\s_k\at{u}} - t_k} \le \log\delta_k \text{ for } k=\rg{1}{d},
\]
under the assumption that $\abs{\prg{t_1}{t_d}} < \log\delta_0 $. The runtime of the algorithm is bounded  by a polynomial in $\nrm{t}$.
\end{restatable}

\begin{restatable}{prop}{targetsize} \label{prop:target-size} Given real numbers $\rg{\delta_0}{\delta_d} > 1$, vector $R^{\min}$ from $\at{0,\infty}^d$, and prime ideals $\rg{\p_1}{\p_M}$ there exist real numbers $C_{\min},C_{\max}$ and vector $R^{\max}$ from $\at{0,\infty}^d$~(computed by the offline part of the procedure TARGET-SIZE, Fig.~\ref{fig:proc-target-size}) such that there exist an algorithm~(online part of the procedure TARGET-SIZE, Fig.~\ref{fig:proc-target-size}) that given non-negative integers $\rg{L_1}{L_M}$, real number $\ve$ and totally positive element $r$ of $F$ finds real numbers $\rg{t_1}{t_d}$ such that 
\begin{eqnarray*}
 t_1 + \log\at{\sqrt{\s_k\at{r}}\ve^2/4} & \in &  \of{ \log R^{\min}_1 + \log \delta_1, \log R^{\max}_1 - \log \delta_1   } \\
 t_k + \log\sqrt{\s_k\at{r}} & \in &  \of{ \log R^{\min}_k + \log \delta_k, \log R^{\max}_k - \log \delta_k   }, k=\rg{2}{d}
\end{eqnarray*}
and $\abs{\prg{t_1}{t_d}} < \log\delta_0 $. The algorithm succeeds under the assumption that $L_1 \log N\at{\p_1} + \ldots +  L_M \log N\at{\p_M} - 4\log\at{1/\ve} \in [C_{\min},C_{\max}] $.

The runtime of the algorithm is bounded by a polynomial in $\log\at{1/\ve}$ and $\log T_2\at{r}$. The norm $\nrm{t}$ is bounded by the function that is linear in the same variables.

\end{restatable}

\begin{proof}[Proof of Theorem~\ref{thm:suitable-norm}]
First we prove that the procedure terminates in polynomial time in $L_1,\ldots,L_M$ when it returns FALSE. Indeed, when ideal $\trg{\p_1^{L_1}}{\p_M^{L_M}}$ does not have a totally positive generator procedure TOTALLY-POS-GENERATOR returns FALSE and procedure SUITABLE-Q-NORM terminates~(line~\ref{ln:sqn:false}, Fig.~\ref{fig:proc-suitable-q-norm}). According to Proposition~\ref{prop:totally-pos} procedure TOTALLY-POS-GENERATOR  runs    in polynomial time in $L_1,\ldots,L_M$.

Next we consider the case when the output of SUITABLE-Q-NORM procedure is TRUE. Let us first prove that the output of SUITABLE-Q-NORM procedure~(Fig.~\ref{fig:proc-suitable-q-norm}) is correct in this case. By Proposition~\ref{prop:totally-pos} algebraic integer $r$~(line~\ref{ln:sqn:r}, Fig.~\ref{fig:proc-suitable-q-norm}) is totally positive and $r\z_F=\trg{\p_1^{L_1}}{\p_M^{L_M}}$. By Proposition~\ref{prop:unit-adjust} algebraic integer $u$~(line~\ref{ln:sqn:u}, Fig.~\ref{fig:proc-suitable-q-norm}) is a unit. Therefore $ru^{2}$ is also totally positive and also generates ideal $\trg{\p_1^{L_1}}{\p_M^{L_M}}$.

It remains to show that $\s_k\at{r u^2}$ satisfy required inequalities. By Proposition~\ref{prop:unit-adjust} unit $u$~(line~\ref{ln:sqn:u}, Fig.~\ref{fig:proc-suitable-q-norm}) computed by procedure UNIT-ADJUST satisfies the following inequalities:
\[
\abs{ \log\abs{\s_k\at{u} }  - t_k } \le  \log\delta_k, k=\rg{1}{d},
\]
because by Proposition~\ref{prop:target-size} procedure TARGET-SIZE ensures that $\abs{\prg{t_1}{t_d}} < \log\delta_0 $. Now we see that
\begin{eqnarray*}
 \log{\sqrt{\s_k\at{ru^2}}} -  \log{\sqrt{\s_k\at{r}}} & = &  \log\abs{\s_k\at{u}} \in [t_k - \log \delta_k , t_k + \log \delta_k ] \\
  \log\at{{\sqrt{\s_1\at{ru^2}}}\ve^2/4} -  \log\at{\sqrt{\s_1\at{r}}\ve^2/4} & = &  \log\abs{\s_1\at{u}} \in [t_1 - \log \delta_1 , t_1 + \log \delta_1 ]
\end{eqnarray*}
This immediately implies that 
\begin{eqnarray*}
\log{\sqrt{\s_k\at{ru^2}}} & \in &  [\log R^{\min}_k,\log R^{\max}_k], \text{ for } k = 2,\ldots,d \\
\log\at{ {\sqrt{\s_1\at{ru^2}}}\ve^2/4} & \in &  [\log R^{\min}_1,\log R^{\max}_1]
\end{eqnarray*}
We now show that the runtime of our algorithm is bounded by a polynomial in $\log\at{1/\ve}$. All $L_k$ are bounded by function linear in $\log\at{1/\ve}$.  Procedure TOTALLY-POS-GENERATOR runs in polynomial time and produces $r$ such that $\log T_2\at{r}$ is bounded by a function linear in $L_1,\ldots,L_M$. This ensures that procedure TARGET-SIZE outputs $t_1,\ldots,t_d$ such that their bit size is bounded by polynomial in $\log\at{1/\ve}$. It also ensures that $\nrm{t}$ is bounded by a function linear in $\log\at{1/\ve}$. This ensures that procedure UNIT-ADJUST runs in polynomial time. Note that for unit $u$~(computed in line~\ref{ln:sqn:u}, Fig.~\ref{fig:proc-suitable-q-norm}) we have
\[
 \log \sqrt{T_2\at{u}} = \log \sqrt{\sum_{k=1}^d \s^2_k\at{u} } \le \sqrt{d} \max_k \log\abs{\s_k} \le \sqrt{d}\max_k \at{t_k + \log \delta_k}  \le \sqrt{d}\at{  \nrm{t} + \max_k \log \delta_k.}
\]
Therefore $ \log T_2\at{u}$ is bounded by a function linear in $\log\at{1/\ve}$. Hence, the time spent on computing $ru^{2}$ is bounded by polynomial in $\log\at{1/\ve}$ as required \cite{BF:2012}. We have shown that procedure SUITABLE-Q-NORM runs in polynomial time.
\end{proof}

\begin{figure}[!th]

\protect\caption{\label{fig:proc-totally-pos-generator}TOTALLY-POS-GENERATOR procedure. }

\rule[0.5ex]{1\columnwidth}{1pt}

\begin{centering}
%auto-ignore
%!TEX root = quaf.tex

\begin{algorithmic}[1]
\fxinput Prime ideals $\rg{\p_1}{\p_M}$ of $\z_F$
\inp Non-negative integers $\rg{L_1}{L_M}$  defining ideal $\trg{\p_1^{L_1}}{\p_M^{L_M}}$
\Procedure{TOTALLY-POS-GENERATOR}{}
\offline
\State Find minimal $N_k$ such that ideal $\p^{N_k}_k$ has totally positive generator $r_k$
\ForAll{ $\rg{n_1}{n_k} \in \set{\rg{0}{N_1 - 1}} \times \dots \times \set{\rg{0}{N_M - 1}}   $ }
\If{ ideal $\trg{\p^{n_1}_1}{\p^{n_M}_M}$ has totally positive generator $r'$ }
\State $r\of{n_1,\ldots,n_M} \gets  r'$ 
\Else
\State $r\of{n_1,\ldots,n_M} \gets 0$ 
\EndIf
\EndFor
\spec  
\online
\State $l_k \gets L_k \mod N_k, s_k \gets (L_k - l_k) / N_k $ 
\If{ $r\of{l_1,\ldots,l_k} = 0 $ } \label{ln:tpg:check}
\State \Return FALSE, 0
\Else
\State \Return TRUE, $\trg{r^{s_1}_1}{r^{s_M}_M}r\of{l_1,\ldots,l_k}$
\EndIf
\EndProcedure
\out FALSE, 0 if ideal $\trg{\p_1^{L_1}}{\p_M^{L_M}}$ does not have totally positive generator, otherwise
\Statex[2] TRUE, $r$ where $r\z_F = \trg{\p_1^{L_1}}{\p_M^{L_M}}$ and $r$ is totally positive 
\end{algorithmic}
\par\end{centering}

\rule[0.5ex]{1\columnwidth}{1pt}

\end{figure}

Let us now discuss procedure TOTALLY-POS-GENERATOR~(Fig.~\ref{fig:proc-totally-pos-generator}) in more details. Finding totally positive generator of the ideal is strictly more difficult than finding a generator of the ideal. The latter problem is known to be hard and there is no polynomial time algorithm known for it. In our case, because the ideal we interested has special form $\p_1^{L_1}\ldots\p_m^{L_M}$ we can find a totally positive generator for it by precomputing certain information about $\p_1,\ldots,\p_m$~(in the offline part of TOTALLY-POS-GENERATOR procedure, Fig.~\ref{fig:proc-totally-pos-generator}). Let us now prove the following proposition:

\totallypos*

\begin{proof}
Let us first prove the correctness of the online part of the algorithm. In is not difficult to see that
\[
 \trg{\p_1^{L_1}}{\p_M^{L_M}} = \trg{\at{\p_1^{N_1}}^{s_1}}{\at{\p_M^{N_M}}^{s_M}} \trg{\p_1^{l_1}}{\p_M^{l_M}}
\]
We know what each of ideals $\p_k^{N_k}$ has totally positive generator. Therefore $\trg{\p_1^{L_1}}{\p_M^{L_M}}$ has totally positive generator if and only if ideal $\trg{\p_1^{l_1}}{\p_M^{l_M}}$ does. This is what is checked on line~\ref{ln:tpg:check} in Fig.~\ref{fig:proc-totally-pos-generator}. We have shown that the procedure always returns FALSE, when $\trg{\p_1^{L_1}}{\p_M^{L_M}}$ does not have totally positive generator, and TRUE when it does. In the case when the procedure returns TRUE we have $r\of{l_1,\ldots,l_k}\z_F =  \trg{\p_1^{l_1}}{\p_M^{l_M}} $ and therefore
\[
 \trg{\p_1^{L_1}}{\p_M^{L_M}}= \trg{r_1^{s_1}\z_F}{r_M^{s_M}\z_F} \cdot r\of{l_1,\ldots,l_k} \z_F
\]
We see that $r_1^{s_1}\ldots r_M^{s_M} r\of{l_1,\ldots,l_k}$ is totally positive element of $\z_F$ that generates ideal $\trg{\p_1^{L_1}}{\p_M^{L_M}}$.

Let us now show that algorithm runs in polynomial time. The number of multiplications we need to perform is bounded by $\prg{L_1}{L_M}$.  Note that
\[
 T_2\at{\trg{x_1}{x_n}} \le d \prod_{k=1}^{n} T_2\at{x_k}, \text{ for } x_k \in F.
\]
Therefore, each time we perform a multiplication, the value $\log T_2$ of the arguments is bounded by a function linear in $\rg{L_1}{L_M}$. We conclude that $r_1^{s_1}\ldots r_M^{s_M} r\of{l_1,\ldots,l_k}$ is computed in time polynomial in $L_1,\ldots,L_M$. The inequality above also implies that $\log T_2$ of the algorithm is bounded by a function linear in $L_1,\ldots,L_M$.

It remains to show the correctness of the offline part of procedure. First, note that $N_k$ always exist. The fact that the class group of the number field is always finite implies that for each ideal $\p_k$ there exists a number $N'_k$ (dividing the order of class group) such that ideal $\p^{N'_k}_k$ is principal and has generator $\xi_k$. Note that $\xi_k^2$ is totally positive generator of $\p^{2 N'_k}_k$. We have shown that $N_k \le 2N'_k$. Computing the order of the class group and testing if the ideal is principal are two standard problems solved in computational number theory. Checking that given ideal has a totally positive generator can also be done using standard methods. If totally positive generator $r$ exists then the ratio of the ideal generator $\xi$ and $r$ must be a unit. Therefore once $\xi$ is known it remains to find unit $u$ such that $\xi u$ is totally positive. This can be done if and only if $\xi\z_F$ has a totally positive generator. It is not difficult to check this algorithmically, once the system of fundamental units of $\z_F$ has been computed.
\end{proof}

Note, that in case of Clifford+T and Clifford+$R_z\at{\pi/16}$ gate set the narrow class group of $F$ is trivial and therefore any ideal has a totally positive generator. This significantly simplifies procedure TOTALLY-POS-GENERATOR.

\begin{figure}[!ht]

\protect\caption{\label{fig:proc-target-size}Procedure TARGET-SIZE}

\rule[0.5ex]{1\columnwidth}{1pt}

\begin{centering}
%auto-ignore
%!TEX root = quaf.tex

\begin{algorithmic}[1]
\fxinput $\rg{\p_1}{\p_M},\rg{R_1}{R_d},\rg{\delta_0}{\delta_{d}}$
\Statex[1] $\rg{\p_1}{\p_M}$ are prime ideals of $\z_F$
\Statex[1]$\rg{R_1}{R_d},\rg{\delta_0}{\delta_{d}}$ are real numbers 
\inp $\rg{L_1}{L_M},\ve,r$
\Statex[1] $\rg{L_1}{L_M}$ are integers 
\Statex[1] $\ve$ is a real number, $r$ is a totally positive element of $F$.
\Procedure{TARGET-SIZE}{}
\offline
\State $C_{\min} \gets 2\at{\sum_{k=1}^{d}\log R_{k}+\sum_{k=1}^{d}\log\delta_{k}+2\log2},\, C_{\max} \gets C_{\min} + 2\log{\delta_0}$
\State $R^{\max}_k \gets \exp\at{ \log R_k^{\min} + 2\log \delta_k + \log \delta_0 / d } $
\State \Return $C_{\min}, C_{\max}, R^{\max}$
\spec $C_{\min},C_{\max}$ -- real numbers, $R^{\max}$ vector from $\at{0,\infty}^d$
\online
\State \assert $\prg{L_1 \log N\at{\p_1}}{L_M\log N\at{\p_M}} - 4\log(1/\ve) \in [C_{\min},C_{\max}]$
\State $t_1 \gets \log R^{\min}_1 + \log \delta_1 - \log\at{\sqrt{\s_1\at{r}}\ve^2/4}$ \label{ln:ts:sone}
\State $t_k \gets \log R^{\min}_k + \log \delta_k - \log\sqrt{\s_k\at{r}}$, for $t=\rg{2}{d}$ \label{ln:ts:sk}
\State \Return $t_1,\ldots,t_d$ 
\EndProcedure
\out $\rg{t_1}{t_d}$ are real number such that $\abs{\prg{t_1}{t_d}}\le \log \delta_0$,  
\Statex[2] $t_1 + \log\at{\sqrt{\s_k\at{r}}\ve^2/4}   \in    \of{ \log R^{\min}_1 + \log \delta_1, \log R^{\max}_1 - \log \delta_1 }$  , 
\Statex[2] $t_k + \log\sqrt{\s_k\at{r}}  \in   \of{ \log R^{\min}_k + \log \delta_k, \log R^{\max}_k - \log \delta_k   }, k=\rg{2}{d}
$

\end{algorithmic}
\par\end{centering}

\rule[0.5ex]{1\columnwidth}{1pt}

\end{figure}

On a high level procedure TARGET-SIZE~(Fig.~\ref{fig:proc-target-size}) is about finding a solution to the system of linear inequalities for $t_1,\ldots,t_d$ mentioned in the statement of the Proposition~\ref{prop:target-size}. In the offline part we compute constant $C$ that guarantees that the system has a solution. In addition we show that the solution to such a system can be represented by numbers of moderate bit-size, can be found in polynomial time and that the norm of the solution vector is bounded. Now we proceed to the formal proof.

\targetsize*
\begin{proof}
Let us first prove correctness of the procedure TARGET-SIZE~(Fig.~\ref{fig:proc-target-size}). On lines~\ref{ln:ts:sone},\ref{ln:ts:sk} in Fig.~\ref{fig:proc-target-size} we assign to $t_k$ the smallest values for which the constraints on $t_k$ are satisfied.  Next we show that our choice of $C_{\min},C_{\max}$ and $R^{\max}$ ensures that all other constraints are also satisfied. Indeed, $R^{\max}_k \ge R^{\min}_k + 2 \log \delta_k$ therefore intervals for $t_k$ are non-empty. It remains to show that $\abs{\prg{t_1}{t_d}} \le \delta_0 $. This follows from the following equalities:  
\begin{eqnarray*}
\prg{t_1}{t_d} & = & \sum_{k=1}^d \log R^{\min}_k + \sum_{k=1}^d \log \delta_k  -  \log\at{ \sqrt{ N_{F/\q}\at{r}} \ve^2/4 } \\
 & = & C_{\min}/2 - \frac{1}{2}\sum_{k=1}^{M} L_k \log N\at{\p_k} + 2\log\at{1/\ve}
\end{eqnarray*}
Therefore value $\prg{t_1}{t_d}$ belongs to the interval $\of{\at{C_{\min}-C_{\max}}/2,0}$. The definition of $C_{\max}$ precisely implies that $\of{\at{C_{\min}-C_{\max}}/2,0} = \of{-\log \delta_0, 0 }$ which gives requiered bound on the sum of $t_k$. Note that the anaysis above performed for $t_k$ will hold for any $t'_k$ in the interval $[t_k,t_k + \log\delta_0/d ]$. 

Parameter $\log\delta_0$ is needed to account for finite precision arithmetic we are using. It is not difficult to see that as soon as precision of arithmetic used is smaller then $\log\delta_0 /  C_1 d$ for sufficiently big fixed constant $C_1$ numbers $t_k$ computed within mentioned precision will satisfy all required constraints. It is sufficient to perform the calculation up to fixed precision independent on the online part of the algorithm input. This implies that all calculations in the online part can be performed in polynomial time. We finally note that $t_k$ are bounded by functions linear in $\log T_2\at{r}$  and  $\log\at{1/\ve}$ and therefore the same is true for $\nrm{t}$. As we have established that $t_k$ can be computed up to fixed precision, bound on $\nrm{t}$ implies a bound on the number of bits needed to specify each $t_k$. This concludes our proof.
\end{proof} 

% \clearpage

%auto-ignore
%!TEX root = quaf.tex

\subsection{Solution region sampling} \label{sec:sampling}

\begin{figure}[ht!]

\protect\caption{\label{fig:proc-find-point} High level description of RANDOM-INTEGER-POINT procedure used in APPROXIMATE procedure in Fig.~\ref{fig:proc-approximate} }

\rule[0.5ex]{1\columnwidth}{1pt}

\begin{centering}
%auto-ignore
%!TEX root = quaf.tex 

\begin{algorithmic}[1]
\fxinput $\z_K$ -- ring of integers of a CM field of degree $2d$
\inp $\vp,\ve,r$ 
\Procedure{RANDOM-INTEGER-POINT}{}
\offline
\State $z_1,\ldots,z_{2d}$ is a fixed integral basis of $\z_F$
\State $B = [\boldsymbol\sigma\at{z_1},\ldots, \boldsymbol\sigma\at{z_{2d}}]$ is a basis of the lattice associated to $\z_K$,
\State $R^{\min}_k = \frac{1}{2}\sqrt{ \max_j \re\at{\s_{k,+}\at{z_j} }^2 + \max_j \im\at{\s_{k,+}\at{z_j} }^2}$ for $k=1,\ldots,d$
\State $C_{\min},C_{\max}, R^{\max} \gets $ \Call{SUITABLE-Q-NORM}{$\rg{\p_1}{\p_M},R^{\min}$} \Comment See Fig.~\ref{fig:proc-suitable-q-norm}
\State $\mathrm{SHIFTS}$ $\gets \set{z \in \z_K : \nrm{ P_1 \vecs\at{z} } \le 2\sqrt{5}R_1^{\max}, \nrm{ P_k \vecs\at{z} } \le   R_k^{\max} + R_k^{\min}, k=\rg{2}{d} }$\label{ln:rip:shifts}
\State \Comment $P_k : \r^{2d} \rightarrow \r^2 $ is a projector to coordinates $2k-1$ and $2k$   
\State $p_0 \gets 1 / \at{ 8\abs{\mathrm{SHIFTS}}}, M \gets \abs{\mathrm{SHIFTS}}$ \Comment $p_0$ 
indicates how close the distribution
\State \Comment of the procedure output to the uniform over set $\Cand_{r,\vp,\ve}$
\State \Return $C_{\min},C_{\max}$
\spec $C_{\min},C_{\max}$ -- real numbers 
\online
\State $R \gets \sqrt{\s_1\at{r}}, H \gets 2R\ve\sqrt{4-4\ve^2},\,N_{\max}=\rnd{H/\ve^2 R} $ \Comment See Fig.\ref{fig:sampling}
\State $z_c \gets R(1-3\ve^2/4)e^{-t\vp/2},\,\Delta z \gets i e^{-i\vp/2} \ve^2 R /2  $ \Comment See Fig.\ref{fig:sampling}
\State $ \Delta Z \gets \at{\re \Delta z, \im \Delta z, 0, \ldots, 0}  \in \r^{2d}, Z_c \gets \at{\re z_c, \im  z_c, 0, \ldots, 0}  \in \r^{2d} $

\State Sample-found $\gets$ FALSE, $z \gets 0 $
\While{ $\Not$ Sample-found }
\State $N \gets$ random integer from the interval $[ -N_{\max} , N_{\max} ]$ \label{ln:rip:N}
\State $ t \gets Z_c + N \Delta Z $ \label{ln:rip:t}
\State $ m \gets \rnd{B^{-1}t} $ \label{ln:rip:m}\Comment Find $m$ such that $Bm \in t + \CC(B) $
\State Pick $z'$ uniformly at random from set $\mathrm{SHIFTS}$ \label{ln:rip:zpr}
\State $z \gets z' + m_1 z_1 + \ldots + m_{2d} z_{2d}$ \label{ln:rip:ze}
\State Sample-found $\gets \vecs\at{z} \in S_{r,\vp,\ve} \cap \set{ x \in \r^{2d} : \ip{x - t,\Delta Z} \in (-1/2,1/2] }$ \label{ln:rip:z}
\EndWhile 
\State \Return z 
\EndProcedure
\out $z$, the element of $\z_K$ such that 
\Statex[2] $\abs{\s_{k,+}\at{z}} \le R_k$ for  $k=2,\ldots,d$  and  $\re \at{\at{\s_{1,+}\at{z}-z_0}e^{-i\vp/2} } \ge 0, \abs{ \s_{1,+}\at{z} } \le R$  
\Statex[2] where $z_0 = R\at{1-\ve^2}e^{-i\vp/2}$ (see Fig.~\ref{fig:inequality} for the visualization of the condition on $\s_{1,+}\at{z}$)
\end{algorithmic}
\par\end{centering}

\rule[0.5ex]{1\columnwidth}{1pt}
\end{figure}

\begin{figure}[t]

\caption{\label{fig:sampling} Illustration of sampling scheme. The vertical
axis is $ \im \at z$ and the horizontal axis is $ \re \at z$. Small blue circles have radius $R^{\min}_1$. During our sampling procedure we pick the center of blue circle at random. Then we find point $\sum_k {m_k z_k}$ (see Fig.~\ref{fig:proc-find-point}) from $\z_K$ such that $ \s_{1,+}\at{\sum_k {m_k z_k}}$ is in blue circle. Next we compute $z$ by adding randomly picked shift $z'$ to $\sum_k {m_k z_k}$ and check if $\vecs\at{z}$ is in $S_{r,\vp,\ve}$. The set $P_1 S_{r,\vp,\ve}$ is a blue circle segment on the picture below. }

\rule[0.5ex]{1\columnwidth}{1pt}
\begin{centering}
%auto-ignore
%!TEX root = quaf.tex

% \usetikzlibrary{calc}

\begin{tikzpicture}[scale=3,>=latex]
\pgfmathsetmacro{\textshift}{1}
\pgfmathsetmacro{\axessize}{1.4}
\pgfmathsetmacro{\xmax}{3}
\pgfmathsetmacro{\xmin}{-\xmax}
\pgfmathsetmacro{\R}{1.1} %D
\pgfmathsetmacro{\eps}{0.5} %epsilon %epsilon
\pgfmathsetmacro{\phi}{40} %phi 
\pgfmathsetmacro{\cPh}{cos(\phi/2)}
\pgfmathsetmacro{\sPh}{sin(\phi/2)}
\pgfmathsetmacro{\d}{\R*\eps*sqrt(2-\eps*\eps)}
\pgfmathsetmacro{\dH}{\R*\eps*sqrt(1-\eps*\eps/4)}
\pgfmathsetmacro{\RR}{\R*(1-\eps*\eps)}
\pgfmathsetmacro{\RH}{\R*(1-\eps*\eps/2)}
\pgfmathsetmacro{\RC}{\R*(1-\eps*\eps*3/4)}
\pgfmathsetmacro{\DD}{\R*\eps*\eps/2)}

\clip(\xmin,-1.2) rectangle (\xmax,1.8);

\coordinate (O) at (0,0);
\coordinate (R) at ( \R * \cPh,  -\R * \sPh );
\coordinate (RO) at ( \R * \sPh,  \R * \cPh );
\coordinate (ZC) at ( \RR * \cPh,  -\RR * \sPh  );
\coordinate (ZCH) at ( \RH * \cPh,  -\RH * \sPh  );
\coordinate (ZCC) at ( \RC * \cPh,  -\RC * \sPh  );
\coordinate (dZ) at (\d*\sPh,\d*\cPh);
\coordinate (dZH) at (\dH*\sPh,\dH*\cPh);
\coordinate (Z1) at ($(ZC)+(dZ)$);
\coordinate (Z2) at ($(ZC)-(dZ)$);
\coordinate (Z1H) at ($(ZCH)+(dZH)$);
\coordinate (Z2H) at ($(ZCH)-(dZH)$);
\coordinate (Z1PH) at ($(ZC)+(dZH)$);
\coordinate (Z2PH) at ($(ZC)-(dZH)$);
\coordinate (DZ) at (\DD*\sPh,\DD*\cPh );
\coordinate (DZO) at (3*\DD*\cPh,-3*\DD*\sPh );
\coordinate (DZP) at (\DD*\sPh +\DD*\cPh ,\DD*\cPh - \DD*\sPh);
\coordinate (DZM) at (\DD*\sPh -\DD*\cPh ,\DD*\cPh + \DD*\sPh);
\pgfmathsetmacro{\M}{5}

\begin{scope}
\tkzClipCircle(O,R)
\fill [cyan!10!white] (Z1) -- (Z2) -- ($2*(Z2)$) -- ($2*(Z1)$)  ;
\end{scope}

\draw[ dashed] (Z2H) -- ($-1*(Z1H) - (DZO)$);
\draw[dashed] (Z1H) -- ($-1*(Z2H) - (DZO)$);
\draw[ dashed] (Z2) -- ($-1*(Z1)$);
\draw[dashed] (Z1) -- ($-1*(Z2)$);

\draw[<->] ($-1*(Z1H) - (DZO)$)--($-1*(Z2H) - (DZO)$) node [midway, below, sloped] (tN) {$H' =  R\ve\sqrt{4-\ve^2}$};
\draw[<->] ($-1*(Z1)$)--($-1*(Z2)$) node [midway, below, sloped] (tN) {$H =  2R\ve\sqrt{4-4\ve^2}$};

% \foreach \x in {-\M,...,\M}
% {
%  \fill[cyan!20!white] ($(ZCC)+\x*(DZ)+0.5*(DZP)$ ) -- ($(ZCC)+\x*(DZ)-0.5*(DZM)$) -- ($(ZCC)+\x*(DZ)-0.5*(DZP)$) -- ($(ZCC)+\x*(DZ)+0.5*(DZM)$) -- ($(ZCC)+\x*(DZ)+0.5*(DZP)$ ) ;
% } 

\draw[->,gray] (O)--(\axessize,0) node [left,below,black] (tN) {$\quad x=\Re(z)$};
\draw[->,gray] (O)--(0,\axessize) node [above,right,black] (tN) {$y=\Im(z)$};

\tkzDrawCircle(O,R)
\draw (Z1)--(Z2);
\foreach \x in {-\M,...,\M}
{
 \fill[cyan!20!white] ($(ZCC)+\x*(DZ)$) circle (\DD/2 * 0.8 );
 \draw[black] ($(ZCC)+\x*(DZ)$) circle (\DD/2 * 0.8 );
%  \draw ($(ZCC)+\x*(DZ)+0.5*(DZP)$ ) -- ($(ZCC)+\x*(DZ)-0.5*(DZM)$) -- ($(ZCC)+\x*(DZ)-0.5*(DZP)$) -- ($(ZCC)+\x*(DZ)+0.5*(DZM)$) -- ($(ZCC)+\x*(DZ)+0.5*(DZP)$ ) ;
 \filldraw ($(ZCC)+\x*(DZ)$) circle (0.01);
} 

\filldraw (ZCC) circle (0.01) node [left] (tn) {$z_c\quad$};
 \draw[->,gray]   ($(tn.east)+(-0.15,0)$) -- (ZCC);
\filldraw ($(ZCC)+(DZ)$) circle (0.01) node [left] (tn) {$z_c + \Delta z \quad$};
 \draw[->,gray]   ($(tn.east)+(-0.15,0)$) --($(ZCC)+(DZ)$);
 \draw[->] (ZCC) -- ($(ZCC)+(DZ)$);

\node () at (O) {} node [below] (tN) {$0$};
\node[above right] (pA) at (Z1) {$ z_0 - id'e^{-i\vp/2}$};
\node[below right] (pB) at (Z2) {$ z_0 + id'e^{-i\vp/2}$};
\node[right] (pA) at (Z1H) {$ z_0 - ide^{-i\vp/2}$};
\node[right] (pB) at (Z2H) {$ z_0 + ide^{-i\vp/2}$};
\draw[->] (O)--($(RO)$) node [midway, below, sloped] (tN) {$\quad R$};

\node () at (2.1,1.1) {$\begin{array}{lcl}%
z_0 & = & R(1-\ve^2)e^{-i\vp/2} \\ % 
z_c & = & R(1-3\ve^2/4)e^{-i\vp/2} \\ % 
d' & = & \frac{R\ve}{2}\sqrt{2-\ve^2/4}\\ %
d & = & R\ve\sqrt{2-\ve^2}  \\ %
\Delta z & = & i e^{-i\vp/2} \ve^2 R /2 
\end{array}$};
% \tkzDrawPoints[color=black,fill=red,size=6](O,A,B,C,D,AT,BT,CT,DT,ATS,BTS,CTS,DTS);

\tkzDrawPoints[color=black,fill=red,size=6](O,R,Z1,Z2,ZCC,Z1H,Z2H);
\end{tikzpicture}
\end{centering}
\rule[0.5ex]{1\columnwidth}{1pt}

\end{figure}
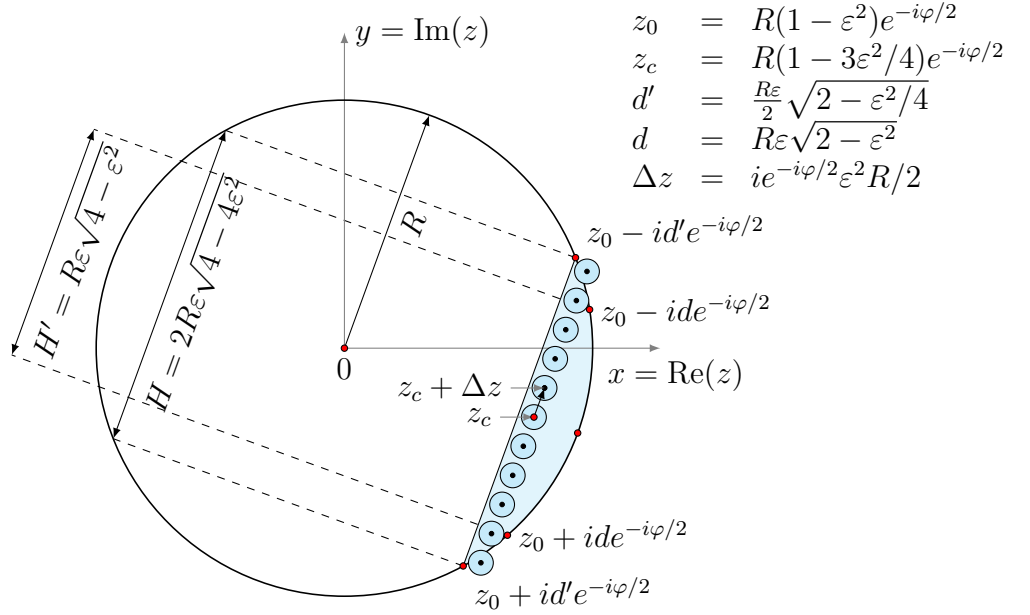

In this subsection we prove the Theorems~\ref{thm:find-point} used in the proof of Theorem \ref{thm:main} in Section~\ref{sec:high:level:desc}:

\findpoint*

\begin{proof} \label{thm:find-point-proof}

The procedure returns points from the set $\Cand_{r,\vp,\ve}$ because on line~\ref{ln:rip:z} in Fig.~\ref{fig:proc-find-point} we select $z$ such that $\vecs\at{z} \in S_{r,\vp,\ve}$ (see Equation~(\ref{eq:cand})). Let us next show that the procedure terminates, on average, after fixed number of loop interations. We estimate the probability $p_{\Cand}$ of variable Sample-found being TRUE (see line~\ref{ln:rip:z} in Fig.~\ref{fig:proc-find-point}). The number of the main loop iterations of RANDOM-INTEGER-POINT procedure has a geometric distribution with parameter $p_{\Cand}$. Let us now lower bound $p_{\Cand}$. Let $H'=R\ve\sqrt{4-\ve^2}$~(see Fig.~\ref{fig:sampling}). When the absolute value of $N$ (on line~\ref{ln:rip:N} in Fig.~\ref{fig:proc-find-point}) is less then $\rnd{H'/ \at{\ve^2 R}}-1$ (see Fig.~\ref{fig:sampling}) the shifted fundamental region $t + \CC\at{B}$ is a subset of $S_{r,\vp,\ve}$ (note that $t$ is computed on line~\ref{ln:rip:t} in Fig.~\ref{fig:proc-find-point} based on $N$). In the case if $z'=0$ (line \ref{ln:rip:zpr} in Fig.~\ref{fig:proc-find-point}) inclusion $t + \CC\at{B} \subset S_{r,\vp,\ve}$ implies that $\vecs\at{z}$ is in $S_{r,\vp,\ve}$. The probability of $z' = 0$ is equal to $1/\abs{\mathrm{SHIFTS}}$ (see line~\ref{ln:rip:shifts} in Fig.~\ref{fig:proc-find-point}). For this reason we can lower bound  $p_{\Cand}$ as: 
\[
 p_{\Cand} \ge \frac{2\rnd{H'/ \at{\ve^2 R}} - 1  }{2\rnd{H/\at{\ve^2 R}}+1} \cdot \frac{1}{\abs{\mathrm{SHIFTS}}} \ge \frac{ H'/H - 2\at{\ve^2 R / 2 H}  }{ 1 + 2\at{\ve^2 R / 2 H} } \cdot \frac{1}{\abs{\mathrm{SHIFTS}}}
\]
We note that $H'/H = \sqrt{4-\ve^2}/\at{2\sqrt{4-4\ve^2}} \in (1/2,1/\sqrt{3})$ and $(\ve^2 R / 2 H') = \ve / \at{ 4\sqrt{ 4 - 4\ve^2}} \in (0,1/(8\sqrt{3}))$. The constraint $\ve \in (0,1/2)$ implies that $p_{\Cand}$ is lower bounded by constant independent on the input to the online part of RANDOM-INTEGER-POINT procedure. 

We have shown that the online part of the procedure consist of a fixed number of arithmetic operations on average. To show that the procedure runs in polynomial time it is sufficient to show that the absolute value of the logarithm of absolute precision required for the computation is bounded by a polynomial in $\log\at{1/\ve}$. Let's analyze line~\ref{ln:rip:m} in Fig.~\ref{fig:proc-find-point} in more details. Let $\appr{B^{-1}}$ be an approximation to $B^{-1}$ to within precision $\delta_c$ and let $t'$ be an approximation to $t$ within precision $\delta_c$ and let $m' = \appr{B^{-1}} t'$. Now we show that the norm of the projection of $B m' - t'$ can also be bounded in terms of $R^{\min}_k$. For convenience, let $P_k$ be a projector on the subspace spanned by $e_{2k-1},e_{2k}$.
\[
 \nrm{P_{k}\at{B m' - t}} \le \nrm{t-t'} +
 \nrm{ P_k  B\at{ \rnd{\appr{B^{-1}}t'} - \appr{B^{-1}}t' } } +
 \nrm{t'} \nrm{B\appr{B^{-1}}-I}
\]
Now we see that $\nrm{t-t'}$ is bounded by $\delta_c$, the second term in the sum above is bounded by $R^{\min}_k$ and the third term is bounded by some fixed constant times $\delta_c \nrm{t'}$. This implies that we can find $m'$ such that $\nrm{P_{k}\at{B m' - t}} \le R^{\min}_k + \delta'_c$. The absolute value of the logarithm of absolute precision required for the computation is bounded by polynomial in $\log\at{1/\ve}$ and $\log\at{1/\delta'_c}$ because $\log \nrm{t'}$ is bounded by polynomial in $\log\at{1/\ve}$ . This is sufficient for our purposes because it is sufficient to choose $\log\at{1/\delta'_c}$ to be of order $\log \nrm{t'}$.

It remains to show that we get every point from $\Cand_{r,\vp,\ve}$ with probability at least $p_0 / \abs{\Cand_{r,\vp,\ve}} $. We first introduce some notation convenient for the proof. Let $\xi_N$ be a random variable corresponding to $N$ (line \ref{ln:rip:N} in Fig.~\ref{fig:proc-find-point}), $\xi_{z'}$ be a random variable corresponding $z'$ (line \ref{ln:rip:zpr} in Fig.~\ref{fig:proc-find-point}) and $\xi_z$ be a random variable corresponding to value $z$ computed on line~\ref{ln:rip:ze} in Fig.~\ref{fig:proc-find-point}. Random variable $\xi_z$ is a function of $\xi_N$ and $\xi_{z'}$. Random variable $\xi_N$ takes integer values in the range $[-N_{\max},N_{\max}]$ with equal probability and $\xi_{z'}$ takes values in the finite set $\mathrm{SHIFTS}$ also with equal probability. Variables $\xi_{z'}$ and $\xi_N$ are independent. Let $z_0$ be a fixed element of $\Cand_{r,\vp,\ve}$. It is sufficient to lower bound $P\at{\xi_z = z_0} $ because the probability of the output of RANDOM-REGION-CENTER being $z_0$ is $P\at{\xi_z = z_0}/P\at{\xi \in \Cand_{r,\vp,\ve} }$. Given $z_0$ there are unique values of $N$ (line~\ref{ln:rip:N}), $m$ (line~\ref{ln:rip:m}) and $t$ (line~\ref{ln:rip:t}), denoted as $N_0, m^0$ and $t_0$, such that $\vecs\at{z_0}$ is in 
\[
S_{r,\vp,\ve} \cap \set{ x \in \r^{2d} : \ip{x - t_0,\Delta Z} \in (-1/2,1/2] }.
\]
The equality 
\[
P \at{\xi_z = z_0} = P\at{\xi_z = z_0 / \xi_N = N_0 } P\at{\xi_N = N_0} = \frac{P\at{\xi_z = z_0 / \xi_N = N_0 }}{2 N_{\max} + 1}
\]
implies that it is sufficient to lower bound $P\at{\xi_z = z_0 / \xi_N = N_0 }$ and relate $2 N_{\max} + 1$ to the size of $\Cand_{r,\vp,\ve}$. Note that $P\at{\xi_z = z_0 / \xi_N = N_0 } = P\at{ z_0 = \xi'_z + \sum_k m_k^0 z_k } $. This implies that $P\at{\xi_z = z_0 / \xi_N = N_0 }$ is $1/\abs{\mathrm{SHIFTS}}$ if $z_0 - \sum_k m_k^0 z_k $  belongs to set $\mathrm{SHIFTS}$. Let us show that $z_0 - \sum_k m_k^0 z_k $ is always in $\mathrm{SHIFTS}$. It is sufficient to show that 
\begin{eqnarray*}
 & & \nrm{P_k\at{\vecs\at{z_0} - \vecs\at{ \sum_k m_k^0 z_k }} } = \nrm{P_k\at{\vecs\at{z_0} - B m^{0}} } \le  R^{\max}_k +  R^{\min}_k  \text{ for } k=\rg{2}{d} \\
 & &  \nrm{P_1\at{\vecs\at{z_0} - B m^{0}} }  \le 2\sqrt{5}R^{\max}_1
\end{eqnarray*}
It is useful to note that 
\[
 \nrm{P_k\at{\vecs\at{z_0} - B m^{0}} } \le \nrm{P_k\at{\vecs\at{z_0} - t_0 } } + \nrm{P_k\at{ t_0 - B m^{0} } } \le \nrm{P_k\at{\vecs\at{z_0} - t_0 } } + R_{k}^{\min}
\]
The fact that $\vecs\at{z_0}$ is in $ S_{r,\vp,\ve} \cap \set{ x \in \r^{2d} : \ip{x - t_0,\Delta Z} \in (-1/2,1/2] }$ implies that 
\begin{eqnarray*}
 & & \nrm{P_k\at{\vecs\at{z_0} - t_0 } } \le \sqrt{\s_k\at{r}}  \le  R^{\max}_k \text{ for } k=\rg{2}{d}. 
\end{eqnarray*}
To establish bound on $\nrm{P_1\at{\vecs\at{z_0} -  B m^{0} } } $ we observe that $P_1 \vecs\at{z_0}$ and $P_1 B m^{0}$ both belong to a set with the diameter $\sqrt{5}\ve^2 R/2$. We also have shown that: 
\[
\Cand_{r,\vp,\ve} \subset \mathrm{SHIFTS} +   \set{ \vecs^{-1} \at{ B \rnd{ B^{-1}\at{Z_c + N\Delta Z} } } : N \in [-N_{\max},N_{\max}] }
\]

Finally we note that 
\[
 \set{   z \in Z_K : \vecs\at{z} \in Z_c + N\Delta Z + \CC(B), \abs{N} \le \rnd{H'/ \at{\ve^2 R}} - 1, N \in \z } \subset \Cand_{r,\vp,\ve}.
\]
This implies that if $\Cand_{r,\vp,\ve}$ is non-empty we have: 
\[
\frac{ \abs{\Cand_{r,\vp,\ve}} }{2 N_{\max} + 1} \ge \frac{ \abs{\Cand_{r,\vp,\ve}} }{2 H/(\ve^2 R) + 2 }\ge \frac{ \abs{\Cand_{r,\vp,\ve}} }{  \frac{H }{H' } \at{2 \rnd{ H'/(\ve^2 R) } + 1} + 2} \ge \frac{ \abs{\Cand_{r,\vp,\ve}} }{6  + 2 \abs{\Cand_{r,\vp,\ve}}} \ge \frac{1}{8}
\]
We conclude that $P\at{\xi_z = z_0} \ge \frac{1 }{\abs{\Cand_{r,\vp,\ve} } } \cdot p_0$, where $p_0$ is $\frac{1}{8\abs{\mathrm{SHIFTS}}}$. Above derivation also gives us required bounds on the size of $\Cand_{r,\vp,\ve}$. 
\end{proof}

\subsubsection{Implementation aspects} In practice we are looking for best possible value of the additive constants $C_{\min},C_{\max}$  in the Theorem~\ref{thm:main} we can achieve while maintaining the polynomial runtime of the online part of the algorithm. We refer the reader to the Appendix~\ref{sec:appendix-b} for the version of the procedure we use in our implementation to obtain the numerical results reported in Section~\ref{sec:examples}. In practice we use the Nearest Plane Algorithm~\cite{BABAI:1986}. It is also possible to show that $R^{\min}_k$ can be chosen to be based on $\CC\at{B^\ast}$, not based on $\CC\at{B}$. We ensure that the basis we use is Hermite-Korkine-Zolotarev \cite{Kannan:1983,Nguyen:2004,Hanrot:2007} reduced which makes it possible to guarantee that $R^{\min}_k$ are bounded by some functions of discriminant of $\z_K$ and these bounds are independent on the choice of the basis of $\z_K$. We also use a simpler version of the sampling procedure. The simpler version does not ensure that the distribution of procedure outcomes is close to uniform, but works well in practice. 

%\end{comment}

% \clearpage

%auto-ignore
%!TEX root = quaf.tex

% \newpage

\subsection{Multiplicative approximation using unit group}

\begin{figure}[ht!]

\protect\caption{\label{fig:proc-unit-adjust} UNIT-ADJUST procedure}
\rule[0.5ex]{1\columnwidth}{1pt}

\begin{centering}

%auto-ignore
%!TEX root = quaf.tex

\begin{algorithmic}[1]
\fxinput $F$, $u_1,\dotsc,u_{d-1}$
\Statex[1] $F$ is a totally real number field of degree $d$
\Statex[1] $u_1,\dotsc,u_{d-1}$ form a system of fundamental units of $F$
\inp $\rg{t_1}{t_d}$
\Statex[1] $t_1,\ldots,t_d$ -- real numbers of the same precision $n$
\Procedure{UNIT-ADJUST}{}
\offline
\State $u_1,\dotsc,u_{d-1} \leftarrow$ FUNDAMENTAL-UNITS($F$)
\State $\delta_0 \gets \sqrt{ \max_{k,j}\set{\abs{\s_k\at{u_j}},\abs{\s_k\at{u^{-1}_j}} } } $ \Comment $\delta_0 > 1$
\State $B \leftarrow 
\pmat{\log|\sig_1(u_1)| & \cdots & \log |\sig_1(u_{d-1})| & 2\log \delta_0\\ 
\vdots  & & \vdots & \vdots \\ 
\log|\sig_d(u_1)| & \cdots & \log |\sig_d(u_{d-1})| & 2 \log \delta_0}$

\State $\delta_k \leftarrow \delta_0 \sqrt{\prod_{i=1}^{d-1}\max\{|\sig_k(u_i)|,|\sig_k(u_i^{-1})|\}}$ 

\spec $\rg{\delta_1}{\delta_d}$ -- real numbers such that $\delta_k > 1$ and  

$\CC(B)\subset [-\log \delta_1,\log\delta_1] \times \cdots \times [-\log \delta_d,\log\delta_d]$

\online
\State \assert $\abs{\prg{t_1}{t_d}} < \log \delta_0$ \Comment Make sure the point is in the span of the lattice 
\State $m \leftarrow \rnd{B^{-1} t}$
\State  $u\leftarrow \prod_{i=1}^{d-1}u_i^{m_i}$
\EndProcedure
\out unit $u \in \z_F^\times$ such that
\Statex for all $k=\rg{1}{d}$ : $\abs{\log\abs{\s_k\at{u}} - t_k } \le \log \delta_k $

\end{algorithmic}

\par\end{centering}

\rule[0.5ex]{1\columnwidth}{1pt}

\end{figure}

%\begin{itemize}
%\item Reduce the problem to Closest Vector Problem or approximate Closest Vector Problem for lattice with integer valued Gram matrix ; see previous section for options for solving CVP or Approximate CVP;
%\begin{itemize}
%\item Use logarithmic embedding of unit group
%\item Analyse what happens when the lattice with real valued Gram matrix is rounded into lattice with integer valued Gram matrix.
%\item Show that one can use rounded lattice
%\end{itemize}
%\item Show that the bitsize of the unit output by UNIT-ADJUST (See Fig.~\ref{fig:proc-unit-adjust}) is polynomial in the relevant parameter
%\item Prove correctness of UNIT-ADJUST (See Fig.~\ref{fig:proc-unit-adjust}); bound its runtime
%\end{itemize}

In this section we prove

\unitadjust*

%\newpage

The offline part of procedure UNIT-ADJUST computes a system of fundamental units $u_1,\dotsc, u_{d-1} \in \bZ_F^\times$ and outputs
\[\delta_k = \delta_0 \sqrt{\prod_{j=1}^{d-1}\max\{|\sig_k(u_j)|,|\sig_k(u_j^{-1})|\}}\]
for $k = 1,\dotsc,d$, where $\delta_0 > 1$ is some fixed constant.

When called with a target vector $t \in \bR^d$ satisfying $|(t,\bsone_d)| < \log \delta_0$, the online part of UNIT-ADJUST simply rounds off $t$ in the basis
\[B = \pmat{\log|\sig_1(u_1)| & \cdots & \log |\sig_1(u_{d-1})| & 2 \log \delta_0 \\
\vdots  & & \vdots & \vdots  \\
\log|\sig_d(u_1)| & \cdots & \log |\sig_d(u_{d-1})| & 2\log \delta_0}\]
to the lattice vector $Bm$, where $m = \rnd{B^{-1} t} \in \bZ^d$.  Then it returns the unit $u = u_1^{m_1} \cdots u_{d-1}^{m_{d-1}}$.

\begin{proof}[Proof that UNIT-ADJUST is correct]
Because $Bm$ is the unique lattice vector contained in the shifted parallelepiped $t + \CC(B)$, the following inequalities hold for $k = 1,\dotsc,d$:
\begin{eqnarray*}
|(Bm)_k - t_k| &\leq&  \max\{x_k : x\in \CC(B)\} \\
&=& \frac 12 \max_{y \in \{\pm 1\}^{d}} (By)_k \\
&=& \log \delta_0 + \frac 12 \sum_{j=1}^{d-1}\big|\log(|\sig_k(u_j)|)\big| \\
&=& \log \delta_k,
\end{eqnarray*}
\end{proof}
It is also worth noting the above shows that
\[\CC(B) \subset [-\log \delta_1,\log\delta_1] \times \cdots \times [-\log\delta_d,\log\delta_d]\]
and $\norm{Bm - t}_B \leq 1$, where $\norm{x}_B := \inf\{y > 0 : x \in \CC(B)y\}$.

%\newpage

Now we show that the running time is a polynomial in $\norm{t}$ and in the number of bits used to specify $t_k$
\begin{proof}[Proof that UNIT-ADJUST runs in polynomial time]
Suppose that the $t_k$ are given with $n$ bits of precision.  Then they can be specified using $O(n + \log|t_k|)$ bits as $t_k = \pm 2^{\ell-n} s$, where $\ell =  \ceil{\log_2 |t_k|}$ and $s \in \{0,\dotsc,2^{n-1}\}$ is an $n$-bit integer.
First we observe that because the number field is fixed and $\delta_0 > 1$ is an arbitrary fixed constant, the the inverse $B^{-1}$ can be precomputed to sufficiently high precision and stored during the offline part.  The vector $m$ can therefore be computed in polynomial time. We also note that its norm is bounded by a polynomial in $\norm{t}$.  Indeed,
\begin{eqnarray*}
\norm{m}   & \leq & \norm{m - B^{-1}t} + \norm{B^{-1}t} \\
&\leq & \frac{\sqrt{ d}}{2} + \lambda_\mmin(B)^{-1}\norm{t} \\
&\leq & O(\norm{t}).
\end{eqnarray*}
This further implies that each $|m_i| \leq O(\norm{t})$, so that the output unit $u = u_1^{m_1} \cdots u_{d-1}^{m_{d-1}}$ can be computed by polynomially-many multiplications of the fundamental units $u_i$.  Therefore,
\[\log\norm{u} = O(|m_1| \log \norm{u_1} + \cdots + |m_d| \log \norm{u_d}) = O(\poly(\norm t , \norm{u_1},\dotsc,\norm {u_d})),\]
implying that the output unit can indeed be computed in polynomial time.
\end{proof}

\subsubsection{Implementation aspects} For proving that our algorithm runs in polynomial time it is sufficient to show that $\delta_k$ are fixed numbers for a given quaternionic gate set specification. It does not in principal matter how big they are. However, we see that the additive constant $C_{\min}$ in Theorem~\ref{thm:main} depends on values of $\delta_k$. When implementing our algorithm in practice we aim to achieve smallest possible constant $C_{\min}$ while maintaining good performance. For this reason we use the Nearest Plane Algorithm~\cite{BABAI:1986} instead of the simple round off procedure shown in Fig.~\ref{fig:proc-unit-adjust}. The Nearest Plane Algorithm has runtime that is polynomial in dimension of the lattice and bit sizes of the entries of vectors involved in the computation. In Appendix~\ref{sec:appendix-b} we show the pseudo-code for the variant of UNIT-ADJUST procedure we use in practice. Precision and complexity analysis above can be extended to this more practical approach. We refer the reader to \cite{NS:2009} for a careful precision and performance analysis of a variant of the Nearest Plane algorithm.

The results of applying the nearest plane algorithm depends on the quality of the basis used with it. In practice we apply Hermite-Korkine-Zolotarev \cite{Kannan:1983,Nguyen:2004,Hanrot:2007} or LLL\cite{NS:2009} reduction to the unit lattice basis during the offline step of the algorithm. This allows us to further lower the contribution from $\log \delta_1,\ldots,\log \delta_k$ to the additive constant $C_{\min}$. Value of $\log \delta_0$ can be chosen to be very small, and its contribution to $C_{\min}$ can be made negligible without high computational overhead. Values of $\delta_k$ computed based on reduced basis can also be related to the value of the regulator of number field $F$ and known techniques for bounding the regulator can be applied to bound them.

Computing the system of fundamental units of the number field is known to be a hard problem~\cite{Thiel:1995,Hallgren:2005,Eisentrager:2014} and can be to costly even for the offline part of the algorithm. In practice we can circumvent this issue to some extend. For our purposes it is sufficient to know the generators of the finite index subgroup of the unit group, but not the unit group itself. Frequently generators of such a subgroup can be computed much faster than the system of fundamental units \cite{PZ:89}, or are even known in analytic form \cite{Washington:82}. We used this approach to obtain values of constant $C_{\min}$ for some of the higher degree number fields.

\subsection{On the conjecture related to the approximation algorithm performance} \label{sec:conjecture}

On a high level, the performance of our approximation algorithm depends on the properties of the set of all possible solutions to QAP (Problem~\ref{prob:qap}). Let us recall the statement of QAP:
\qap*

Recall also, that map $U_q$ is constructed in Section~\ref{sec:quat-and-unitaries} using the embedding $\sigma : F \rightarrow \r$ that is a part of quaternion gate set specification. 

Next we construct a formal description to the set of all solution to QAP. Let $L$ be a generalized Lipschitz order in $\at{\frac{a,b}{F}}$ (Section~\ref{sec:lip-order}). We will use the following set as a part of the description of all possible solutions to QAP: 
\[
\sln_{r,\vp,\ve} = \set{ q \in L : \nrm{U_q - R_z\at{\vp}}\le \ve, \nrd\at{q} = r }
\]
The set of all possible norms of quaternions with given cost vector $L_1,\ldots,L_M$ is given by
\[
\nrd_{L_1,\ldots,L_M} = \set{ r \in \z_F : r\z_F = \p_1^{L_1}\ldots\p_M^{L_M}, r \text{ is totally positive} }
\]
Using the notation above, the set of all solutions to given instance of QAP is given by:
\[
\sln_{\mathrm{QAP}} = \bigcup_{r \in \nrd_{L_1,\ldots,L_M} } \sln_{r,\vp,\ve}.
\]
Note that for any unit $u$ for $\z^{\times}_F$ it is the case that $U_q = U_{q u}$. For this reason, the set of all solutions can be obtained as: 
\[
\sln_{\mathrm{QAP}} = \bigcup_{r \in \nrd_{L_1,\ldots,L_M} / (\z^{\times}_F)^2 } \sln_{r,\vp,\ve},
\]
where the set $\nrd_{L_1,\ldots,L_M} / (\z^{\times}_F)^2$ is finite. Its size is equal to the size of the set: 
\[
\set{ u \in \z^{\times}_F : u \text{ is totally positive } } / (\z^{\times}_F)^2. 
\]
Let us now discuss the structure of the set $\sln_{r,\vp,\ve}$. Consider $q$ from $\sln_{r,\vp,\ve}$. Quaternion $q$ can be described by two elements $z_1,z_2$ of $\z_K$ as
\[
 q = e_1^{-1}\at{z_1} + e_2^{-1}\at{z_2} 
\]
Note that equality $r = z_1 z^{\ast}_1 - b z_2 z^{\ast}_2$ and condition $ \nrm{U_q - R_z\at{\vp}}\le \ve$ imply that 
\[
 \re\at{\at{\s_{1,+}\at{z_1}-z_0}e^{i\vp/2}}\ge 0,\,\s_{k,+}\at{z_1} \le \s_k\at{r} 
\]
where $z_0=\sqrt{\sigma_1\at{r}}(1-\ve^2)e^{-i\vp/2}$. In other words $z_1$ belongs to the set $\Cand_{r,\vp,\ve}$ defined as: 
\[
\Cand_{r,\vp,\ve} = \set{ z \in \z_K : \re\at{\at{\s_{1,+}\at{z_1}-z_0}e^{i\vp/2}}\ge 0,\,\s_{k,+}\at{z_1} \le \s_k\at{r} } = \set{ z : \vecs\at{z} \in S_{r,\vp,\ve}}
\] 
The observations above allows us to rewrite the set $\sln_{r,\vp,\ve}$ as following: 
\[
\sln_{r,\vp,\ve} = \bigcup_{z_1 \in \Cand_{r,\vp,\ve} }\set{ e_1^{-1}\at{z_1} + e_2^{-1}\at{z_2} : z_2 \in \z_K, \abss{z_2} = \at{r - \abss{z_1}}/(-b) }
\] 
Note that some sets in the above union can be empty, because the relative norm equation $\abss{z_2} = \at{r - \abss{z_1}}/(-b) $ does not always have a solution. Motivated by this fact we can define the set 
\[
\term_{r,\vp,\ve} = \set{ z_1 \in \Cand_{r,\vp,\ve} : \text{ there exists } z_2 \in \z_K \text{ such that } \abss{z_2} = \at{r - \abss{z_1}}/(-b)  } 
\]
Let us assume that we have an oracle for solving the relative norm equations and drawing points from $\term_{r,\vp,\ve}$. Under this assumptions we could have the following algorithm for solving QAP: 
\begin{enumerate}
\item Pick random $r$ from $\nrd_{L_1,\ldots,L_M}$
\item Pick $z_1$ from $\term_{r,\vp,\ve}$
\item Find $z_2$ by solving relative norm equation $\abss{z_2} = \at{r - \abss{z_1}}/(-b) $
\item Return $q = e_1^{-1}\at{z_1} + e_2^{-1}\at{z_2}$
\end{enumerate}
Suppose now, that we don't have an oracle for drawing points from $\term_{r,\vp,\ve}$. We can modify our algorithm as following: 
\begin{enumerate}
\item Pick random $r$ from $\nrd_{L_1,\ldots,L_M}$
\item Pick random element $z_1$ from $\Cand_{r,\vp,\ve}$
\item Check if $z_1$ is in $\term_{r,\vp,\ve}$. If this is not the case return to Step 2. 
\item Find $z_2$ by solving relative norm equation $\abss{z_2} = \at{r - \abss{z_1}}/(-b) $
\item Return $q = e_1^{-1}\at{z_1} + e_2^{-1}\at{z_2}$
\end{enumerate}
Note that if the ratio $\abs{\term_{r,\vp,\ve}}/\abs{\Cand_{r,\vp,\ve}}$ were in $\Omega\at{1/\log\at{1/\ve}}$, then our algorithm would still run in polynomial time. In practice we don't have an oracle that solves all relative norm equations in polynomial time ( or even checks if there is a solution to given relative norm equation). However, if we restrict the possible right hand sides of the relative norm equation we can check the existence of the solution and find one in polynomial time. This motivates the following definition: 
\[
\pterm_{r,\vp,\ve} = \set{ z_1 \in \term_{r,\vp,\ve} : N_{K/\q}\at{\at{r - \abss{z_1}}/(-b)} \text{ is a rational prime}  }
\]
This gives us the following algorithm, which is very close to the procedure APPROXIMATE in Fig.~\ref{fig:proc-approximate}: 
\begin{enumerate}
\item Pick random $r$ from $\nrd_{L_1,\ldots,L_M}$
\item Pick random element $z_1$ from $\Cand_{r,\vp,\ve}$
\item Check if $z_1$ is in $\pterm_{r,\vp,\ve}$. If this is not the case return to Step 2. 
\item Find $z_2$ by solving relative norm equation $\abss{z_2} = \at{r - \abss{z_1}}/(-b) $
\item Return $q = e_1^{-1}\at{z_1} + e_2^{-1}\at{z_2}$
\end{enumerate}
If the ratio $\abs{\pterm_{r,\vp,\ve}}/\abs{\Cand_{r,\vp,\ve}}$ were in $\Omega\at{1/\log\at{1/\ve}}$ and we were drawing samples from $\Cand_{r,\vp,\ve}$ sufficiently uniformly, the algorithm above would still run in polynomial time. In this case in the absence of the oracle for solving arbitrary norm equation. Note the above discussion implies, that 
\[
 \frac{\abs{\pterm_{r,\vp,\ve}}}{\abs{\Cand_{r,\vp,\ve}}} = \frac{\abs{\pterm_{r u^2,\vp,\ve}}}{\abs{\Cand_{r u^2,\vp,\ve}}} \text{ for any } u \in \z^{\times}_F.
\]
For this reason, the ratio above is well defined for $r / (\z^{\times}_F)^2$. The conjecture that implies that our algorithm runs in polynomial time is the following:
\begin{conj} \label{cnj:main} Keeping the notation introduced before in this section, for any $r$ from 
\[
\nrd_{L_1,\ldots,L_M} / (\z^{\times}_F)^2 = \set{ r \in \z_F : r\z_F = \p_1^{L_1}\ldots\p_M^{L_M}, r \text{ is totally positive} } / (\z^{\times}_F)^2
\]
the ratio $\abs{\pterm_{r,\vp,\ve}}/\abs{\Cand_{r,\vp,\ve}}$ is in $\Omega\at{1/\log\at{1/\ve}}$. 
\end{conj}

% \subsubsection{Questions related to the conjecture}
Recently it was pointed out by Peter Sarnak \cite{Sarnak:2015} that the questions about properties of solutions to QAP~(Problem~\ref{prob:qap}) can be studied using the methods developed in \cite{S1,S2} and used to study the spectral gap of the Hecke operator associated to different gate sets.

\def\exPrimes{{\mathcal P}}

\section{Relative norm equations}\label{sec:relativeNormEqs}

In this section we show how a solution $z\in K$ to a relative norm equation of the form $N_{K/F}\at z=e$ between a CM field $K$ and its totally real subfield $F=K\cap \r$ can be efficiently computed, provided such solutions exist at all for the given right hand side $e\in \z_F$. The totally positive element $e$ arises from the RANDOM-INTEGER-POINT step in lines \ref{ln:appr:find-pt} and \ref{ln:appr:e} in the main algorithm (see Figure \ref{fig:proc-approximate}). 
Formally, we deal with the following problem:
\begin{prob}[CM relative norm equation]\label{prob:NE}
\label{prob:rnf-1}Let $K/F$ be a CM field of constant degree over $\q$ and let $e$ be a totally positive element of
$\z_{F}$. The task is to find an element $z$ of $\z_{K}$ such that $N_{K/F}\at z= z z^* = e$ in time polynomial in the bit-size of $e$, provided such an element $z$ exists.
\end{prob}

In the following we describe our approach in solving a relative norm equation as in Problem \ref{prob:NE} and give pseudo-code implementations of IS-EASILY-SOLVABLE step in line \ref{ln:appr:is-easy} and FAST-SOLVE-NORM-EQ in line \ref{ln:appr:norm-eq} of the main algorithm in Figure \ref{fig:proc-approximate}. 

Relative norm equations of the form
$N_{K/F}\at z=e$ have been studied in the literature before. Early approaches include various methods that proceed by establishing a bounding box that will contain a solution provided it exists and then checks the candidates in the bounding box \cite{BS:67,PZ:89,FJP:97,Garbanati:80}. Unfortunately, these methods are exponential in the bit-size of the right hand side. Next, there is a method based on $S$-units \cite{Simon:2002,Coh2}. This requires the factorization of the right hand side of the equation, along with precomputation of the relative class group of the extension $K/F$, and some additional data that is dependent on the right hand side. Therefore, it is not clear that the resulting algorithm runs in polynomial time. Relative norm equations have also been studied in the context of cryptanalysis of lattice based cryptography, e.g., of the NTRU system. See also \cite{GS:2002}, where an algorithm is described to solve relative norm equations for cyclotomic fields over their totally real subfield. This algorithm uses Fermat's little theorem for ideals in $\z_K$ in conjunction with LLL reduction to find a solution, which is known to exist in the context in which the algorithm is applied, see also \cite{GGH:2013}. However, like Simon's $S$-unit based algorithm, the algorithm relies on some properties of the right hand side, and therefore does not seem to run in time that is a polynomial in the bit-size of the right hand side.

We take a different route in which precompute a finite set of attributes of  $K/F$ that do not depend on the right hand side $e$. Our method is similar to \cite[Algorithm 4.10]{KV} in that it reduces the problem for general right hand side to a bounded size instance. Furthermore, we leverages the fact that $K$ is a CM field and that the right hand side is of a particular form, which we call benign integers. These are characterized in terms of a finite set ${\mathcal P} = \{ \mathfrak{p}_1, \ldots, \mathfrak{p}_k\} \subset \z_{F}$:  

\begin{dfn}[$\exPrimes$-benign integers]
\label{def:benign}
\label{prob:rnf-2} Let $\exPrimes$ be a set of prime ideals of $\z_{F}$. An integer $e \in \z_{F}$ is called {\em benign} if it is totally positive and the prime factorization of the ideal generated by $e$ satisfies
\begin{equation}\label{eq:benign}
e\z_{F}=\mathfrak{q} \prod_{\p\in \exPrimes}\p^{e\at{\p}},
\end{equation}
where $\mathfrak{q}$ is prime and ${e\at{\p}}\geq 0$ for all $\p \in \exPrimes$.
\end{dfn}

The primes in $\exPrimes$ are defined by the user parameter $P$ in algorithm APPROXIMATE (Figure \ref{fig:proc-approximate}) and are precomputed by the offline part of procedure IS-EASILY-SOLVABLE (Figure \ref{fig:proc-is-easily-solvable}). 

\subsection{Measuring the bit-size of the input}

There are several natural ways to measure the bit-size of the algebraic numbers that are involved as the input and the output of a relative norm equation. We briefly discuss some of these definitions and show that in our case they are all within a constant factor of each other. 

Let $K/\q$ be a Galois extension and let $B = \{ b_1, \ldots, b_n \}$ be a basis for $\z_K$ over $\q$, where $n = [K:\q]$, i.e., $\z_K = \bigoplus_{i=1}^n b_i \z$. Any $x\in \z_K$ can then  be represented as $x = \sum_{i=1}^n b_i x_i$ where $x_i \in \z$, i.e., we can define the bit-size with respect to $B$ as $\|x\|_B := \sum_{i=1}^n |x_i|$. Alternatively, for CM fields, have $n=2d$ where $[F:\q]=d$ and we can use the quadratic form $T_2(x) = \sum_{i=1}^d |\sigma_i(x)|^2$ as a measure for the bit-size of $x$. Also, we can use a notion of bit-size that is valid for general ideals $I \subseteq \z_K$ and not just for the principal ideals $e\z_F$: following \cite{BF:2012} we first choose a matrix $M\in \z^{n \times n}$ for a basis of $I$ expressed over an integral basis of $\z_K$. If $M$ is in Hermite Normal Form, then each entry can be bounded by $|\det(M)| = N_{K/\q}(I)$, i.e., we can define $S(I) := n^2 \log_2(N_{K/\q}(I))$ as the bit-size of $I$. For principal ideals, as in \cite{BF:2012} we define $S(x) := n \log_2(\max_i |x_i|)$.

It turns out that $T_2(x)$ and $S(x)$ are related. More precisely, we have the following result that we cite from \cite{BF:2012}:
\begin{lem}\label{lem:bitsizeNorm}
Let $K/\q$ be a CM field, let $x\in \z_K$, and let $\Delta_K = \det(T_2(b_i,b_j))^2$ denote the discriminant of $K$. Then the bound $\frac{1}{2} \log_2(T_2(x)) \leq \widetilde{O}\left( S(x)/n + n^2 + \log_2(\Delta_K)\right)$ holds. Furthermore, we have that $S(x) \leq \widetilde{O}(d (d+\frac{1}{2} \log_2(T_2(x))))$.
\end{lem}

We next establish a bound that allows to relate $T_2(x)$ to the bit-size of the expansion of $x$ with respect to any given basis $B$.

\begin{lem}\label{lem:bitsizeBasis}
Let $K/\q$ be a CM field and let $B = \{ b_1, \ldots, b_n \}$ be an integral basis for $K$ over $\q$, where $n = [K:\q]$. For $x = \sum_{i=1} x_i b_i \in K$ define the bit size of $x$ with respect to the basis $B$ as $\| x\|_B := \sum_{i=1}^n |x_i|^2$. Let $M = [Tr(b_i b_j^*)]_{i,j=1,\ldots, n}$ be the Gram matrix of $B$ and let $\lambda_{max}$ and $\lambda_{min}$ be the largest, respectively smallest eigenvalue of $M$. Then
\[
\lambda_{min} \|x\|_B /2 \leq T_2(x) \leq \lambda_{max}\| x \|_B /2.
\]
\end{lem}
\begin{proof}
Let $x\in K$ and let $x=\sum_{i=1}^n x_i b_i$ be its expansion over the chosen basis. Recalling that $T_2(x) = {Tr}_{F/\q}(x x^*) = {Tr}_{K/\q}(x x^*)/2$ we obtain that $T_2(x) = \sum_{i,j=1}^n x_i x_j^* Tr(b_i b_j^*)/2$.  We can rewrite this as $T_2(x) = (x_1, \ldots, x_n) M (x_1, \ldots, x_n)^t/2$ where $M$ is the integer valued, symmetric, and positive-definite matrix with entries $M_{i,j} = Tr(b_i b_j^*)$. %i.e., $T_2(x) = {\bf x} M {\bf x}^t$, where ${\bf x}=(x_1, \ldots, x_n)$. 
By diagonalizing $M$ in an eigenbasis, we see that $\lambda_{min} \|x\|_B/2 \leq T_2(x) \leq \lambda_{max} \| x \|_B$/2 as claimed.
\end{proof}

In the approach taken in this paper the field $K$ is considered to be a constant. This implies that quantities such as the degree $[K:\q]$ or the discriminant $\Delta_F$ of the totally real subfield $F = K^+$ are constants.  By choosing $B$ to be an LLL-reduced basis we obtain from Lemma \ref{lem:bitsizeBasis} that $T_2(x)$ and $\|x|\|_B$ are related by a constant factor, see also \cite{Belabas:2004}. To summarize, Lemmas \ref{lem:bitsizeNorm} and \ref{lem:bitsizeBasis} imply that all measures of bit-size considered in the following are within constant factors of each other.

\subsection{Hermite normal form and lifting of ideals}

In order to represent ideals in rings of integers, Hermite normal forms are an indispensable tool \cite{coh3,Coh2}. We cite a result from \cite{SL:96} that allows to give a polynomial bound on both, the complexity of computing a Hermite Normal Form (HNF) of an integer matrix, and the bit-size of the output. See also \cite{Storjohann:2000} for a discussion and comparison with other efficient algorithms to compute HNFs.

\begin{thm}\label{thm:hnf}
Let $A \in \z^{n \times m}$ be a rank $r$ integer matrix and let $\|A\| := \max_{i,j} |A_{i,j}|$.  There exists a deterministic algorithm that computes the HNF of $A$ in time $\widetilde{O}(m^{\theta}n \log{\|A\|})$, where the $\widetilde{O}$ notation ignores log-factors and $2
\leq \theta \leq 2.373$ is the exponent for matrix multiplication.
\end{thm}
We now discuss how to lift primes ideals in $\z_F$ to prime ideals in $\z_K$.
Write $K=F(\theta)$, where $\theta$ is a primitive element. Recall that the conductor of $K/F$ is defined as $\mathfrak{F} = \{ x \in \z_K : x \z_K \subseteq \z_F[\theta] \}$. It is well known, see e.g. \cite[Ch.~I, Prop.~8.3]{Neukirch:99} or \cite[Ch.~3, Thm.~27]{Marcus:77} or \cite[Prop.~2.3.9]{Coh2}, that at least the prime ideals that are coprime with $\mathfrak{F}$ can be easily lifted via a reduction to factoring the minimal polynomial of $\theta$ over a suitable finite field as described in following. Being coprime with the conductor is equivalent to being coprime with $|\z_K / \z_F[\theta]|$ which leaves only a finite set of primes for which it does not hold, i.e., these primes will be added to the exceptional set $\exPrimes$.

\begin{thm}\label{thm:lyingOver}
Let $K/F$ be a Galois extension where $K=F(\theta)$. Denote by $g(X) \in F[X]$ the minimal polynomial of $\theta$ and by $\mathfrak{F}$ the conductor of $K/F$. Let $\p$ be a prime ideal in $\z_F$ that is coprime with $\mathfrak{F}$. Let $\f = \z_F/\p$ be the finite field corresponding to the residues mod $\p$ and let
\[
\overline{g}(X) = \overline{g}_1(X)^{e_1} \cdots \overline{g}_r(X)^{e_r} \in \f[X]
\]
	be the factorization of $\overline{g}(X) = g(X) \; {\rm mod} \; \p$ into irreducible polynomials $\overline{g}_i(X) = g_i(X) \; {\rm mod} \; \p \in \f[X]$ where all $\overline{g}_i(X)$ are chosen to be monic polynomials. Then the ideals
\begin{equation}\label{eq:lyingOver}
\mathfrak{P}_i := \p \z_K + g_i(\theta) \z_K
\end{equation}
are precisely the prime ideals of $\z_K$ that are lying over $\p$ and all these ideals are pairwise different.
\end{thm}

We next analyze the time-complexity of computing the list of ideals $\mathfrak{P}_i$ lying over $\mathfrak{p}$. Factoring of a polynomial $f(X)$ of degree $n$ over a finite field $\mathbb{F}_q$ of size $q$ is known to run in time polynomial in $n$ and $\log(q)$ \cite[Theorem 14.14]{GG:2003}: 

\begin{thm}\label{thm:polyfactoring}
Let $\f_q$ be a finite field and let $f(X) \in \f_q[X]$ be a polynomial of degree $n$. Then there exists a probabilistic polynomial time algorithm that computes the factorization of $f(X) = \prod_{i=1}^k f_i(X)$ into irreducible polynomials over $\f_q$. The probability of success of the algorithm is at least $1/2$ and the expected running time can be bounded by $\widetilde{O}(n^2 \log{q})$.
\end{thm}

The algorithm in Theorem \ref{thm:polyfactoring} proceeds in $3$ stages, namely (i) squarefree factorization, (ii) equal degree factorization, and finally (iii) distinct degree factorization. As in our case $[K:F]=2$ we only have two possibilities of possible splitting behavior of $\overline{g}(X)$: either a) this polynomial is irreducible which according to Theorem~\ref{thm:lyingOver} corresponds to the case in which $\p$ is inert as $g_i(\theta)=0$ and $\p$ itself generates a prime ideal in $\z_K$ or b) the polynomial splits as $\overline{g}(X)=\overline{g}_1(X) \overline{g}_2(X)$ where both factors  are linear. This corresponds to the case where there are two ideals $\mathfrak{P}_1$ and $\mathfrak{P}_2$ lying over $\p$. On account of Galois theory in this case have that $\mathfrak{P}_2 = \mathfrak{P}_1^*$. 

Using Theorem \ref{thm:polyfactoring} we obtain a refined version of Theorem \ref{thm:lyingOver} that bounds the running-time of finding the ideals lying over a given prime ideal $\mathfrak{p}$ in terms of the bit complexity of $\mathfrak{p}$. 

\begin{cor}\label{cor:lyingOver}
Let $K/F$ be a CM field where $K=F(\theta)$. Let $\p \subseteq \z_F$ be a prime ideal and let $n = S(\p)$ be the bit-size of $\p$. Then there exists a polynomial time algorithm to compute all ideals $\mathfrak{P}$ lying over $\mathfrak{p}$. Furthermore, the bit-size of the $\mathfrak{P}$ is polynomial in $n$.
\end{cor}
\begin{proof}
Let $g(X) \in F[X]$ be the minimal polynomial of $\theta$ and let $\f_q := \O_F/\mathfrak{p}$ be the finite field that arises as the residue field of $\mathfrak{p}$. Using Theorem \ref{thm:polyfactoring} we see that we can compute the factorization $\overline{g}(X)=\prod_{i=1}^r \overline{g}_i(X)$ into irreducibles in polynomial time in $n$. For each factor $\overline{g}_i$ we can determine a corresponding lifting $g_i(X)$ by considering the components in of $\f_q$ as elements of $F$ while maintaining the same bit-size, i.e., we obtain that the element $g_i(\theta)$ has a bit-size that is polynomial in $n$ as well. Now, we construct the $2$-generator representation of $\mathfrak{P}$ as in eq.~(\ref{eq:lyingOver}). We next find an HNF representation of this ideal: as 
$\theta$ obeys a quadratic equation $F$ we can choose $B = \{1, \theta\}$ as a basis of $K/F$. Hence, we obtain a set of generators in the form
\[
H = \left(\begin{array}{ccccc}
A & 0 & a & b \\
0 & A & c & d
\end{array}
\right) \in \z^{2d\times 4d}
\]
where $n=2d$, $A$ is an HNF for $\mathfrak{p}$, and the matrix $\left(\begin{array}{cc} a & b \\ c & d \end{array}\right)$ is the expansion over $\z$ of the linear map that describes the multiplication by $g_i(\theta)$ with respect to $B$. Note that all coefficients of $H$ are bounded in bit-size by a polynomial in $n$. We now use Theorem \ref{thm:hnf} to compute an HNF $H^\prime$ for $H$, and hence the ideal $\mathfrak{P}$, in time polynomial in $n$. Theorem \ref{thm:hnf} also implies that the output size, i.e., all coefficients of $H^\prime$, are polynomial in $n$.
\end{proof}

\subsection{Outline of our algorithm to solve relative norm equations}
\label{sec:steps}

Let $K$ be a CM field over its totally real subfield $F := K^+$ and denote by $\{\sigma_1, \ldots, \sigma_d\}$ the real embeddings of $F$ into $\r$, i.e, we have that $[K:F]=2$ and $[F:\q]=d$. The extension $K/F$ is Galois and its Galois group is generated by complex conjugation, i.e., ${\rm Gal}(K/F)=\langle \tau \rangle$, where $\tau(x) := x^*$. Furthermore, denote by $\z_F$ and $\z_K$ the rings of integers in $K$ and $F$, respectively. Recall that for an ideal $I\subseteq \z_K$ the norm is defined as $N_{K/F}(I) := |\z_K/I|$ which for principal ideals $I=(x)$ coincides with the usual definition as the product of all Galois conjugates, i.e., $N_{K/F}((x)) = x \cdot x^*$.

For a given element $e \in \z_F$, our approach to solving the relative norm equation
$N_{K/F}(z)=e$ where $z \in \z_K$ relies on the observation that if $I \cdot I^* = e \z_F$ is a factorization of ideals in $\z_F$ and $\eta$ is an arbitrary non-zero element in the lattice generated by $I$, then we have that $N_{K/F}(I) | N_{K/F}(\eta)$. This alone would not be a very useful property as potentially the quantity on the right might be unbounded. Using the fact that $[K:F]=2$ is constant and that $K$ is a CM field we will however be able to show that for suitable $\eta$ the quotient $N_{K/F}(\eta) /N_{K/F}(I)$ will be a constant that just depends on the CM field $K$ and not on the right hand side $e$ of the norm equation (\ref{prob:NE}). Furthermore, we will show that we can find such $z$ in polynomial time, provided that $e$ is benign.

In the next two sections we show that there exists a probabilistic algorithm  that runs in polynomial time (with respect to the bit-size of the right hand side $e$) and finds an element $z$ such that $N_{K/F}\at z=e$ or else reports that no such element exists. We show this first for the case where $e$ generates a prime ideal and then in a subsection section for the case of general benign $e$. The algorithm proceedings in several stages:
\begin{itemize}
\item[Step 1] Determine whether the right hand side $e\z_{F}$ is prime, respectively benign.
\item[Step 2] If $e\z_{F}$ is prime, respectively benign, extract $\mathfrak{q}$ from the prime decomposition  eq.~(\ref{eq:benign}).
\item[Step 3] Compute a prime ideal $\mathfrak{Q}$ lying over $\mathfrak{q}$ using Theorem \ref{thm:lyingOver}.
\item[Step 4] For all ideals $\mathfrak{p_i}$ in $\exPrimes$, precompute all ideals $\mathfrak{P_{i,j}}$ lying over them. This can be done offline.
\item[Step 5] For all exponents ${\bf e}_{i,j}$, where the tuple ${\bf e}$ is taken from a bounded set of candidates, compute $I := \mathfrak{Q} \cdot {\mathfrak{P}_{1,1}^{{\bf e}_{1,1}}} \cdots \ldots \cdot {\mathfrak{P}_k}^{{\bf e}_{k,\ell_k}}$ and apply the following steps to all candidates. In the case of prime $e\z_F$ this step can be omitted and $I := \mathfrak{Q}$.
\item[Step 6] Compute LLL reduced basis for lattice $I$ to obtain an approximation $\eta$ to the SVP for $I$.
\item[Step 7] Compute the norm of $\eta$ and compute $\gamma := N_{K/F}(\eta)/N_{K/F}(I)$ and
attempt to solve the norm equation for $\gamma$ using a known method such as \cite{Simon:2002}. Let $w$ be such that $N_{K/F}(w)=\gamma$.
\item[Step 8] Output ``fail'' if no such solution exists. Otherwise, return a solution $z=\eta w\in \z_K$.
\end{itemize}

We will provide a proof that all steps can be performed by a classical algorithm whose runtime is polynomial in the bit-size of $e$, where in Section \ref{sec:NEprime} we show this for the somewhat simpler, however in practice frequently occurring, case where the right hand side $e\z_F$ generates a prime ideal and then sketch in Section \ref{sec:NEbenign} how the case of any benign $e$ can be handled. Before we can prove this, we need another technical result, namely that it is indeed possible to find an element $\eta$ as needed for Steps 6 and 7 such that the co-factor $\gamma := N_{K/F}(\eta)/N_{K/F}(I)$ is bounded.

\subsection{Bounding the co-factor}\label{sec:bounding}

Assuming that the norm equation $I \cdot I^* = e \z_K$ is solvable implies that $I=\xi \z_K$ for some element $\xi\in \z_K$. We now show that we can find an element $\eta \in \xi \z_K$ such that the quotient of the norms of $I$ and $\eta \z_K$ is a constant. We consider $I$ to be a $2d$-dimensional $\z$ lattice, where $d = [F:\q]$. This means that there exists a basis $\{a_1, \ldots, a_{2d}\} \subset K$ such that $I = a_1 \z + \ldots + a_{2d} \z$. Recall further that there is a quadratic form on $I$ defined by $(x,y) := {\rm Tr}_{K/\q}(x y^*)$ and that the Gram matrix $G_{i,j} := (a_i, a_j)$ is integer valued, i.e., $G \in \z^{2d \times 2d}$. Furthermore, for the volume of the fundamental parallelepiped of $I$, the identity ${\rm vol}(I) = \sqrt{\det(G)}$ holds \cite{MG:2002}. The fact that $I=\xi \z_K$ is principal is used in the following lemma to compute ${\rm vol}(I)$ in terms of the absolute norm of $\xi$:
\begin{lem}\label{lem:volume}
For each $I= \xi \z_K$, we have that ${\rm vol}(I) = {\rm vol}(\z_K) \cdot N_{K/\q}(\xi)$.
\end{lem}
\begin{proof}
We first choose a basis $\{b_1, \ldots, b_{2d}\}$ of $\z_k$ over $\z$, i.e., $\z_K = b_1 \z + \ldots b_{2d} \z$. With respect to this basis, multiplication with the fixed element $\xi$ is a linear transformation $M_\xi$ defined via $M_\xi (x_1, \ldots, x_{2d}) = \xi (b_1 x_1 + \ldots + b_{2d} x_{2d})$, and the determinant of $M_\xi$ is equal to the norm $N_{K/\q}(\xi)$. Note also, as $K$ is a CM field, all Galois automorphisms $\sigma \in {\rm Gal}(K/\q)$ come in complex conjugate pairs, i.e., $N_{K/\q}(\xi)\geq 0$, i.e., $N_{K/\q}(\xi) = |\det(M_\xi)|$.

By applying a base change to the Gram matrix $G$ in which we go from pairs of conjugates $\sigma_i, \overline{\sigma}_i$ to ${\rm Re}(\sigma_i), {\rm Im}(\sigma_i)$, where by convention we order the Galois automorphisms in such a way that the first $d$ are pair-wise non-conjugates under complex conjugations, i.e., $\overline{\sigma}_i \not= \sigma_j$ for all $1 \leq i < j \leq d$.

For general $x_i \in K$ we denote by $V(x_1, \ldots, x_{2d})$ the matrix
\[
V(x_1, \ldots, x_{2d}) = \left(
\begin{array}{ccc}
{\rm Re}(\sigma_1(x_1)) & \ldots & {\rm Re}(\sigma_1(x_{2d})) \\
{\rm Im}(\sigma_1(x_1)) & \ldots & {\rm Im}(\sigma_1(x_{2d})) \\
\vdots & \ddots & \vdots \\
{\rm Re}(\sigma_d(x_1)) & \ldots & {\rm Re}(\sigma_d(x_{2d})) \\
{\rm Im}(\sigma_d(x_1)) & \ldots & {\rm Im}(\sigma_d(x_{2d}))
\end{array}
\right).
\]

Using this matrix we can then express the volume of $I$ as ${\rm vol}(I) = \sqrt{\det(G)} = |\det(V(a_1, \ldots, a_{2d}))|$, where the set $\{a_i : i =1, \ldots, 2d\}$ forms a basis for $I$ over $\z$.

Next, we observe that the matrices $V_{I} := V(a_1, \ldots, a_{2d})$ and $V_{\z_K} := V(b_1, \ldots, b_{2n})$ are related via $V_I = M_\xi V_{\z_K}$. From this we conclude that
\[
{\rm vol}(I) = \sqrt{\det(G)} = |\det(V_I)|
= |\det(M_\xi)| \cdot |\det(V_{\z_K})|
= N_{K/\q}(\xi) \cdot {\rm vol}(\z_K),
\]
as claimed.
\end{proof}

We now show how to find an element $\eta \in \xi \z_K$ such that the quotient of the norms of $I$ and $\eta \z_K$ is a constant as mentioned in the beginning of this section.

\begin{lem}\label{lem:bounded}
Let $I=\xi \z_K$ be an ideal in $\z_K$ such that $I \cdot I^* = e \z_K$. Then there exists $\eta \in \z_K$ such that $N_{K/\q}(\xi/\eta)$ is upper bounded by a constant $C_{K}$ that depends just on the extension $K/\q$.
\end{lem}
\begin{proof}
As above, we consider $I$ as a lattice $I = a_1 \z + \ldots + a_{2d} \z$. We now use the LLL algorithm on the basis $\{a_1, \ldots, a_{2d}\} \subset \r^{2d}$. Using the LLL algorithm described in \cite{NS:2005} it is known that the first vector $v = {\bf b}_1$ in the LLL reduced basis $\{ {\bf b}_1, \ldots, {\bf b}_{2d} \}$ for $I$ satisfies the following bound:
\[
\| v \| \leq (4/3)^{(2d-1)/4} {\rm vol}(I)^{1/{2d}}.
\]
As $I=\xi \z_K$ is by assumption principal and $v \in I$, there exists an element $\eta \in \z_K$ such that $v = \eta \xi$. In order to finish the proof of the lemma, it remains to show that the norm of $\eta$ is upper bounded by a constant that just depends on $K$ alone and is in particular independent of the right hand side $e$ of the norm equation:

\begin{eqnarray*}
N_{K/\q}(\eta) & = & N_{K/\q}(v) \; / \; N_{K/\q}(\xi)
 =  \prod_{i=1}^{d} | \sigma_i(v)|^2 \; / \; N_{K/\q}(\xi) \\
& \stackrel{{\rm AGM}}{\leq} & \left( \frac{1}{d} \sum_{i=1}^{d} | \sigma_i(v)|^2 \right)^d / \; N_{K/\q}(\xi)
= \left( \frac{1}{2d} {\rm Tr}_{K/\q}(v v^*) \right)^d \; / \; N_{K/\q}(\xi) \\
& = & \left(\frac{1}{2d}\right)^{d} \| v \|^{2d} \; / \; N_{K/\q}(\xi)
\quad \stackrel{\rm {LLL}}{\leq}  \quad \left( \frac{1}{2d} \left(\frac{4}{3}\right)^{(2d-1)/2}\right)^d
{\rm vol}(I) / \; N_{K/\q}(\xi) \\
& \stackrel{{\rm Lemma}~\ref{lem:volume}}{=} &
\left( \frac{1}{2d} \left(\frac{4}{3}\right)^{(2d-1)/2}\right)^d
N_{K/\q}(\xi) \; {\rm vol}(\z_K) \; / \; N_{K/\q}(\xi)\\
& = &
\left( \frac{1}{2d} \left(\frac{4}{3}\right)^{(2d-1)/2}\right)^d {\rm vol}(\z_K) =: C_K
\end{eqnarray*}
where the first inequality is the arithmetic-geometric-mean inequality (AGM) and $v$ is the first basis vector obtained via LLL reduction for $\delta_{LLL} \equiv 1$ and $\eta_{LLL} \equiv 1/2$ as in \cite{NS:2005}.
\end{proof}

Using Lemma \ref{lem:bounded} will finally put us into a position where we can solve norm equations as in Problem \ref{prob:NE} efficiently, in case $e$ is a benign number as defined in Definition \ref{def:benign}. In Section \ref{sec:NEprime} establish this result first for the case where $e \z_F$ is a prime ideal itself as this case is relatively straightforward. We deal with the more general case of benign $e$ in Section \ref{sec:NEbenign}.

\subsection{Constructing the solution: prime case}\label{sec:NEprime}

We now provide a proof of the following theorem which we restate from Section \ref{sec:high:level:desc}.

\nix{
\begin{thm}
\label{thm:norm-eq} Given totally positive element $e$ of $F$ there exists an algorithm for testing if the instance of integral relative norm equation in $K/F$
\[
 zz^{\ast} = e , z \in \z_K
\]
can be solved in polynomial time in $\log T_2\at{e}$ (procedure IS-EASILY-SOLVABLE) If the test is passed, there exist another algorithm for deciding if the solution exists and finding it that runs in time polynomial in $\log T_2\at{e}$~(procedure FAST-SOLVE-NORM-EQ). Procedure IS-EASILY-SOLVABLE returns true at least for cases when the ideal $e\z_F$ is prime.
\end{thm}
}

\normeq*

We have now all ingredients in place in order to prove that the norm equations arising in the context of our approximation method can be solved in time that is polynomial in the input size. We first prove this for the case where the right hand side $e$ generates a prime ideal $\mathfrak{p}=e \z_F$ and leave the more general case of a benign integer $e$ for the next section.

\bigskip
\noindent
{\em Proof (of Theorem \ref{thm:norm-eq}).}
Let $n=S(e)$ be the bit-size of the ideal generated by $e$. We go through all Steps 1--7 described in Section \ref{sec:steps} and assert that all operations can be performed in time that is upper bounded by a polynomial in $n$.

In Step 1 we run a test which is described in subroutine IS-EASILY-SOLVABLE shown in Figure \ref{fig:proc-fast-solve}. The norm computations in lines \ref{ln:easy:norm1} and \ref{ln:easy:norm2} are clearly polynomial in $n$. Finally, for the primality test at the last line of the subroutine one can use for instance a probabilistic test such as Miller-Rabin \cite[Section~8.2]{coh3} or a deterministic test \cite{AKS:2004}.

\begin{figure}[hbt]
\protect\caption{\label{fig:proc-is-easily-solvable}Check if the right hand side is a benign integer in the sense of Definition \ref{prob:rnf-2}.}
\rule[0.5ex]{1\columnwidth}{1pt}
\begin{centering}
%auto-ignore
%!TEX root = quaf.tex
\begin{algorithmic}[1]
\fxinput $F$, $\z_{F}$, $P$
\Statex[1] $\z_{F}$ is the ring of integers of the totally real field $F$  and $P$ is a user defined parameter to
\Statex[1] select the set of exceptional primes $\exPrimes$. 
\inp $e$
\Statex[1] Check if the given element $e \in \z_{F}$ is benign with respect to a set $\exPrimes$ of user defined primes.
\Procedure{IS-EASILY-SOLVABLE}{}
\offline
\State $\exPrimes_0 \leftarrow {\rm AllBoundedPrimes}(\z_F, P)$ \Comment Enumerate all prime ideals $\mathfrak p$ with $N({\mathfrak p})\leq P$.
\State $\exPrimes_1 \leftarrow {\rm AllPrimeDivisors}(\Delta_F)$
\State $\exPrimes \leftarrow \exPrimes_0 \cup \exPrimes_1$ \Comment Yields finite set $\exPrimes = \{\mathfrak{p}_1, \ldots, \mathfrak{p}_k\}$ of exceptional primes.
\State \Call{FAST-SOLVE-NORM-EQ}{$K,F,\z_F,\z_K,\CP$} \Comment $K = F(\sqrt{b})$, see Section~\ref{sec:basics}
\online
\State $u \leftarrow N_{F/\q}(e)$ \label{ln:easy:norm1}
\While{exists $\mathfrak{p}\in \exPrimes$ such that $N_{F/\q}(\mathfrak{p}) \big| u$} \label{ln:easy:norm2}
\State $u \leftarrow u/N_{F/\q}(\mathfrak{p})$
\EndWhile
\EndProcedure
\out ${\rm IsPrime}(u)$  \label{ln:easy:isprime}
\Statex Offline precomputation of exceptional primes $\exPrimes$ and online test whether $e$ is benign via trial divisions by elements of $\exPrimes$ and a final primality test in $\z$.%
\end{algorithmic}
\par\end{centering}
\rule[0.5ex]{1\columnwidth}{1pt}
\end{figure}

Step 2 can be done by computing quotients of the form $(e) (\mathfrak{p}_i)^{-1}$ which can be done in polynomial time in the input bit-size and at an increase per division that is also at most polynomial \cite{BF:2012}. Eventually this yields the prime ideal $\mathfrak{q}$. 

This step is also done in subroutine IS-EASILY-SOLVABLE. All subsequent steps and line numbers refer to subroutine FAST-SOLVE-NORM-EQ shown in Figure \ref{fig:proc-fast-solve}. 

For Step 3, in line \ref{ln:norm-eq:hnf} of subroutine FAST-SOLVE-NORM-EQ, we use Theorem \ref{thm:lyingOver} and the complexity analysis given in Corollary  \ref{cor:lyingOver} in order to compute an HNF for the ideal $\mathfrak{Q}$ lying over $\mathfrak{q}$ in polynomial time and almost polynomial increase of the bit-size.

Steps 4 in line \ref{ln:norm-eq:offline} is an offline computation which does not count toward the cost of the online solution of the norm equation.

Step 5 does not have to be carried out as by assumption in this subsection we assume that $\mathfrak{q}$ is prime, i.e., there is only one prime ideal $\mathfrak{Q}$ that we have to consider. We will discuss this step and the consequences for the subsequent steps in case $e$ is benign but not prime in the next subsection.

Step 6 in line \ref{ln:norm-eq:lll} involves the computation of a reduced lattice basis for the ideal corresponding to $\mathfrak{Q}$ from the HNF that was computed in Step 3. Using bounds on the complexity of the LLL algorithm \cite{NS:2009} we see that the running time of this step is polynomial in the input size $n$ and so is the bit-size of the short vector $\eta$ that is produced by this computation.

For Step 7 in line \ref{ln:norm-eq:constant} we use a known method for solving norm equations such as Simon's $S$-unit based algorithm \cite{Simon:2002} for which an implementation e.g. in Magma \cite{MAGMA:1997} is available. As the element $\gamma$ is constant and does not depend on the input size $n$ we can assume that this computation can be done in constant time that does not affect our overall running time.

Finally, in Step 8 we catch the case where there is no solution for $\gamma$, which then implies that there is no solution for $e$ and combine $\eta$ and $w$ using an ideal multiplication into the final solution $z$ to the norm equation  $N_{K/F}(z)=e$. This again can be done in polynomial time.

\begin{figure}[hbt]
\protect\caption{\label{fig:proc-fast-solve} Compute an integer solution $x\in \z_{K}$ to the norm equation $N_{K/F}(z)=e$, provided it exists.}
\rule[0.5ex]{1\columnwidth}{1pt}
\begin{centering}
%auto-ignore
%!TEX root = quaf.tex
\begin{algorithmic}[1]
\fxinput $K$, $F$, $\z_{F}$, $\z_{K}$, $\exPrimes$
\Statex[1]  $\z_{F}$ are the integers of the totally real subfield $F$ of a CM field $K$, $\z_K$ are the integers of $K$ 
\Statex[1] and $\exPrimes=\{\mathfrak{p}_1, \ldots, \mathfrak{p}_k\}$ is a finite set of exceptional primes.
\inp $e$
\Statex[1] Find integer solution $z\in \z_{K}$ to $N_{K/F}(z)=e$, where $e\in\z_F$, provided it exists.
\Procedure{FAST-SOLVE-NORM-EQ}{}
\offline
\For{$i=1,\ldots,k$}
\State $\mathfrak{P}_i \leftarrow$ PrimesLyingOver$(\mathfrak{p}_i)$ \label{ln:norm-eq:offline}
\EndFor
\State $\theta \leftarrow$ PrimitiveElement$(K/F)$
\State $f(X) \leftarrow$ MinimalPolynomial$(\theta)$
\spec  Precomputed primes lying over the primes in $\exPrimes$, the minimal polynomial of a primitive element $\theta$.
\online
\State \assert IS-EASY-SOLVABLE($e$)
\State \Comment Compute prime ideal $\mathfrak{q}$ if test is passed successfully.
\State \assert PASS-SOLVABILITY-TEST($e$)
\State \Comment Check existence of a global rational solution via the Hasse principle.
\State $g(X) \leftarrow$ Lift(ModularFactorization($f(X),\z_F/\mathfrak{q})$ \Comment Obtain a lifted factor $g(X)\in \z_{F}[X]$.
\State $I \rightarrow {\rm HNF}(\mathfrak{q} \z_{K}, g(\theta))$ \label{ln:norm-eq:hnf} 
\State $\eta \leftarrow {\rm FirstBasisVector}({\rm LLL}(I))$ \label{ln:norm-eq:lll}
\State $\gamma \leftarrow e (N_{K/F}(\eta))^{-1}$
\State $w \leftarrow$ SolveConstantNormEquation$(\gamma)$ \label{ln:norm-eq:constant} 
\State \Comment Use any method for solving relative norm equations $N_{K/F}(w)=\gamma$ for the remaining bounded size element $\gamma \in \z_{F}$. If this fails, repeat the steps following line \ref{ln:norm-eq:hnf} with $I \leftarrow I \cdot {\mathfrak P_i}$ for $i=1,\ldots,k$ until success in line \ref{ln:norm-eq:constant}, otherwise report ``failure.'' 
\EndProcedure
\out Result $z = w\eta$ or ``failure'' if no solution $w$ exists.
\Statex Reduced the computation of a solution to the norm equation to a finite problem that is independent of the input size.
\end{algorithmic}
\par\end{centering}
\rule[0.5ex]{1\columnwidth}{1pt}

\end{figure}

A pseudocode description of the Steps 3--8 is given in subroutine FAST-SOLVE-NORM-EQ shown in Figure \ref{fig:proc-fast-solve}. From the above discussion of the steps it follows that the overall runtime is polynomial in the input size $n=S(e)$ and by Lemma \ref{lem:bitsizeNorm} therefore also polynomial in the bit-size $\log{T_2(e)}$ as claimed. \hfill $\Box_{prime}$

\subsection{Constructing the solution: general case}\label{sec:NEbenign}

We briefly discuss the implications of $e$ being benign but not prime. This will involve a change in Step 5, i.e., instead of only considering the prime ideals $\mathfrak{Q}$ lying over $\mathfrak{q}$ we consider all ideals $I$ that can be formed by multiplying $\mathfrak{Q}$ with the ideals lying over the various prime factors $\mathfrak{q}_i \in \exPrimes$. All ideals lying over $\exPrimes$ can be precomputed without any additional cost to the online part. Also, we note that even though the number of ideals to be considered grows significantly, this increase is still just a constant as for any given input parameter $P$ of FAST-SOLVE-NORM-EQ in Figure \ref{fig:proc-fast-solve} the size of $\exPrimes = \exPrimes_0 \cup \exPrimes_1$ is constant. The prime ideals in $\exPrimes_0$ can be found by enumeration and those in $\exPrimes_1$ can be found by factoring the discriminant $\Delta_F$ which, as argued before, can be done in polynomial time as $F$ is constant.  

In order to compute $I$ we now multiply $\mathfrak{Q}$ with all possible combinations of factors that are extensions $\mathfrak{P}$ of ideals in $\exPrimes$. Each such multiplication can be done in polynomial time with polynomial increase in bit-size~\cite{BF:2012}. Then we perform each of the subsequent Steps 6--7 as in the prime case and report success if any of the considered cases leads to a successful solution $z$ and ``failure'' otherwise. Overall, also in case of benign right hand sides $e$ we obtain a polynomial time classical algorithm to solve the norm equation $N_{K/F}(z)=e$. \hfill $\Box_{benign}$

\subsection{Performance improvement: filtering out candidates via the Hasse principle}

It is possible to perform simple tests whether a solution to $N_{K/F}(z) = e$ over the {\em rational} elements of $K$ (i.e., not necessarily elements in $\z_K$) exists. A famous test in this regard is the Hasse Norm Theorem \cite{Neukirch:99} that asserts that a global solution, i.e., a solution over $K$, exists if and only if a solution exists with respect to all local fields associated with $K/F$. More precisely, we have:

\begin{thm}[Hasse Norm Theorem]
Let $K/F$ be a cyclic extension. An element $e \in F^\times$ is a norm of an element in $K^\times$ if and only it is a norm at every prime of $F$, including the infinite primes.
\end{thm}

In practice, it is not necessary to check all primes of $F$, a finite set of primes is sufficient: as described in \cite{AK:2000} the only primes that need to be checked are a) the divisors of the conductor $\mathfrak{F}$ of $K/F$ and b) all finite primes dividing the ideal $e \z_F$. If $e$ is benign, we therefore can efficiently compute the prime factorization and hence can perform this sufficient test for solvability of the norm equation. Note that this test can only be used in this one-sided sense as there are examples of degree $2$ extensions $K/F$ known where for $e \in \z_F$ the equation $N_{K/F}(z) = e$ is solvable over $K$ but not over $\z_K$. In practice, the test is reasonably fast in order to eliminate some candidates $e$. We summarize the pseudo-code for this test PASS-HASSE-SOLVABILITY-TEST in Figure \ref{fig:proc-pass-solvability-test}.

\begin{figure}[hbt]
\rule[0.5ex]{1\columnwidth}{1pt}
\begin{centering}
%auto-ignore
\def\exS{{\mathcal S}}
\begin{algorithmic}[1]
\fxinput $K$, $F$, $\z_{F}$
\Statex[1] $\z_{F}$ is the ring of integers of the totally real subfield $F$ of a CM field $K$.
\inp $e$
\Statex[1] $e \in \z_{F}$ is an element for which we check if the norm equation $N_{K/F}(z)=e$ has a solution $z\in K$. This is a quick check to rule out existence of solutions for some $e$.
\Procedure{PASS-HASSE-SOLVABLILTY-TEST}{}
\offline
\State $\mathfrak{F} \leftarrow {\rm Conductor}(K/F)$
\State $\exS_0 \leftarrow {\rm Factorization}(\mathfrak{F})$
\State $\exS_\infty \leftarrow {\rm InfinitePlaces}(K)$
\online
\State $\exS_1 \leftarrow {\rm PrimeFactorization}(e\z_F)$
\State $\exS = \exS_0 \cup \exS_1 \cup \exS_\infty$
\State isSolvable $\leftarrow$ TRUE
\For{$\mathfrak{p} \in \exS$}
\If{\Not {\rm IsLocallySolvable}($K_{\mathfrak{P}}, F_{\mathfrak{p}},e)$} \Comment $\mathfrak{P}$ is a prime in $\z_{K}$ lying over $\mathfrak{p}$.
\State isSolvable $\leftarrow$ FALSE
\EndIf
\EndFor
\EndProcedure
\out isSolvable
\Statex Based on the Hasse local-global principle, solvability of the norm equation over $K$ is checked via local solvability at all finite places dividing the conductor, divisors of $e$ and all infinite places.
\end{algorithmic}
\par\end{centering}
\rule[0.5ex]{1\columnwidth}{1pt}
\protect\caption{\label{fig:proc-pass-solvability-test}Check if the right hand passes a solvability test whether a {\em rational} solution $z$ to the norm equation exists.}
\end{figure}

\clearpage

%auto-ignore
%!TEX root = quaf.tex

\section{End to end examples of using the framework}\label{sec:examples}

In this section we provide four examples of using our framework. In two examples we apply our framework to reproduce results on Clifford+$T$~\cite{S} and $V$-basis\cite{BGS} gate sets. The exact synthesis framework~\cite{KY1} is necessary for the end to end compilation. For this reason we discuss how to use this framework to reproduce previously known results on exact synthesis from~\cite{BGS},\cite{KMM1},\cite{GKMR} and also results on exact synthesis over Clifford+$R_z\at{\pi/16}$ \cite{SGKM}. Two other examples, corresponds to approximating using gate sets Clifford+$R_z\at{\pi/16}$ and Clifford+$T$+$V$. No number-theoretic style approximation algorithms for this gate sets were known before. A brief summary of the examples is given in Table~\ref{tab:summary}. Amongst other data, the table contains time needed for precomputation stage for all of our examples and the value of the additive constant appearing in Theorem~\ref{thm:main}. All the data about the algorithm performance is based on the implementation of exact synthesis\cite{KY1} and approximation frameworks using computer algebra system MAGMA\cite{MAGMA:1997}. The total number of lines of code needed for it implementation is about $2500$. More examples of running approximation stage of the algorithm are provided in Appendix A on page~\pageref{sec:appendix}. 

% \begin{itemize}
% \item Reproduce \cite{S} for Clifford+T and V gates
% \item The story for Clifford+ $R_z(\pi/16)$
% \item The story for Clifford+T+V
% \end{itemize}

\begin{table}

\caption{ \label{tab:summary} Summary of the results of running the algorithm for different gate sets. The quaternion algebra corresponding to the gate set is $Q=\at{\frac{a,b}{F}}$ where $F = \q\at{\zeta_n} \cap \r$; $N_G$ is the number of generators to be used for exact synthesis; $\of{L:L\cap M}$ is the index of the intersection of $L$ (the generalized Lipschitz order ) and $\mo$ (the maximal order used to define the gate set) in $L$; $N_{cnj}$ is the number of conjugacy classes of maximal orders in quaternion algebra $Q$; $t_{exact}$ is the time in seconds spent on precomputation required for the exact synthesis part of our algorithm; $t_{approx}$ is the time in seconds spent on precomputation required for the approximation part; $C_{\min}$ is the additive constant appearing in Theorem~\ref{thm:main}; $C_{\min} / \log N\at{\p_1}$ is the ratio between $C_{\min}$ and the $\log$ of the norm of the ideal from $S$ with the smallest norm~(see~Definition~\ref{dfn:qag}) }

\begin{centering}
\includegraphics{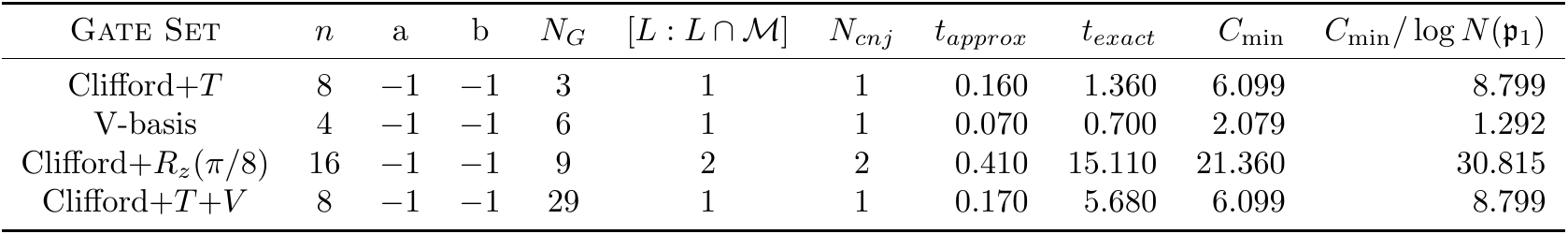}
\end{centering}

\end{table}

\subsection{Clifford+\texorpdfstring{$T$}{T}}
In this section we describe how to obtain results from \cite{S} within our framework. We also discuss the exact synthesis part using the framewok introduced in \cite{KY1}. We follow Section~\ref{sec:esmethods} and show that Clifford$+T$ can be described by totally definite quaternion algebra. We recall the following definition:

\qag*

For Clifford$+T$ gate set we can choose \cite{GKMR}:
\[
\CG=\set{R_\alpha\at{\pi/4},R_\alpha\at{\pi/2} : \alpha = x,y,z }
\]

We separately write $R_\alpha\at{\pi/2}$ because they generate Clifford group. Clifford gates are much cheaper in practice and typical cost function for Clifford+$T$ gate set used in practice is:
\[
cost\at{R_\alpha\at{\pi/4}} = 1, cost\at{R_\alpha\at{\pi/2}} = 0, \alpha = x,y,z.
\]

Let us provide quaternion gate set specification for Clifford$+T$
\begin{itemize}
\item $F=\q\at{\zeta_8+\zeta^{-1}_8}$ where $\zeta_8 = e^{2i\pi/8}$, let also $\theta$ be a primitive element of $F$ (in other words every element of $F$ can be represented as $a_0+a_1\theta$ where $a_0, a_1$ are rational numbers),
\item embedding $\sigma : F \rightarrow \r $ is defined as $\sigma\at{\theta}=\sqrt{2}$,
\item $a=-1$ and $b=-1$,
\item maximal order $\mo$ of quaternion algebra $Q=\at{\frac{-1,-1}{F}}$ is
\[
\z_F + \frac{\at{\i+1}\theta}{2}\z_F + \frac{\at{\j+1}\theta}{2}\z_F + \frac{1+\i+\j+\k}{2}\z_F,
\]
where $\z_F = \z\of{\sqrt{2}}$ is a ring of integers of $F$,
\item $S=\set{\p}$ where $\p = \at{2 - \theta}\z_F$~(note that $2 - \theta $ is totally positive element of $F$).
\end{itemize}
Note that the discriminant of $\mo$ is equal to $\z_F$, therefore $\p$ is coprime to it. This implies that the set $\mo_S$ is infinite.

Using notation $q_z = \i, q_y = \j, q_x = \k $ we obtain set $\CG_Q$ based on the following correspondence
\[
\begin{array}{ccccccc}
 q_{t,\alpha} & = & 1+\theta(1-q_\alpha)/2 & & U_q\at{q_{t,\alpha}} & = & R_\alpha\at{\pi/4} \\
 q_{c,\alpha} & = & \theta(1-q_\alpha)/2 & & U_q\at{q_{c,\alpha}} & = & R_\alpha\at{\pi/2}
\end{array}
\]
where $\alpha=x,y,z$.

The next step is to compute $\CG^{\star}_{\mo,S}$ using the algorithm from \cite{KY1}. We find that quaternion algebra $Q$ has trivial two sided ideal class group and that the number of conjugacy classes of maximal orders of $Q$ is one. In this case the situation is extremely simple. The set $\CG^{\star}_{\mo,S} $ is equal to $\mathrm{gen}_S\at{\mo} \cup \mathrm{gen}_u\at{\mo}$. The set $\mathrm{gen}_S\at{\mo}$ consists of $N\at{\p}+1 = 3$ elements with reduced norm $a_0+a_1\theta$. The set $\mathrm{gen}_u\at{\mo}$ consists of three generators of the finite group of units of maximal order $\mo$ modulo units of $\z_F$. The results of our computations are the following:
\begin{itemize}
\item $\mathrm{gen}_S\at{\mo}=\set{q_1,q_2,q_3}$
where
\[
q_1 = 1/2 + \i/2 + (\theta - 1)\j/2 + (\theta - 1)\k/2 , q_2 = (-\theta + 2)/2 - \theta\j/2,
q_3 = (-\theta + 2)\j/2 + \theta\k/2.
\]
\item $\mathrm{gen}_u\at{\mo}=\set{u_1,u_2,u_3}$ where
\[
u_1 =  1/2 - \i/2 - \j/2 - \k/2, u_2 = -\j, u_3 = \theta/2 - \theta\j/2.
\]
\end{itemize}
Algorithmically, we find that $u_1 = q_{c,z}q_{c,x}$, $u_2=q^2_{c,y}$ and $u_3=q_{c,y}$. We see that the unit group of $\mo$ modulo units of $\z_F$ corresponds to Clifford group. Next we find that $q_1=q_{t,x}q_{c,z}^2 q_{c,x} q_{c,z}$, $q_2=-q_{t,z}q_{c,z}^2 q_{c,x} q_{c,z}$ and $q_3 = q_{t,y}q_{c,y}^3$ up to a unit of $\z_F$. In general, the elements of set $\mathrm{gen}_S\at{\mo}$ are defined up to right-hand side multiplication by a unit of $\mo$. For this reason we can simply choose $\mathrm{gen}_S\at{\mo} = \set{q_{t,\alpha} : \alpha= x,y,z}$. The map Circuit becomes almost trivial in this case. In the next examples we will omit this detail and write  $\mathrm{gen}_S\at{\mo}$ using generators convenient for our application.

Note that $\nrd q_{t,\alpha} \z_F = \p $ therefore the cost vector corresponding to each $q_{t,\alpha}$ is $(1)$. For all elements of the unit group the cost vector is $(0)$. In this example original cost definition completely matches cost obtained based on cost vectors. Table~\ref{tab:clifford-t}, \ref{tab:clifford-pi16-2} shows the results of running our ciruit synthesis algorithm for Clifford+$T$ gate set.

\begin{table}

\caption{\label{tab:clifford-t} Results of running the algorithm for Clifford$+T$. Approximation of rotation $R_z\at{0.1}$ with precision $\ve$ and cost vector $\at{L_1}$. All columns except $N_{tr,\min}$ and $N_{tr,\max}$ are averages over 1000 runs of the algorithm; $L'_1$ is the averaged $T$-count of the found circuits; $\rho\at{U_q,R_z\at{\phi}}$ is the obtained quality of approximation; $N_{tr,\min}, N_{tr,\max}, N_{tr,avg}$ are minimum, maximum and average of the number of the main loop iterations in the procedure APPROXIMATE over all samples; $t_{exact}$ and $t_{approx}$ is time in seconds spent on approximation and exact synthesis stages of the algorithm. }

\begin{centering}
\includegraphics{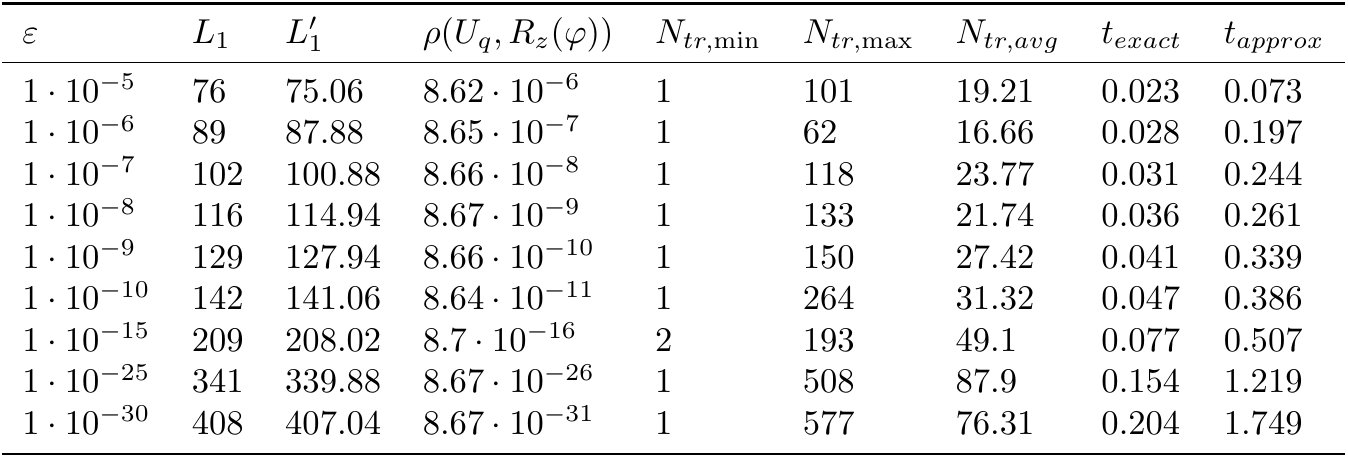}
\end{centering}

\end{table}

% This what MAGMA told us :
% Sanity check: user provided generators and their images
% [ 1/2*th - 1/2*th*i, 1/2*th - 1/2*th*j, 1/2*th - 1/2*th*k, -1 ]
% [ $.1, $.2, $.3, $.4 ]
% Computed unit group generators
% [ 1/2 - 1/2*i - 1/2*j - 1/2*k, -j, 1/2*th - 1/2*th*j ]
% Matching to user provided units:
% [* $.1 * $.3, $.2^2, $.2 *]
% User provide primitive generators:
% [ 1/2*(th + 2) - 1/2*th*i, 1/2*(th + 2) - 1/2*th*j, 1/2*(th + 2) - 1/2*th*k ]
% Primitive generators:
% [ 1/2 + 1/2*i + 1/2*(th - 1)*j + 1/2*(th - 1)*k, -1/2 + 1/2*(-th + 1)*i - 1/2*j
% + 1/2*(-th + 1)*k, 1/2*th + 1/2*(-th + 2)*j ]
% Matching to user provided generators:
% [
%     <1/2 + 1/2*i + 1/2*(th - 1)*j + 1/2*(th - 1)*k, [ 3 ], $.1^2 * $.3 * $.1>,
%     <-1/2 + 1/2*(-th + 1)*i - 1/2*j + 1/2*(-th + 1)*k, [ 1 ], $.4 * $.1^2 * $.3
%     * $.1>,
%     <1/2*th + 1/2*(-th + 2)*j, [ 2 ], $.2^3>
% ]

% OK

\subsection{\texorpdfstring{$V$}{V}-basis}
In this section we describe how to obtain results from \cite{BGS} within our framework.

$V$-basis is defined using the following set
\[
G=\set{\frac{1\pm2iP}{\sqrt{5}},iP : P \in \set{X,Y,Z} }
\]

A typical cost function is:
\[
cost\at{\frac{1\pm2iP}{\sqrt{5}}} = 1, cost\at{iP} = 0, \text{ for } P \in \set{X,Y,Z}
\]

The quaternion gate set specification is:
\begin{itemize}
\item $F=\q$,
\item embedding $\sigma : F \rightarrow \r $ is the only embedding of $\q$ into $\r$,
\item $a=-1$ and $b=-1$,
\item maximal order $\mo$ of quaternion algebra $Q=\at{\frac{-1,-1}{\q}}$ is
$
\z + \i \z + \j \z + \frac{1+\i+\j+\k}{2}\z,
$
\item $S=\set{\p}$ where $\p = 5\z$.
\end{itemize}
Note that the discriminant of $\mo$ is equal to $2\z$ and ideal $\p$ is co-prime to it. This implies that the set $\mo_S$ is infinite.

Using notation $q_Z = \i, q_Y = \j, q_X = \k $ we obtain set $\CG_Q$ based on the following correspondence:
\[
\begin{array}{ccccccc}
 q_{V,\pm P} & = & 1 \pm 2q_P & & U_q\at{q_{V,\pm P}} & = & \at{1\pm2iP} / \sqrt{5} \\
   &  &  & & U_q\at{q_P} & = & iP
\end{array}
\]
where $P \in \set{X,Y,Z}$.

Similarly to the previous section, we compute $\CG^{\star}_{\mo,S}$ using the algorithm from \cite{KY1}. We again find that quaternion algebra $Q$ has trivial two sided ideal class group and that the number of conjugacy classes of maximal orders of $Q$ is one. The set $\CG^{\star}_{\mo,S} $ is equal to $\mathrm{gen}_S\at{\mo} \cup \mathrm{gen}_u\at{\mo}$. The set $\mathrm{gen}_S\at{\mo}$ consists of $N\at{\p}+1 = 6$ elements with reduced norm $5$. The set $\mathrm{gen}_u\at{\mo}$ consists of two generators of the finite group of units of maximal order $\mo$ modulo units of $\z$. The results of our computations are the following:
\begin{itemize}
\item $\mathrm{gen}_S\at{\mo}=\set{q_{V,\pm P} : P \in \set{X,Y,Z} }$
\item $\mathrm{gen}_u\at{\mo}=\set{\i,\at{\i+\j+\k+1}/2}$
\end{itemize}

We got one of the generators of the unit group that cannot be expressed as a product of elements of $\CG_Q$. Indeed, all elements of $\CG_Q$ belong to the Lipschitz order
\[
 L = \z + \z \i + \z \j + \z \k
\]
and $\at{\i+\j+\k+1}/2$ does not. However, our approximation algorithm finds $q$ from $L$. It is possible to show that (in this particular example) the unit of $\mo$ obtained in the end of exact synthesis of $q$ must belong to $L$ and therefore belongs to the subgroup of the unit group of $\mo$ that contained in $L$. After a simple computation we find that this subgroup is generated by $\i,\j,\k$.

\begin{table}[t!]

\caption{\label{tab:v-basis}  Results of running the algorithm for V basis. Approximation of rotation $R_z\at{0.1}$ with precision $\ve$ and cost vector $\at{L_1}$. All columns except $N_{tr,\min}$ and $N_{tr,\max}$ are averages over 1000 runs of the algorithm; $L'_1$ is the average $V$-count of the found circuits; $\rho\at{U_q,R_z\at{\phi}}$ is the obtained quality of approximation; $N_{tr,\min}, N_{tr,\max}, N_{tr,avg}$ are minimum, maximum and average of the number of the main loop iterations in the procedure APPROXIMATE over all samples; $t_{exact}$ and $t_{approx}$ is time in seconds spent on approximation and exact synthesis stages of the algorithm. }

\begin{centering}
\includegraphics{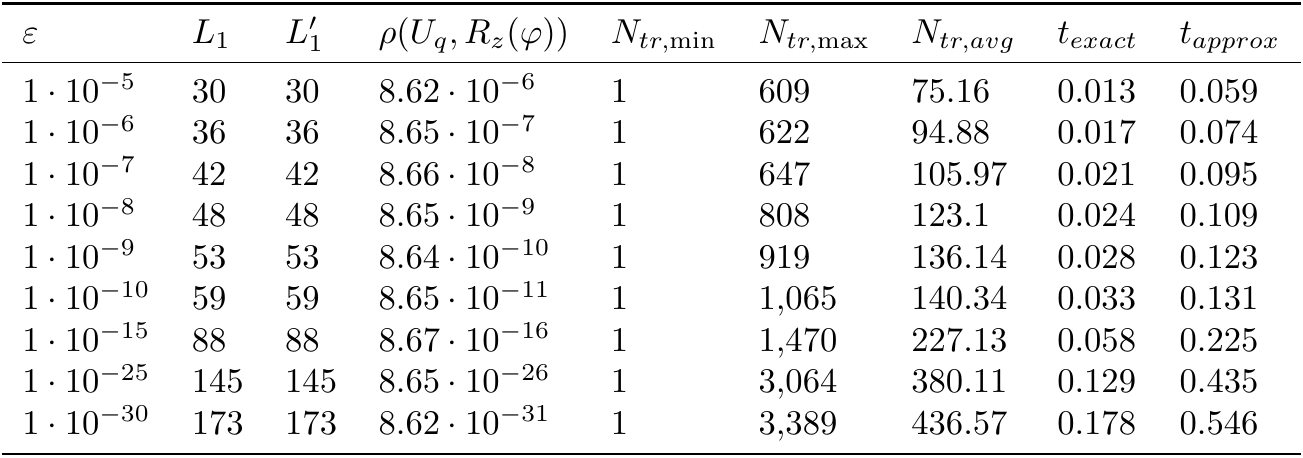}
\end{centering}

\end{table}

Note that $\nrd q_{V,\pm P} \z = \p $ and  the cost vector corresponding to each $ q_{V,\pm P}$ is $(1)$. For all elements of the unit group the cost vector is $(0)$. Similarly to Clifford+$T$ case original cost definition completely matches cost obtained based on cost vectors. Table~\ref{tab:v-basis}, \ref{tab:clifford-pi16-2} show the results of running our ciruit synthesis algorithm for $V$-basis.

\subsection{Clifford+\texorpdfstring{$R_z\at{\frac{\pi}{8}}$}{Rz(pi/8)} } The approximation part of the result for this gate set is new. An exact synthesis algorithm for this gate set was first described in~\cite{SGKM} in the language of $SO(3)$ representation of unitary matrices over the ring $\z\of{\zeta_{16},1/2}$. It can be shown that the output of the approximation stage of our algorithm can be converted to a unitary matrix over $\z\of{\zeta_{16},1/2}$. Therefore the algorithm developed in~\cite{SGKM} can be applied instead of the exact synthesis algorithm for quaternions we use here.

For Clifford+$R_z\at{\frac{\pi}{8}}$ gate set we can choose \cite{SGKM}:
\[
\CG=\set{R_\alpha\at{\pm\pi/8},R_\alpha\at{\pm 3\pi/8}, R_\alpha\at{\pi/4}, R_\alpha\at{\pi/2} : \alpha = x,y,z }
\]

For this example, the quaternion gate set specification is:
\begin{itemize}
\item $F=\q\at{\zeta_{16}+\zeta^{-1}_{16}}$ where $\zeta_{16} = e^{2i\pi/16}$, let also $\theta$ be a primitive element of $F$ (in other words every element of $F$ can be represented as $a_0+\ldots+a_3\theta$ where $a_0,\ldots,a_3$ are rational numbers),
\item embedding $\sigma : F \rightarrow \r $ is defined as $\sigma\at{\theta}=2\cos\at{2\pi/16}$,
\item $a=-1$ and $b=-1$,
\item maximal order $\mo$ of quaternion algebra $Q=\at{\frac{-1,-1}{F}}$ is
\[
\z_F + \frac{\at{\i+1}\xi}{2}\z_F + \frac{\at{\j+1}\xi}{2}\z_F + \frac{1+\i+\j+\k}{2}\z_F,
\]
where $\xi=\theta^2 - 2$, $\s\at{\xi}=\sqrt{2}$ and $\z_F = \z\of{2\cos\at{2\pi/16}}$ is a ring of integers of $F$,
\item $S=\set{\p}$ where $\p = \theta\z_F = \at{-\theta^3+4\theta^2+\theta-2}\z_F$~(note that $-\theta^3+4\theta^2+\theta-2$ is totally positive element of $F$). The 
    %maximal order $\mo$ is essentially the same as one 
    The definition of maximal order $\mo$ has essentially the same shape as the definition
    used for Clifford+$T$ case. The only difference is that it defined using different ring of integers.
\end{itemize}

The discriminant of $\mo$ is equal to $\z_F$, therefore $\p$ is coprime to it. This implies that the set $\mo_S$ is infinite.

Using notation $q_z = \i$, $q_y = \j$, $q_x =\k $, $\eta=\theta^3 - 3\theta$~(note $\s\at{\eta}=2\sin\at{\pi/8}$)  we obtain set $\CG_Q$ based on the following correspondence:
\[
\begin{array}{lccclcc}
 q_{t,\alpha} & = & 1+\xi(1-q_\alpha)/2 & & U_q\at{q_{t,\alpha}} & = & R_\alpha\at{\pi/4} \\
 q_{c,\alpha} & = & (1-q_\alpha)/2 & & U_q\at{q_{c,\alpha}} & = & R_\alpha\at{\pi/2} \\
 q_{1/8,\alpha} & = &\theta\at{1+(\theta-\eta q_\alpha)/2} & & U_q\at{q_{1/8,\alpha}} & = & R_\alpha\at{\pi/8} \\
 q_{3/8,\alpha} & = &\theta\at{1+(\eta-\theta q_\alpha)/2} & & U_q\at{q_{3/8,\alpha}} & = & R_\alpha\at{3\pi/8} \\
\end{array}
\]
where $\alpha\in \set{x,y,z}$.

% xi: th^2 - 2
% sn: th^3 - 3*th
% c3: th^3 - 3*th
% s3: th
% GQT := [ <1+xi*(1-i)/2,Rz(Pi(R)/4),"t,z">,
% 				      <1+xi*(1-j)/2,Ry(Pi(R)/4),"t,y">,
% 				      <1+xi*(1-k)/2,Rx(Pi(R)/4), "t,x"> ];

% GQ := [ <th*(1+(th-sn*i)/2),Rz(Pi(R)/8),"1/8,z">,
% 				      <th*(1+(th-sn*j)/2),Ry(Pi(R)/8),"1/8,y">,
% 				      <th*(1+(th-sn*k)/2),Rx(Pi(R)/8),"1/8,x">,
% 		<th*(1+(c3 - s3*i)/2),Rz(3*Pi(R)/8),"3/8,z">,
% 				      <th*(1+(c3 - s3*j)/2),Ry(3*Pi(R)/8),"3/8,y">,
% 				      <th*(1+(c3 - s3*k)/2),Rx(3*Pi(R)/8),"3/8,x>"> ];

\begin{figure}[!ht]

\caption{\label{fig:tree-pi16} Ideal principality graph \cite{KY1} used to find generators for exact synthesis algorithm. Tree corresponds to quaternion gate set specification for Clifford+$R_z(\pi/8)$ gate set.}

\rule[0.5ex]{1\columnwidth}{1pt}

\begin{centering}
\includegraphics[scale=0.5]{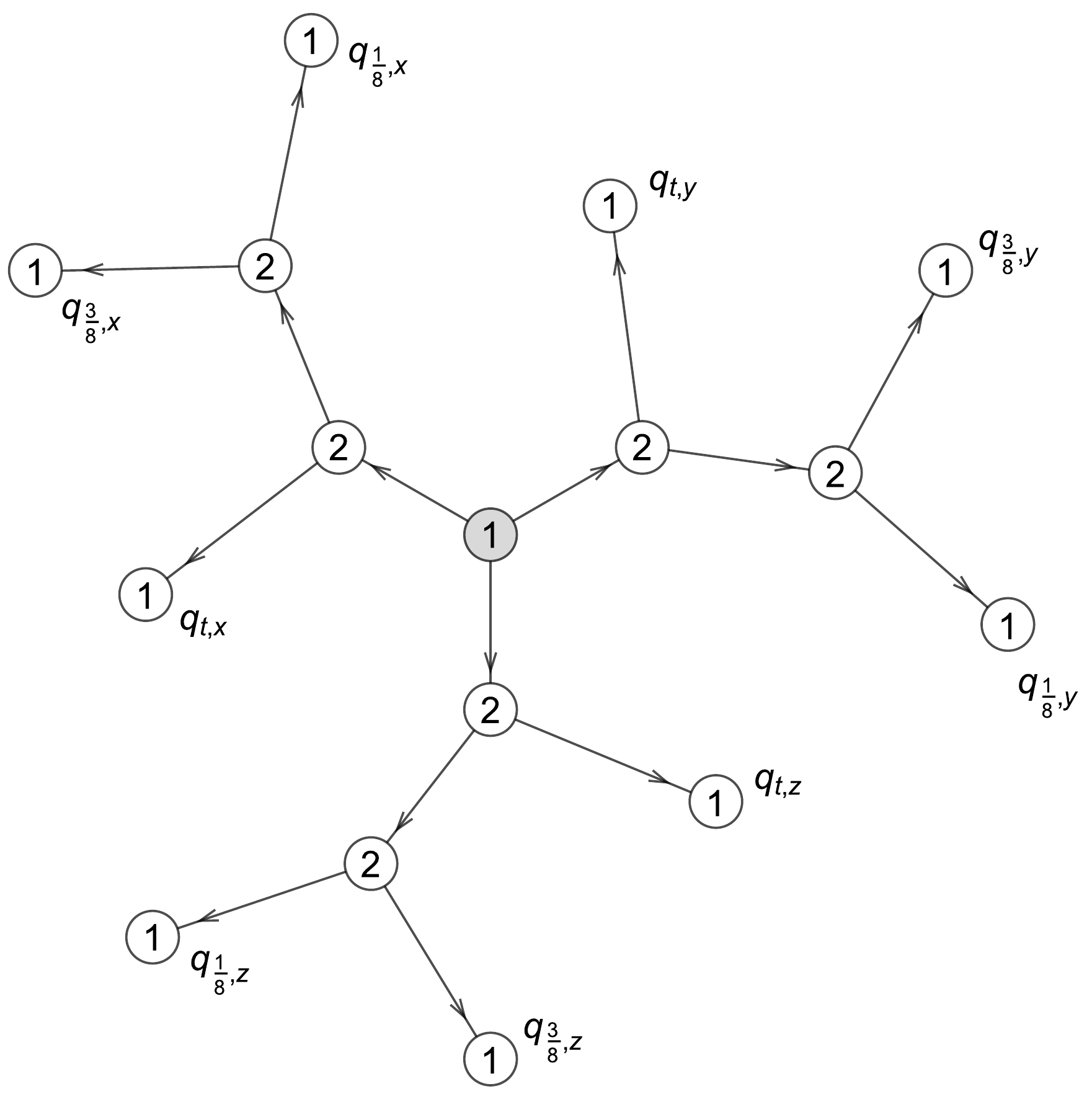}
\par\end{centering}

% Ideal principality graph used to find generators for exact synthesis algorithm. Tree corresponds to quaternion gate set specification for Clifford+$R_z(\pi/8)$ gate set
\rule[0.5ex]{1\columnwidth}{1pt}

\end{figure}

The next step is to compute $\CG^{\star}_{\mo,S}$ using the algorithm from \cite{KY1}. We find that quaternion algebra $Q$ has trivial two sided ideal class group and two different conjugacy classes of maximal orders of $Q$. The set $\CG^{\star}_{\mo,S} $ is equal to $\mathrm{gen}_S\at{\mo} \cup \mathrm{gen}_u\at{\mo}$. The set $\mathrm{gen}_u\at{\mo}$ consists of three generators of the finite group of units of maximal order $\mo$ modulo units of $\z$. As we have two conjugacy classes of maximal orders, we need to build an ideal principality graph (which is a tree in this case, see Fig.~\ref{fig:tree-pi16}) to find the set $\mathrm{gen}_S\at{\mo}$. The result of our computation is the following:
\begin{itemize}
\item $\mathrm{gen}_S\at{\mo}=\set{q_{t,\alpha},q_{1/8,\alpha},q_{3/8,\alpha} : \alpha \in \set{x,y,z}}$
\item $\mathrm{gen}_u\at{\mo}=\set{q_{c,\alpha} :  \alpha \in \set{x,y,z} }$
\end{itemize}

Our computation reproduces the result from \cite{SGKM} showing that all matrices over the ring $\z\of{\zeta_{16},1/2}$ can be exactly represented using gate set $\CG$. Because we have two conjugacy classes of maximal orders, the situation with the cost of generators becomes more interesting. For quaternions $q_{t,\alpha}$ we have $\nrd\at{q_{t,\alpha}} = \p^2$ and their cost vector is $(2)$. For other elements of $\mathrm{gen}_S\at{\mo}$ we have $\nrd\at{q_{1/8,\alpha}} = \p^3$ and $\nrd\at{q_{3/8,\alpha}} = \p^3$ and their cost vector is $(3)$. In the case when $S$ contains only one prime ideal, the cost of each generator from $\mathrm{gen}_u\at{\mo}$ is precisely equal to the distance from the root to corresponding node. Above cost values also reproduce results in \cite{SGKM}. Note that while approximating we only have control over the overall value of the cost vector. If we requested cost $L$ then the result can have any number $L_t$ of the $T$ gates and any number $L_{1/8,3/8}$ of $R_z(\pi/8)$ and $R_z(3\pi/8)$ rotations as soon as $L = 2 L_t + 3 L_{1/8,3/8}$. As usual, the cost of Clifford gates is assumed to be zero.

Another interesting aspect of this example is that generalized Lipschitz order is not contained in maximal order $M$ above. Orders in a totally definite quaternion algebra can be given a structure of the lattice using bilinear form $\tr_{F/\q}\at{q_1 q^{\ast}_2 }$. We find that the index of sub-lattice $L \cap M$ in $L$ is two. This means that half of the points from $L$ belongs $L \cap M$. In our approximation algorithm we test if the result is in $M$ in the end. If this is not the case, we try again. Our experiments show that we get resulting quaternion in $L \cap M$ in half of the experiments. Tables~\ref{tab:clifford-pi16-1}, \ref{tab:clifford-pi16-2} show the results of running our circuit synthesis algorithm for Clifford+$R_z\at{\pi/8}$ gate set.

\begin{table}

\caption{\label{tab:clifford-pi16-1}  Results of running the algorithm for Clifford+$R_z\at{\pi/8}$. Part 1. Approximation of rotation $R_z\at{0.1}$ with precision $\ve$ and cost vector $\at{L_1}$. All columns except $N_{tr,\min}$ and $N_{tr,\max}$ are averages over 1000 runs of the algorithm; $\at{L'_1}$ is a cost vector of a found circuit; $\rho\at{U_q,R_z\at{\phi}}$ is the obtained quality of approximation; $N_{tr,\min}, N_{tr,\max}, N_{tr,avg}$ are minimum, maximum and average of the number of the main loop iterations in the procedure APPROXIMATE over all samples; $t_{exact}$ and $t_{approx}$ is time in seconds spent on approximation and exact synthesis stages of the algorithm. }

\begin{centering}
\includegraphics{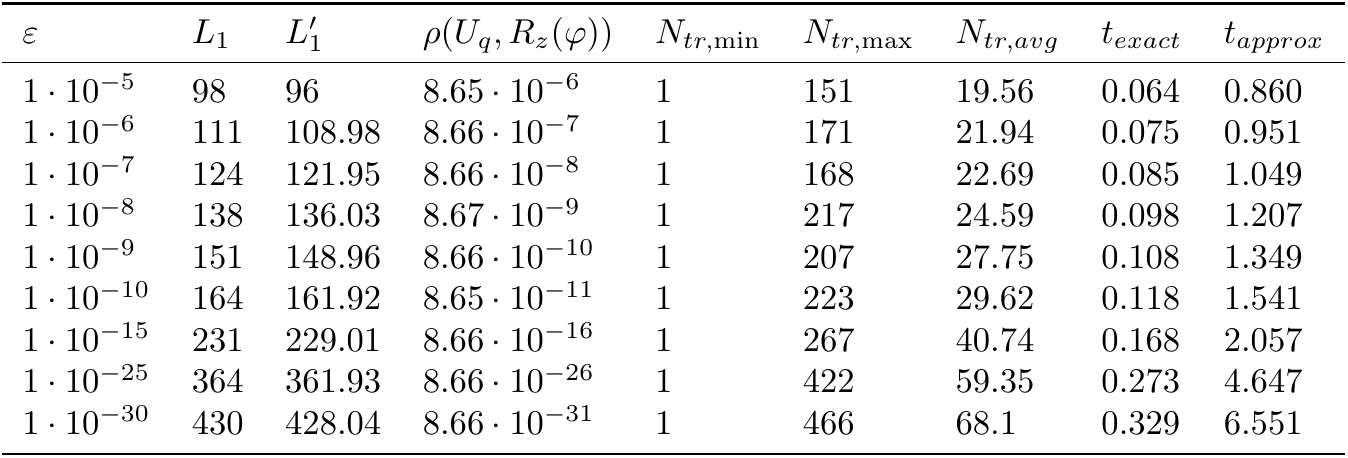}
\end{centering}

\end{table}

\begin{table}

\caption{ \label{tab:clifford-pi16-2}  Results of running the algorithm for Clifford+$R_z\at{\pi/8}$. Part 2. Approximation of rotation $R_z\at{0.1}$ with precision $\ve$ and cost vector $\at{L_1}$.  All columns except $N_{L \cap M}$, $L'_t$, $L'_{1/8,3/8}$  are averages over 1000 runs of the algorithm; $N_{L \cap M}$ is the number of outputs of the procedure APPROXIMATE that are in the maximal order $\mo$; $\at{L'_1}$ is a cost vector of a found circuit; $\rho\at{U_q,R_z\at{\phi}}$ is the obtained quality of approximation; $L'_t$ is the number of $T$ gates in the resulting circuit (averaged over outputs of the procedure APPROXIMATE that are in the maximal order $\mo$ ); $L'_{1/8,3/8}$ is the number of $R_z\at{\pi/8}$ and $R_z\at{3\pi/8}$ gates in the resulting circuit~(averaged in the same way~as~$L'_t$).  }
\begin{centering}
\includegraphics{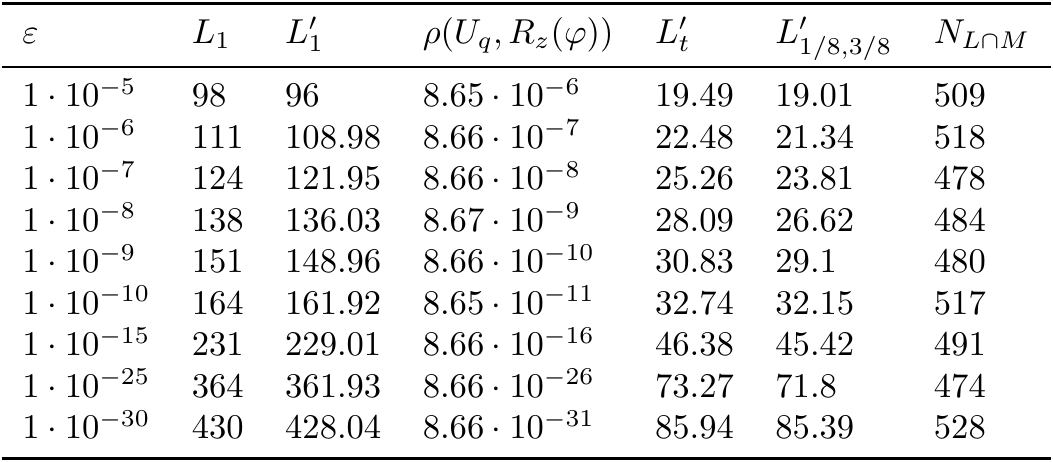}
\end{centering}

\end{table}

\subsection{Clifford+\texorpdfstring{$T$+$V$}{T+V}}
For Clifford$+T+V$ gate set we can choose set $G$ to be:
\[
\CG=\set{R_\alpha\at{\pi/4},R_\alpha\at{\pm2\atan\at{2}},R_\alpha\at{\pi/2} : \alpha = x,y,z }
\]
It is not difficult to check that $R_\alpha\at{\pm2\atan\at{2}}$ correspond to $6$ $V$ gates.

 The gate set specification is similar to Clifford+$T$ case except for the set $S$.
\begin{itemize}
\item $F=\q\at{\zeta_8+\zeta^{-1}_8}$ where $\zeta_8 = e^{2i\pi/8}$, let also $\theta$ be a primitive element of $F$ (in other words every element of $F$ can be represented as $a_0+a_1\theta$ where $a_0, a_1$ are rational numbers),
\item embedding $\sigma : F \rightarrow \r $ is defined as $\sigma\at{\theta}=\sqrt{2}$,
\item $a=-1$ and $b=-1$,
\item maximal order $\mo$ of quaternion algebra $Q=\at{\frac{-1,-1}{F}}$ is
\[
\z_F + \frac{\at{\i+1}\theta}{2}\z_F + \frac{\at{\j+1}\theta}{2}\z_F + \frac{1+\i+\j+\k}{2}\z_F,
\]
where $\z_F = \z\of{\sqrt{2}}$ is a ring of integers of $F$,
\item $S=\set{\p_1,\p_2}$ where $\p_1 = \at{2 - \theta}\z_F$ and $\p_2 = 5\z_F$.
\end{itemize}

Using notation $q_z = \i, q_y = \j, q_x = \k $ we obtain set $\CG_Q$ based on the following correspondence
\[
\begin{array}{lccclcc}
 q_{t,\alpha} & = & 1+\theta(1-q_\alpha)/2 & & U_q\at{q_{t,\alpha}} & = & R_\alpha\at{\pi/4} \\
 q_{v,\pm \alpha} & = & 1\mp 2 q_\alpha & & U_q\at{q_{v,\pm\alpha}} & = & R_\alpha\at{\pm 2\atan\at{2}} \\
 q_{c,\alpha} & = & \theta(1-q_\alpha)/2 & & U_q\at{q_{c,\alpha}} & = & R_\alpha\at{\pi/2}
\end{array}
\]
where $\alpha=x,y,z$.

The next step is to compute $\CG^{\star}_{\mo,S}$ using the algorithm from \cite{KY1}. We find that quaternion algebra $Q$ has trivial two sided ideal class group and that the number of conjugacy classes of maximal orders of $Q$ is one. The set $\CG^{\star}_{\mo,S} $ is equal to $\mathrm{gen}_S\at{\mo} \cup \mathrm{gen}_u\at{\mo}$. The set $\mathrm{gen}_S\at{\mo}$ consists of $N\at{\p_1}+1 = 3$ elements with reduced norm $2-\theta$ and $N\at{\p_2}+1 = 26$ elements with reduced norm $5$. The set $\mathrm{gen}_u\at{\mo}$ consists of three generators of the finite group of units of maximal order $\mo$ modulo units of $\z_F$ and is the same as in Clifford+$T$ case because maximal order $\mo$ is the same.

Let us now discuss the set $\mathrm{gen}_S\at{\mo}$ in more details. We first find that it contains quaternions corresponding to $R_\alpha\at{\pi/4}$ gates and all 6 $V$ gates. We are left with 20 quaternions with reduced norm 5 that we didn't have in the set $\CG_Q$. We express them in terms of elements of $\CG_Q$. Let us introduce the following equivalence relation on quaternions:
\[
 q_1 \sim q_2  \text{ if and only if } q_1 = u_1 q_2 u_2  \text{ for } u_1,u_2 \text{ -- units of M}
\]
In our case it means that corresponding unitaries are equivalent up to a Clifford and therefore will have the same cost of implementation. There are four equivalence classes in $\mathrm{gen}_S\at{\mo}$ corresponding to the relation $\sim$. Two of them are $\set{ q_{t,\alpha} : \alpha \in \set{x,y,z}}$ and $\set{ q_{v,\pm\alpha} : \alpha \in \set{x,y,z}}$. Remaining twenty quaternions with reduced norm 5 split into two classes $c_1$ and $c_2$ of size $8$ and $12$. Next we find that all quaternions from $c_2$ are equal to
\[
	u_1 q_{t,\alpha(1)} q_{v,\pm\alpha(2)} u_2 q^{-1}_{t,\alpha(3)} \text{ where } u_1,u_2 \text{ are units of }\mo, \alpha(k) \in \set{x,y,z}, k=1,2,3.
\]
The quaternions from the set $c_1$ can be expressed as
\[
	u_1 q_{t,\alpha(1)}q_{t,\alpha(2)} q_{v,\pm\alpha(3)} u_2 q^{-1}_{t,\alpha(4)}q^{-1}_{t,\alpha(5)} \text{ where } u_1,u_2 \text{ are units of }\mo, \alpha(k) \in \set{x,y,z} k=\rg{1}{5}.
\]
In practice it can be more beneficial to design circuits for all $26$ gates corresponding to quaternions with norm $5$ directly, because $T$ gates are usually expensive to implement. Table~\ref{tab:clifford-tv} shows the results of running our circuit synthesis algorithm for Clifford+$T$+$V$ gate set.

\begin{table}[t!]
{
\caption{\label{tab:clifford-tv} Results of running the algorithm for Clifford+$T$+$V$. Approximation of rotation $R_z\at{0.1}$ with precision $\ve$ and cost vector $\at{L_1,L_2}$. All columns except $N_{tr,\min}$ and $N_{tr,\max}$ are averages over 1000 runs of the algorithm; $\at{L'_1,L'_2}$ is a cost vector of a found circuit; $\rho\at{U_q,R_z\at{\phi}}$ is the obtained quality of approximation; $N_{tr,\min}, N_{tr,\max}, N_{tr,avg}$ are minimum, maximum and average of the number of the main loop iterations in the procedure APPROXIMATE over all samples; $t_{exact}$ and $t_{approx}$ is time in seconds spent on approximation and exact synthesis stages of the algorithm. }
\begin{centering}
\includegraphics{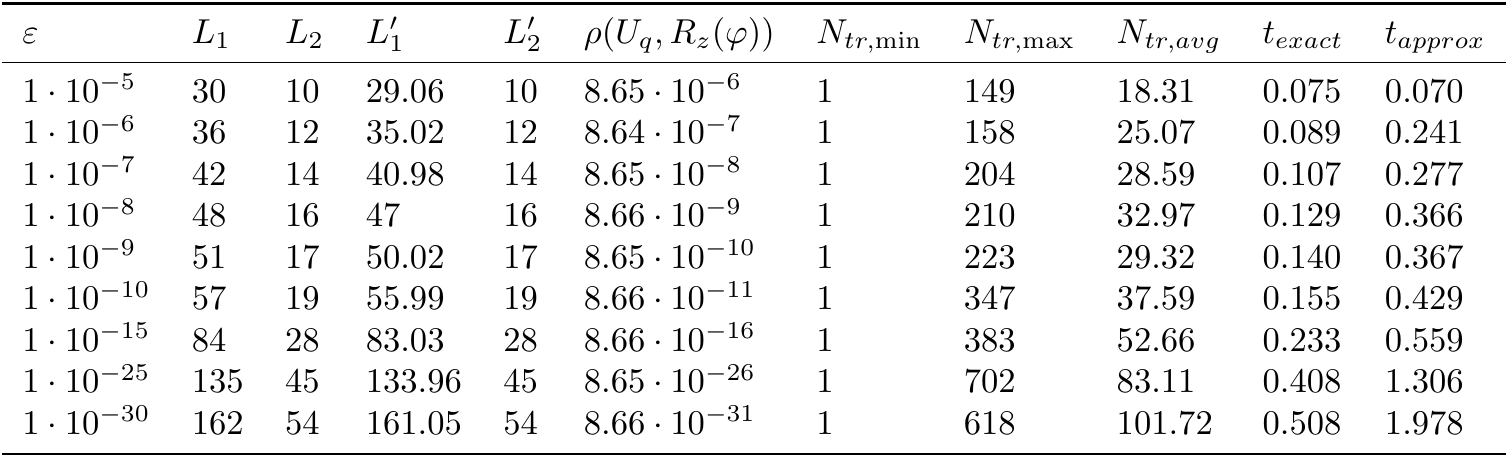}
\end{centering}
}
\end{table}

%\subsection{Other gate sets based on totally definite quaternion algebras} 

%auto-ignore
%!TEX root = quaf.tex

\section{Conclusions}
We introduced a framework for approximate synthesis of single qubit unitary transformations over a universal gate set. Our framework is applicable whenever the gate set is related to totally definite quaternion algebras. Our algorithm runs in time that is polynomial in $\log(1/\varepsilon)$, where $\varepsilon$ is the approximation parameter and the output factorizations produced have length $O(\log(1/\varepsilon))$. We implemented the algorithm in the computer algebra system Magma and demonstrate it by applying it to a wide range of gate sets.

The proof that our algorithm terminates and runs on average in polynomial time relates on the number theoretic conjecture, similarly to~\cite{S,BGS,BRS:2015b}. The question related to this conjectures were recently studied in~\cite{Sarnak:2015} for Clifford+T, V-basis and some other gate sets. Results of our numerical experiments provide indirect evidence that some of these results can be true for a wider range of ``Golden gates'' in the sense of Sarnak \cite{Sarnak:2015}. 

\section*{Acknowledgment} The authors would like to thank the QuArC team for many discussions on previous versions of this work. MR would like to thank Sean Hallgren for discussions on computational number theory during the Dagstuhl seminar on {\em Quantum Cryptanalysis} in September 2015, in particular for pointing out reference \cite{GGH:2013}. 

% \nocite{*}
% \bibliography{quaf}
% \bibliographystyle{plain}

\clearpage

%auto-ignore
%!TEX root = quaf.tex

\section*{Appendix A. Other examples of using approximation algorithm } \label{sec:appendix}

In this appendix we show results of running our approximation algorithm for a series of quaternion gate sets specification. We recall that the quaternion gate set specification is : 

\qgc* 

Our family of examples if parametrized by $n$. Number field $F$ corresponds to the real subfield of cyclotomic field $\q\at{\zeta_n+\zeta^{-1}_n}$ with primitive element $\theta$. Element $a$ of $F$ is chosen such that relative extension $F\at{a}$ is a cyclotomic field, $b=-1$. Approximation part of our algorithm is independent on maximal order $\mo$, so we don't restrict ourselves to any specific choice of $\mo$. Set $S$ contains one prime ideal above $2$. If there is more than one such ideal, we choose it at random. The table below summarizes the series of examples we tried using our algorithm. 

\begin{table}[!ht]
\begin{centering}
\includegraphics{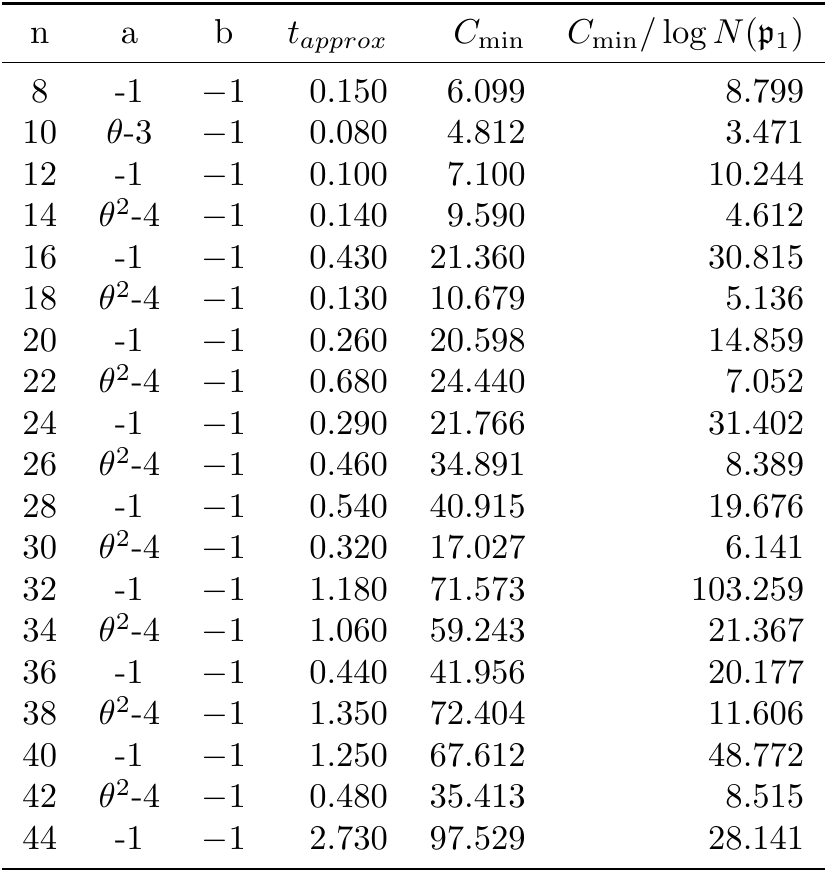}
\end{centering}
\end{table}

\label{page:table-desc}
The table summarizes the results of running offline part of our algorithm : $n$ is the number of example in the family described above; $t_{approx}$ is the time in seconds spent on the offline stage required for the approximation part; $C_{\min}$ is the additive constant appearing in Theorem~\ref{thm:main}; $C / \log N\at{\p_1}$ is the ratio between $C_{\min}$ and the $\log$ of the norm of the ideal in~$S$. Next we show tables with the averages over $100$ runs of our algorithm with different target precisions $\ve$ and target cost vector $\at{L_1}$ and target angle $\vp = 0.1$ for each example for $n=8,10,\ldots,44$.  All columns of the tables except $N_{tr,\min}$ and $N_{tr,\max}$ are averages over 100 runs of the algorithm; $\rho\at{U_q,R_z\at{\phi}}$ is the obtained quality of approximation; $N_{tr,\min}, N_{tr,\max}, N_{tr,avg}$ are minimum, maximum and average of the number of the main loop iterations in the procedure APPROXIMATE over all samples;   $t_{approx}$ is time in seconds spent on online part of the approximation stage of the algorithm.

\begin{table}[!ht]

\begin{centering}
\includegraphics{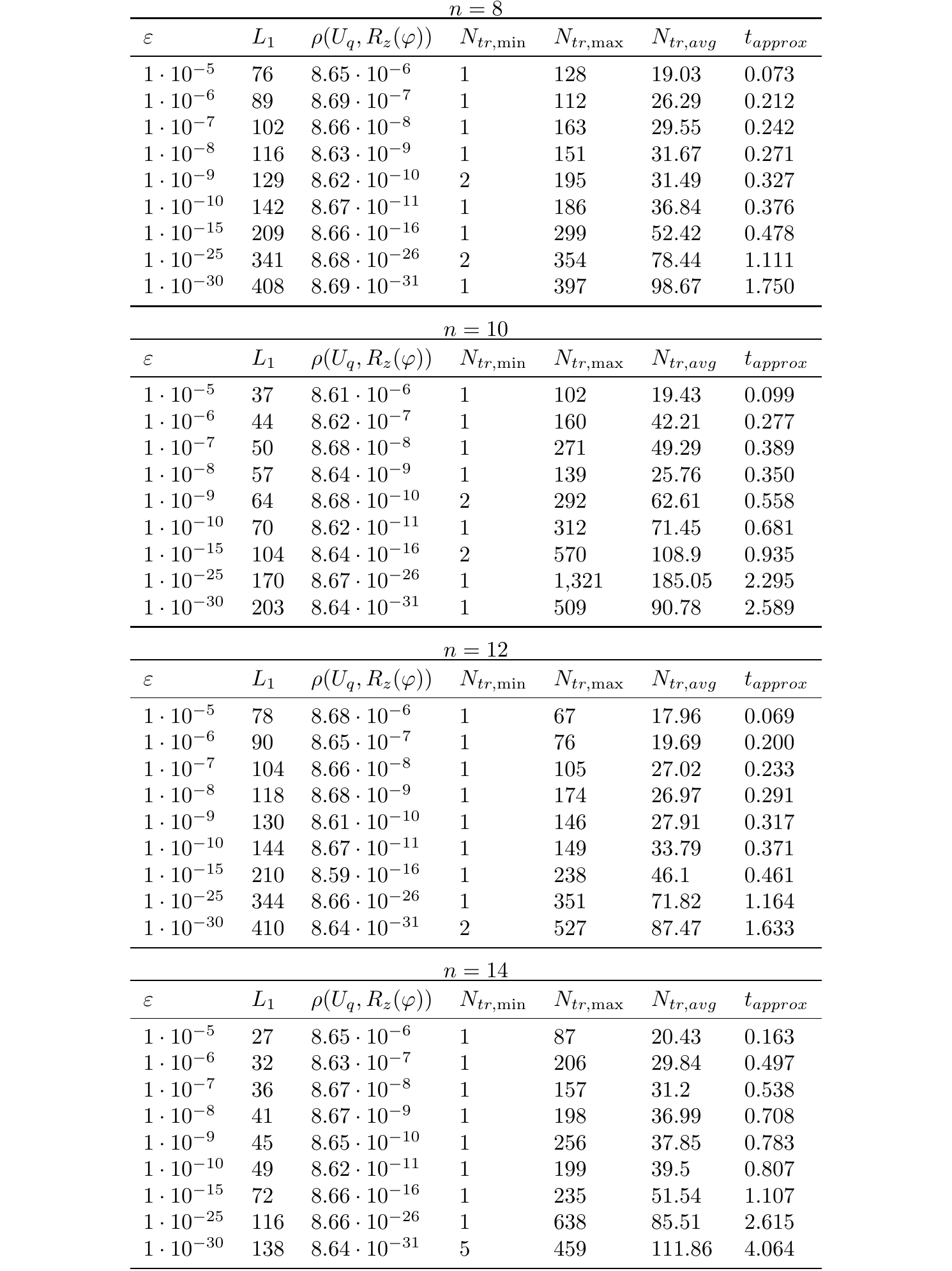}
\end{centering}
\caption{See Page \pageref{page:table-desc} for the description of the columns of the tables.}
\end{table}

\begin{table}[!ht]
\begin{centering}
\includegraphics{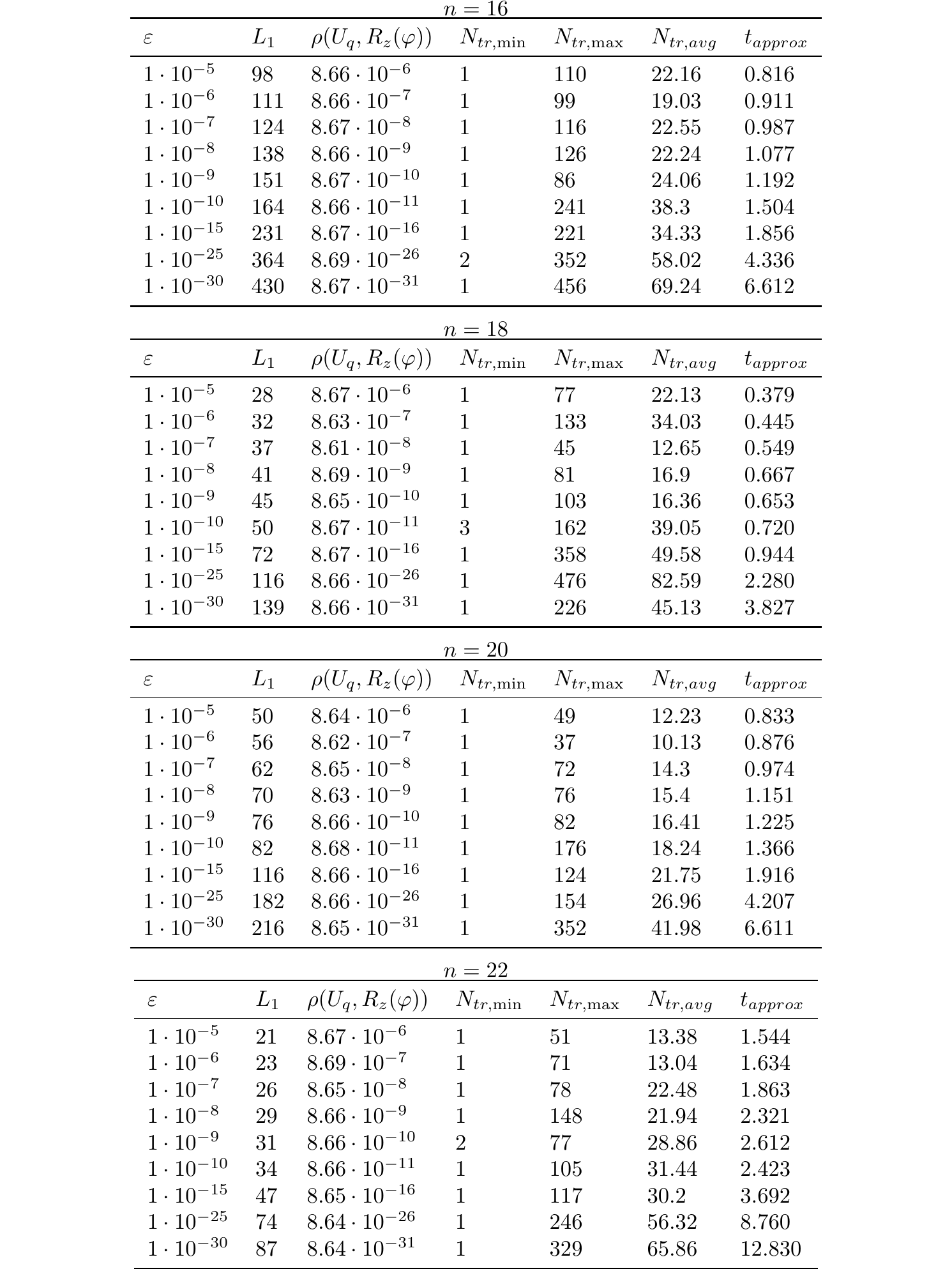}
\end{centering}
\caption{See Page \pageref{page:table-desc} for the description of the columns of the tables.}
\end{table}

\begin{table}[!ht]
\begin{centering}
\includegraphics{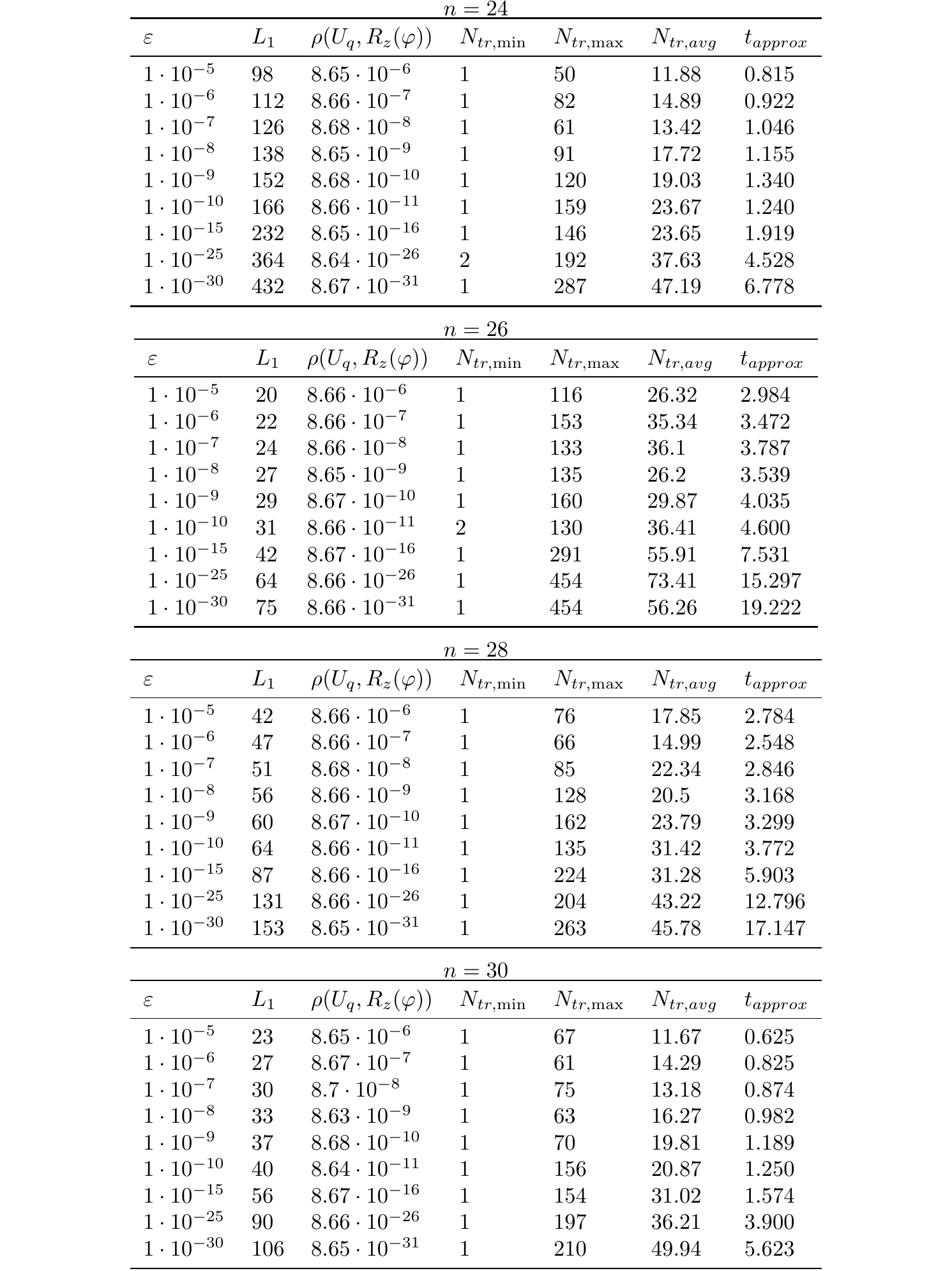}
\end{centering}
\caption{See Page \pageref{page:table-desc} for the description of the columns of the tables.}
\end{table}

\begin{table}[!ht]
\begin{centering}
\includegraphics{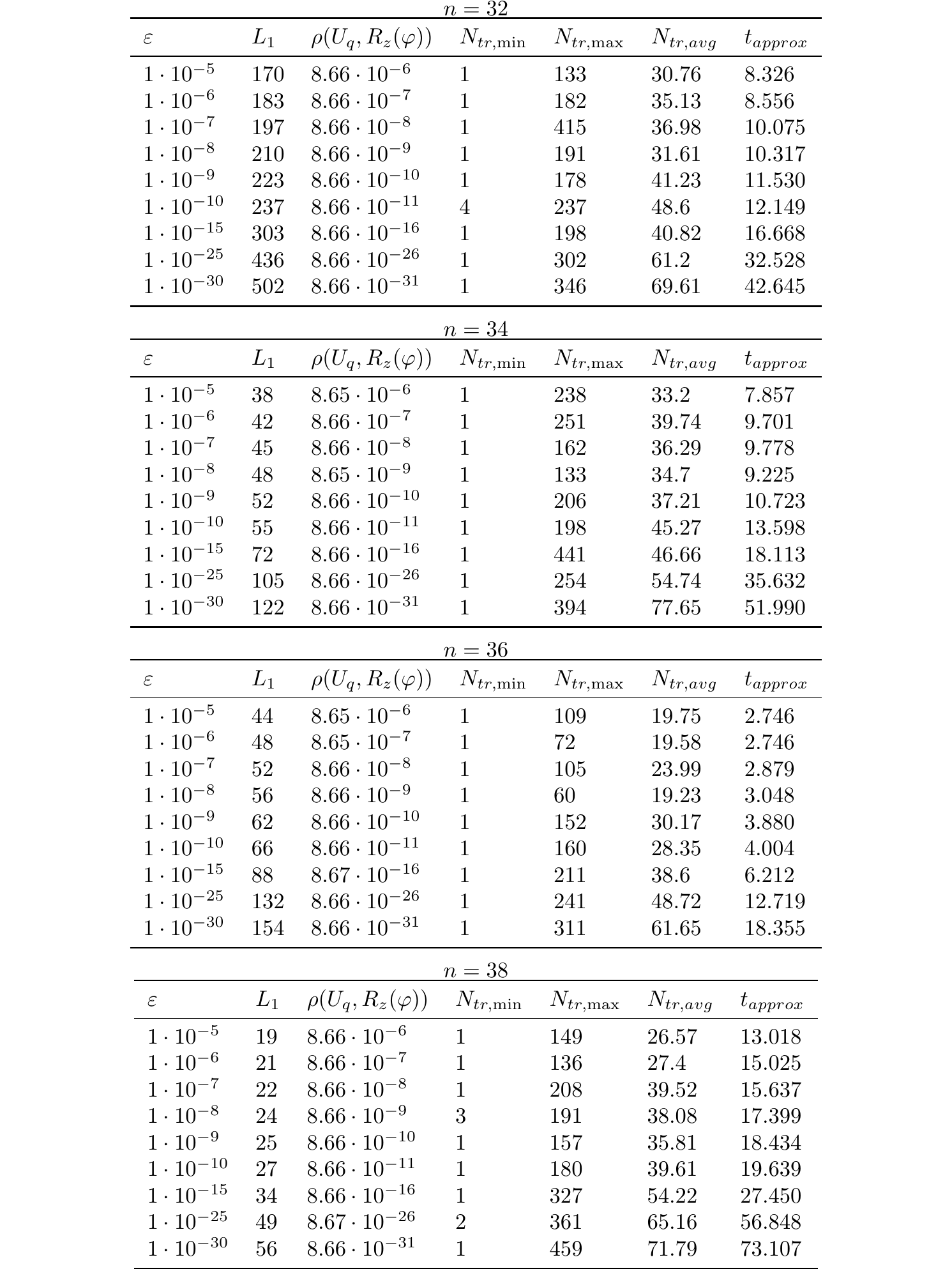}
\end{centering}
\caption{See Page \pageref{page:table-desc} for the description of the columns of the tables.}
\end{table}

\begin{table}[
!ht]
\begin{centering}
\includegraphics{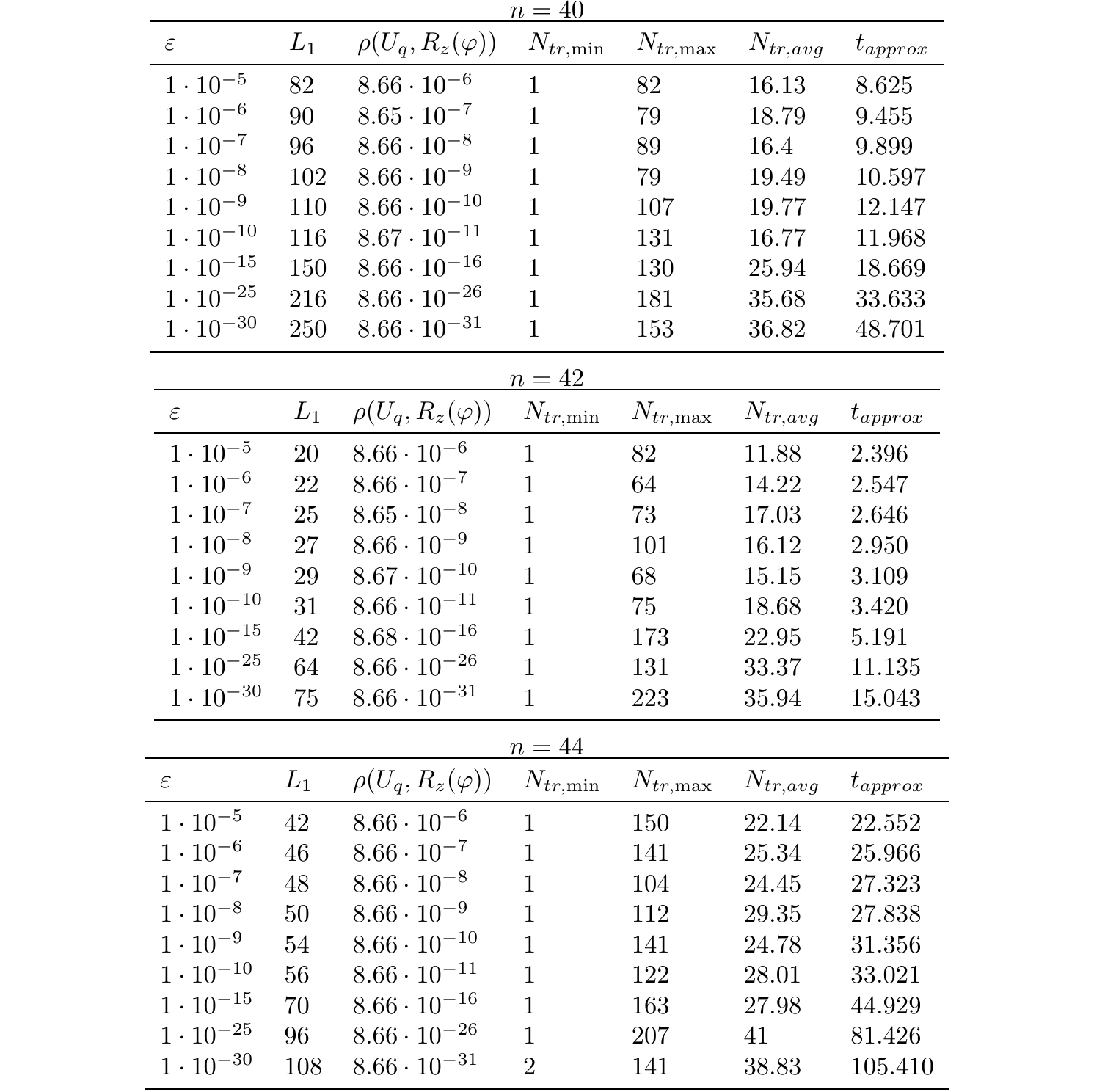}
\end{centering}
\caption{See Page \pageref{page:table-desc} for the description of the columns of the tables.}
\end{table}

\clearpage
%auto-ignore
%!TEX root = quaf.tex

\section*{Appendix B. Implementation details } \label{sec:appendix-b}

In this appendix we provide a pseudo-code for the versions of procedures UNIT-ANDJUST~(Fig.~\ref{fig:proc-unit-adjust-impl}) and RANDOM-INTEGER-POINT~(Fig.~\ref{fig:proc-find-point-impl}) that we use in the implementation of our algorithm. They are necessary to reproduce values of additive constant listed in Section~\ref{sec:examples} and Appendix A on Page~\pageref{sec:appendix}. They both rely on the Nearest Plane Algorithm~\cite{BABAI:1986} shown in Fig.~\ref{fig:proc-nearest-plane-impl}). For given target vector $t$ the Nearest Plane algorithm finds unique lattice vector $Bv$ inside $t + \CC\at{B^{*}}$. See Section~\ref{sec:lattice} for the definition of $\CC\at{B^{*}}$ and other related definitions. 
%%%%%%%%%%%%%%%%%%%%%%%%%%%%%%%%%%%%%%%%%% 

\begin{figure}[ht!]
\protect\caption{\label{fig:proc-nearest-plane-impl} NEAREST-PLANE\cite{BABAI:1986} procedure used in our implementation }
\rule[0.5ex]{1\columnwidth}{1pt}

\begin{centering}

%auto-ignore
%!TEX root = quaf.tex 

\begin{algorithmic}[1]
\fxinput $B = b_1, \ldots, b_n $ -- basis of a lattice of rank $n \le d$ in $\r^d$
\inp $t$ -- vector in $\r^d$
\Procedure{NEAREST-PLANE}{}
\offline
\State $b_1^{*}, \ldots,b_{n}^{*} \gets $\Call{GRAM-SHCMIDT}{$B$} \Comment Computes GSO of B, see Section~\ref{sec:lattice} 
\spec $b_1^{*}, \ldots,b_{n}^{*}$ -- GSO of B
\online 
\State $v = \at{v_1,\ldots,v_n}$ is a vector of integers of length $n$
\State $r \gets t$
\For{ $k \in 1,\ldots,n$ }
\State $v_{n+1-k} \gets \rnd{\ip{r,b^{*}_{n+1-k}}/\nrm{b^{*}_{n+1-k}}^2}$ \Comment ensures $\ip{r,b^{*}_{n+1-k}}/ \nrm{b^{*}_{n+1-k}}^2 \le 1/2 $
\State $r \gets r - v_{n+1-k} b_{n+1-k} $
\EndFor
\State \Return $v$
\EndProcedure
\out $v$ such that $Bv \in t + \CC\at{B^{*}}$
\end{algorithmic}

\par\end{centering}

\rule[0.5ex]{1\columnwidth}{1pt}

\end{figure}

%%%%%%%%%%%%%%%%%%%%%%%%%%%%%%%%%%%%%%%%%%

\begin{figure}[ht!]

\protect\caption{\label{fig:proc-unit-adjust-impl} UNIT-ADJUST procedure used in our implementation. We compute with $\log \delta_k$. In the pseudo-code $\log \delta_k$ should be perceived as a variable name. We also try to minimize $\sum \log \delta_k$ by applying different basis reduction techniques to $B$ and picking the reduced basis that gives the smallest $\sum \log \delta_k$. Variable $\log \delta_0$ is chosen based precision used for arithmetic operations in the algorithm. }
\rule[0.5ex]{1\columnwidth}{1pt}

\begin{centering}

%auto-ignore
%!TEX root = quaf.tex

\begin{algorithmic}[1]
\fxinput $F$, $u_1,\dotsc,u_{d-1}$
\Statex[1] $F$ is a totally real number field of degree $d$
\Statex[1] $u_1,\dotsc,u_{d-1}$ form a system of fundamental units of $F$
\inp $\rg{t_1}{t_d}$
\Statex[1] $t_1,\ldots,t_d$ -- real numbers of the same precision $n$
\Procedure{UNIT-ADJUST}{}
\offline
\State $u_1,\dotsc,u_{d-1} \leftarrow$ \Call{INDEPENDENT-UNITS}{$\z_F$} 
\State \Comment $\langle u_1,\ldots,u_{d-1} \rangle $ is a finite index subgroup of  $\z_F^{\times}$

\State $B \leftarrow 
\pmat{\log|\sig_1(u_1)| & \cdots & \log |\sig_1(u_{d-1})| \\ 
\vdots  & & \vdots \\ 
\log|\sig_d(u_1)| & \cdots & \log |\sig_d(u_{d-1})| }$

\State $b_1^{*}, \ldots,b_{d-1}^{*} \gets $\Call{NEAREST-PLANE}{$B$} \Comment Computes GSO of B, see Section~\ref{sec:lattice} 

\State $\log \delta_k \leftarrow \frac{1}{2} \sum_{j=1}^{d-1} \abs{\ip{b_j^{*},e_k}}$ \Comment $e_1,\ldots,e_d$ is the standard basis of $\r^d$

\spec $\rg{ \log \delta_1}{ \log \delta_d}$ -- real numbers such that $\log \delta_k > 0$ and  

\Statex[4] $\CC(B^{*})\subset [-\log \delta_1,\log\delta_1] \times \cdots \times [-\log \delta_d,\log\delta_d]$ \Comment see Section~\ref{sec:lattice}

\online
\State \assert $\abs{\prg{t_1}{t_d}} < \log \delta_0$ \Comment Make sure the point is in the span of the lattice 
\State $m \leftarrow $\Call{NEAREST-PLANE}{t} \Comment Find $m$ such that $Bm \in t + \CC(B^{*}) $
\State  $u\leftarrow \prod_{i=1}^{d-1}u_i^{m_i}$
\EndProcedure
\out unit $u \in \z_F^\times$ such that
\Statex for all $k=\rg{1}{d}$ : $\abs{\log\abs{\s_k\at{u}} - t_k } \le \log \delta_k $

\end{algorithmic}

\par\end{centering}

\rule[0.5ex]{1\columnwidth}{1pt}

\end{figure}

%%%%%%%%%%%%%%%%%%%%%%%%%%%%%%%%%%%%%%%%%%

\begin{figure}[ht!]

\protect\caption{\label{fig:proc-find-point-impl} RANDOM-INTEGER-POINT procedure used in our implementation. In our implementation, we try to minimize $\sum \log R_k$ by applying different basis reduction techniques to $B$ and picking the reduced basis that gives the smallest $\sum \log R_k$. }
\rule[0.5ex]{1\columnwidth}{1pt}

\begin{centering}

%auto-ignore
%!TEX root = quaf.tex 

\begin{algorithmic}[1]
\fxinput $\z_K$ -- ring of integers of a CM field of degree $2d$
\inp $\vp,\ve$ -- real numbers, $r$ -- totally positive element of $\z_F$
\Procedure{RANDOM-INTEGER-POINT}{}
\offline
\State $z_1,\ldots,z_{2d}$ is a fixed integral basis of $\z_F$
\State $B = [\boldsymbol\sigma\at{z_1},\ldots, \boldsymbol\sigma\at{z_{2d}}]$ is a basis of the lattice associated to $\z_K$,
\State $b_1^{*}, \ldots,b_{2d}^{*} \gets $\Call{NEAREST-PLANE}{$B$} \Comment Computes GSO of B, see Section~\ref{sec:lattice} 
\State $R^{\min}_k = \frac{1}{2}\sqrt{ \max_j \abss{\ip{b^{*}_j,e_{2k-1}}} + \max_j \abss{\ip{b^{*}_j,e_{2k}}} }$ for $k=1,\ldots,d$
\Statex \Comment $e_1,\ldots,e_{2d}$ is the standard basis of $\r^{2d}$
\State $C_{\min},C_{\max}, R^{\max} \gets $ \Call{SUITABLE-Q-NORM}{$\rg{\p_1}{\p_M},R^{\min}$} \Comment See Fig.~\ref{fig:proc-suitable-q-norm}
\State \Return $C_{\min},C_{\max}$
\spec $C_{\min},C_{\max}$ -- real numbers 

\online
\State $R \gets \sqrt{\s_1\at{r}}, H' \gets R\ve\sqrt{4-\ve^2},\,N_{\max}=\floor{H'/\ve^2 R} $ \Comment See Fig.\ref{fig:sampling}
\State $z_c \gets R(1-3\ve^2/4)e^{-t\vp/2},\,\Delta z \gets i e^{-i\vp/2} \ve^2 R /2  $ \Comment See Fig.\ref{fig:sampling}
\State $ \Delta Z \gets \at{\re \Delta z, \im \Delta z, 0, \ldots, 0}  \in \r^{2d}, Z_c \gets \at{\re z_c, \im  z_c, 0, \ldots, 0}  \in \r^{2d} $

\State $N \gets$ random integer from the interval $[ -N_{\max}+1 , N_{\max}-1 ]$ 
\State $ t \gets Z_c + N \Delta Z $ 
\State $ m \gets$\Call{NEAREST-PLANE}{$t$}  \Comment Find $m$ such that $Bm \in t + \CC(B^{*}) $
\State $z \gets   m_1 z_1 + \ldots + m_{2d} z_{2d}$  
\State \Return z  \Comment $\vecs\at{z} \in S_{r,\vp,\ve}$
\EndProcedure
\out $z$, the element of $\z_K$ such that 
\Statex[2] $\abs{\s_{k,+}\at{z}} \le R_k$ for  $k=2,\ldots,d$  and  $\re \at{\at{\s_{1,+}\at{z}-z_0}e^{-i\vp/2} } \ge 0, \abs{ \s_{1,+}\at{z} } \le R$  
\Statex[2] where $z_0 = R\at{1-\ve^2}e^{-i\vp/2}$ (see Fig.~\ref{fig:inequality} for the visualization of the condition on $\s_{1,+}\at{z}$)
\end{algorithmic}

\par\end{centering}

\rule[0.5ex]{1\columnwidth}{1pt}

\end{figure}

\end{document}